\documentclass[12pt]{article}

\usepackage[margin=1in]{geometry}
\usepackage{setspace}

\usepackage{amsmath,amssymb,amsfonts,amsthm,mathtools}
\usepackage{bm}
\usepackage{bbm}
\usepackage{tikz}
\usetikzlibrary{arrows.meta, positioning}
\usepackage{comment}
\usepackage{algorithm}
\usepackage{algpseudocode}
\usepackage{blindtext}
\usepackage{parskip}
\usepackage{graphicx}
\usepackage{booktabs}
\usepackage{multirow}
\usepackage{subcaption}
\usepackage{natbib}
\usepackage{authblk}
\allowdisplaybreaks
\usepackage[colorlinks=true,
            linkcolor=blue,
            citecolor=blue,
            urlcolor=blue]{hyperref}

\usepackage{cleveref}

\newcommand{\halpha}{\hat{\boldsymbol{\alpha}}}
\newcommand{\hbeta}{\hat{\boldsymbol{\beta}}}

\newcommand{\hPhi}{\hat{\Phi}}
\newcommand{\hPsi}{\hat{\Psi}}

\newcommand{\talpha}{\hat{\boldsymbol{\alpha}}^*}
\newcommand{\hG}{\hat{\mathbf{G}}}
\newcommand{\G}{\mathbf{G}}
\newcommand{\PhiStar}{\Phi}
\newcommand{\ew}{\epsilon_{n,\delta,w}}
\newcommand{\ewtilde}{\epsilon_{n,\delta,\tilde w}}
\newcommand{\es}{\epsilon_{n,\delta,\sigma}}

\newcommand{\lam}{\lambda_0}

\newcommand{\hsigma}{\hat{\boldsymbol{\sigma}}^2}
\newcommand{\tsigma}{\hat{\boldsymbol{\sigma}}^{*2}}
\newcommand{\hD}{\hat{\mathbf{D}}}
\newcommand{\D}{\mathbf{D}}
\newcommand{\hW}{\hat{W}}
\newcommand{\W}{W}
\newcommand{\hC}{\hat{C}}
\newcommand{\C}{C}
\newcommand{\ea}{\epsilon_{n,\delta,\alpha}}

\newcommand{\tbeta}{\hat{\boldsymbol{\beta}}^*}
\newcommand{\hH}{\hat{\mathbf{H}}}
\newcommand{\Hmat}{\mathbf{H}} 
\newcommand{\PsiStar}{\Psi}

\theoremstyle{plain}
\theoremstyle{plain}
\newtheorem{theorem}{Theorem}[section]
\newtheorem{proposition}[theorem]{Proposition}
\newtheorem{lemma}[theorem]{Lemma}

\theoremstyle{definition}
\newtheorem{definition}[theorem]{Definition}
\newtheorem{assumption}{Assumption}[section]
\theoremstyle{remark}
\newtheorem{remark}[theorem]{Remark}
\definecolor{codegreen}{rgb}{0,0.6,0}
\definecolor{codegray}{rgb}{0.5,0.5,0.5}
\definecolor{codepurple}{rgb}{0.58,0,0.82}
\definecolor{backcolour}{rgb}{0.95,0.95,0.92}

\crefname{assumption}{assumption}{assumptions}
\Crefname{assumption}{Assumption}{Assumptions}

\crefname{theorem}{theorem}{theorems}
\Crefname{theorem}{Theorem}{Theorems}

\crefname{lemma}{lemma}{lemmas}
\Crefname{lemma}{Lemma}{Lemmas}

\crefname{figure}{fig.}{}
\Crefname{figure}{Fig.}{}

\DeclarePairedDelimiter\floor{\lfloor}{\rfloor}

\newcommand{\E}{\mathbb{E}}

\title{\vspace{-1.5cm}
\bfseries Causal Inference with Categorical Unobserved Confounder via Mixture Learning
}

\author[1,2]{Aytijhya Saha}
\author[1,2]{Stephen Bates}
\author[1,2]{Devavrat Shah}
\vspace{-1cm}

\affil[1]{Department of Electrical Engineering and Computer Science, MIT}
\affil[2]{Laboratory for Information and Decision Systems, MIT}
\affil[ ]{\texttt{\{aytijhya, s\_bates, devavrat\}@mit.edu}}
\date{}

\begin{document}

\maketitle
\vspace{-0.5cm}

\begin{abstract}
Unobserved confounding is a fundamental challenge for estimating causal effects. 
To address unobserved confounding, recent literature has turned to two different approaches --- proxy variables and the use of multiple treatments.  
The first approach, commonly referred to as proximal causal inference, requires proxies to be assigned to specific asymmetric roles: treatment-inducing proxies (negative control exposures), variables that act as common causes of the treatment and outcome, and outcome-inducing proxies (negative control outcomes). 
In practice, however, identifying variables that satisfy these asymmetric roles can be difficult depending on the application domain. The second approach, commonly referred to as the ``Deconfounder," deals with multiple conditionally independent treatments. There has been limited
progress towards developing a consistent estimation method for this setting.
As the primary contribution of this work, we establish that causal effects are identifiable in both multi-proxy and multi-treatment settings when the unobserved confounder is categorical
under suitable conditions. Our approach builds on a mixture learning perspective: we show that the underlying confounding structure can be recovered by identifying the corresponding mixture distribution. We propose an estimation procedure based on tensor decomposition, which allows consistent recovery of the latent structure and comes with non-asymptotic guarantees. Simulation studies and real data experiments demonstrate that the proposed method performs well even with limited data.

\end{abstract}

\section{Introduction}

\medskip

Estimating the causal effect of a treatment on an outcome is a fundamental challenge across disciplines, from personalized medicine and economics to digital marketing. The gold standard for causal inference is the randomized controlled trial (RCT). However, RCTs are often unethical, expensive, or infeasible, necessitating the use of observational data. A pervasive obstacle in observational studies is the presence of unobserved confounders—latent variables that influence both the treatment assignment and the outcome. Failure to account for these confounders introduces bias, rendering standard estimators like regression or propensity score matching inconsistent.
For example, in personalized medicine, a physician’s decision to treat a patient (Treatment $A$) depends on the patient's underlying health status (Confounder $U$), which also dictates the patient's recovery (Outcome $Y$). 
While some health metrics are recorded in the electronic health record (EHR), the ``true" underlying health state remains latent. 
Ignoring this latent state can lead to spurious conclusions (known as ``Simpson's paradox''), such as observing that intensive care correlates with higher mortality simply because sicker patients receive more aggressive care.
To address unobserved confounding, recent literature has turned to two different approaches --- proxy variables—observed covariates that are not confounders themselves but are noisy measurements of the latent state, and multiple treatments that help reconstruct the unknown confounder.

The proximal causal learning framework established by \cite{miao2018identifying} and \cite{tchetgen2020introduction} 
provides a powerful non-parametric identification strategy. However, it requires that an analyst can correctly classify proxies into three bucket types: 
treatment-inducing proxies ($Z_1$, negative control exposure), variables that are common causes of
the treatment ($Z_2$), outcome variables, and outcome-inducing proxies ($Z_3$, negative control outcome).
Precisely, this is depicted in \Cref{fig:dag_proxy} with $U$ being the latent confounder, $A$ representing
action or intervention, $Y$ representing outcome and $Z_1, Z_2$ and $Z_3$ being proxy variables
of specific types as described. In practice, finding variables that satisfy these specific, 
asymmetric roles---where one affects treatment but not outcome, and vice versa---can be challenging depending upon the domain.

Another thread of work is based on multiple treatments introduced by \cite{wang2019blessings}. The key insight is that when several treatments are assigned simultaneously and are conditionally independent given the latent confounder, the joint distribution of treatments may reveal the underlying latent structure. However,
the assumptions required for strict identifiability have been a subject of active discussion in the literature, with \citet{ogburn2019comment} highlighting scenarios where point identification may not hold. 
Further, the estimation method utilized often relies on approaches such as expectation maximization (EM) or variational inference. While these approaches can be effective in practice, generally they do not come with global convergence properties. Consequently, there remains a valuable opportunity to develop multi-treatment methodologies that offer identifiability guarantees coupled with globally convergent estimation algorithms.

In this work, we investigate the following two questions.

\begin{quote}
\textbf{Multi-proxy setting.} Is it possible to obtain causal identifiability guarantees in a multi-proxy setting without requiring the proxies to be assigned \emph{asymmetric} roles, together with an efficient estimation procedure that enjoys finite-sample and global convergence guarantees?
\end{quote}

\begin{quote}
\textbf{Multi-treatment setting.} Is it possible to obtain clean causal identifiability guarantees in a multi-treatment setting, along with a globally convergent estimation procedure?\end{quote}
We answer both questions in the affirmative under suitable assumptions. Both rely on the same key insight: we demonstrate that when $U$ is categorical, i.e., takes a finite number of values, and conditionally independent proxies or treatments are available (specifically, three or more), the identification problem becomes simpler --- the latent confounding structure, i.e., the underlying mixture distribution, is identifiable \citep{allman2009identifiability}. Further, using tensor decomposition \citep{allman2009identifiability, anandkumar2012method,anandkumar2014tensor,song2014nonparametric}, such mixture distributions can be estimated from 
finitely many observations in a globally consistent manner. This mixture learning approach allows us to ``unmix" the latent confounder distribution. We also note that in many real-world applications, it is not unreasonable to assume that confounders are categorical, and similar assumptions have been considered in prior work \citep{shi2020multiply,burauel2024controlling,zhang2026discrete}.

\begin{figure}[ht]
    \centering
    \begin{subfigure}[b]{0.49\textwidth}
        \centering
\begin{tikzpicture}[
    scale=0.8, every node/.style={transform shape}, 
    node distance=1cm and 1.2cm,
    every node/.style={draw, circle, minimum size=0.8cm},
    latent/.style={draw, circle, dashed},
    obs/.style={draw, circle},
    arrow/.style={->, thick}
]

\node[latent] (U) {$U$};
\node[obs, below left=of U]  (Z1) {$Z_1$};
\node[obs, below=of U]       (Z2) {$Z_2$};
\node[obs, below right=of U] (Z3) {$Z_3$};
\node[obs, below left=of Z2]  (A) {$A$};
\node[obs, below right=of Z2] (Y) {$Y$};

\draw[arrow] (U) -- (Z1);
\draw[arrow] (U) -- (Z2);
\draw[arrow] (U) -- (Z3);
\draw[arrow] (Z2) -- (Z3);
\draw[arrow] (Z2) -- (Z1);
\draw[arrow] (U) -- (A);
\draw[arrow] (U) -- (Y);
\draw[arrow] (A) -- (Y);
\draw[arrow] (Z2) -- (A);
\draw[arrow] (Z2) -- (Y);
\draw[arrow] (Z1) -- (A);
\draw[arrow] (Z3) -- (Y);

\end{tikzpicture}
\caption{DAG used in existing proximal causal inference}
        \label{fig:dag_proxy}
    \end{subfigure}
    \hfill
    \begin{subfigure}[b]{0.49\textwidth}
        \centering
\begin{tikzpicture}[
     scale=0.8, every node/.style={transform shape}, 
    node distance=1cm and 1.2cm,
    every node/.style={draw, circle, minimum size=0.8cm},
    latent/.style={draw, circle, dashed},
    obs/.style={draw, circle},
    arrow/.style={->, thick}
]

\node[latent] (U) {$U$};
\node[obs, below left=of U]  (Z1) {$Z_1$};
\node[obs, below=of U]       (Z2) {$Z_2$};
\node[obs, below right=of U] (Z3) {$Z_3$};
\node[obs, below left=of Z2]  (A) {$A$};
\node[obs, below right=of Z2] (Y) {$Y$};

\draw[arrow] (U) -- (Z1);
\draw[arrow] (U) -- (Z2);
\draw[arrow] (U) -- (Z3);
\draw[arrow] (U) -- (A);
\draw[arrow] (U) -- (Y);
\draw[arrow] (A) -- (Y);

\foreach \z in {Z1,Z2,Z3} {
    \draw[arrow] (\z) -- (A);
    \draw[arrow] (\z) -- (Y);
}
\end{tikzpicture}
\caption{Our multi-proxy DAG}
        \label{fig:our_dag}
    \end{subfigure}
    \caption{All the proxies ($Z_1,Z_2,Z_3$) play a symmetric role in our DAG, but not in the existing proximal causal inference framework.}
    \label{fig:dag}
\end{figure}
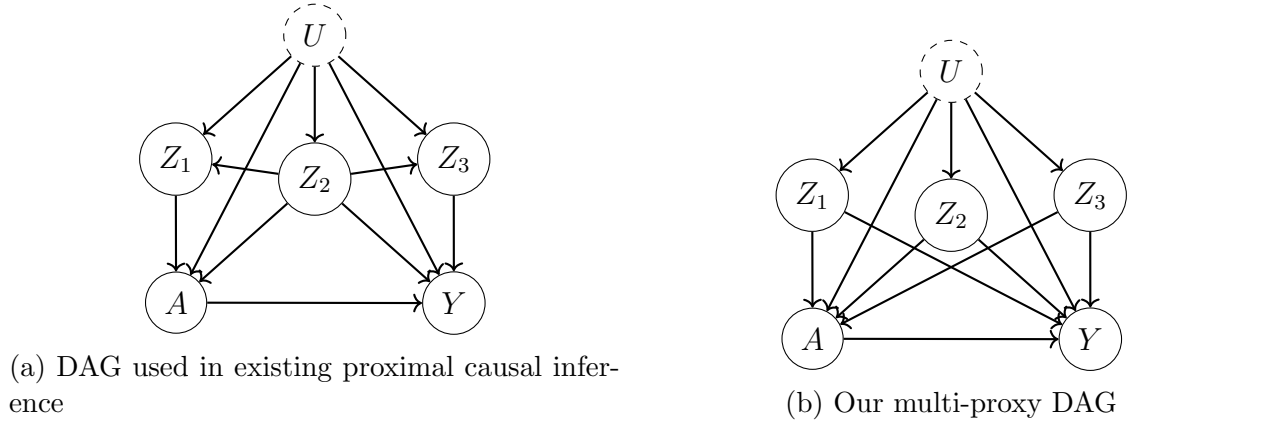

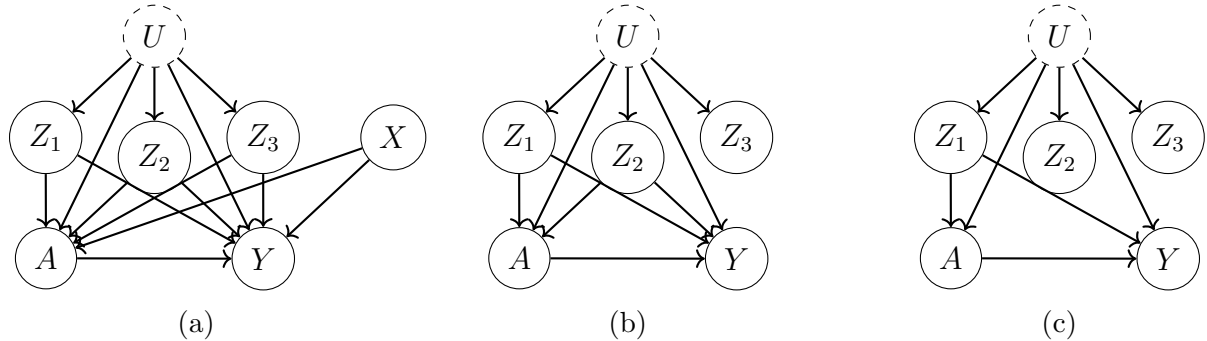
\begin{figure}[ht]
    \centering
    \begin{subfigure}[b]{0.3\textwidth}
        \centering
\begin{tikzpicture}[
     scale=0.8, every node/.style={transform shape}, 
    node distance=0.7cm and 0.8cm,
    every node/.style={draw, circle, minimum size=0.4cm},
    latent/.style={draw, circle, dashed},
    obs/.style={draw, circle},
    arrow/.style={->, thick}
]

\node[latent] (U) {$U$};
\node[obs, below left=of U]  (Z1) {$Z_1$};
\node[obs, below=of U]       (Z2) {$Z_2$};
\node[obs, below right=of U] (Z3) {$Z_3$};
\node[obs, right=of Z3] (X) {$X$};
\node[obs, below left=of Z2]  (A) {$A$};
\node[obs, below right=of Z2] (Y) {$Y$};

\draw[arrow] (U) -- (Z1);
\draw[arrow] (U) -- (Z2);
\draw[arrow] (U) -- (Z3);
\draw[arrow] (U) -- (A);
\draw[arrow] (U) -- (Y);
\draw[arrow] (A) -- (Y);
\draw[arrow] (X) -- (A);
\draw[arrow] (X) -- (Y);
\foreach \z in {Z1, Z2, Z3} {
    \draw[arrow] (\z) -- (A);
    \draw[arrow] (\z) -- (Y);
}
\end{tikzpicture}
\caption{}
        \label{fig:our-dag-1}
\end{subfigure}
    \hfill
        \begin{subfigure}[b]{0.3\textwidth}
        \centering
\begin{tikzpicture}[
    scale=0.8, every node/.style={transform shape}, 
    node distance=0.7cm and 0.8cm,
    every node/.style={draw, circle, minimum size=0.4cm},
    latent/.style={draw, circle, dashed},
    obs/.style={draw, circle},
    arrow/.style={->, thick}
]

\node[latent] (U) {$U$};
\node[obs, below left=of U]  (Z1) {$Z_1$};
\node[obs, below=of U]       (Z2) {$Z_2$};
\node[obs, below right=of U] (Z3) {$Z_3$};
\node[obs, below left=of Z2]  (A) {$A$};
\node[obs, below right=of Z2] (Y) {$Y$};

\draw[arrow] (U) -- (Z1);
\draw[arrow] (U) -- (Z2);
\draw[arrow] (U) -- (Z3);
\draw[arrow] (U) -- (A);
\draw[arrow] (U) -- (Y);
\draw[arrow] (A) -- (Y);

\foreach \z in {Z1,Z2} {
    \draw[arrow] (\z) -- (A);
    \draw[arrow] (\z) -- (Y);
}
\end{tikzpicture}
\caption{}
        \label{fig:our_dag-2}
    \end{subfigure}
    \hfill
    \begin{subfigure}[b]{0.3\textwidth}
        \centering
\begin{tikzpicture}[
    scale=0.8, every node/.style={transform shape}, 
    node distance=0.7cm and 0.8cm,
    every node/.style={draw, circle, minimum size=0.4cm},
    latent/.style={draw, circle, dashed},
    obs/.style={draw, circle},
    arrow/.style={->, thick}
]

\node[latent] (U) {$U$};
\node[obs, below left=of U]  (Z1) {$Z_1$};
\node[obs, below=of U]       (Z2) {$Z_2$};
\node[obs, below right=of U] (Z3) {$Z_3$};
\node[obs, below left=of Z2]  (A) {$A$};
\node[obs, below right=of Z2] (Y) {$Y$};

\draw[arrow] (U) -- (Z1);
\draw[arrow] (U) -- (Z2);
\draw[arrow] (U) -- (Z3);
\draw[arrow] (U) -- (A);
\draw[arrow] (U) -- (Y);
\draw[arrow] (A) -- (Y);

\foreach \z in {Z1} {
    \draw[arrow] (\z) -- (A);
    \draw[arrow] (\z) -- (Y);
}
\end{tikzpicture}
\caption{}
        \label{fig:our_dag-3}
    \end{subfigure}
    \hfill
    \caption{Variations of the model in \Cref{fig:our_dag} that are accommodated by our framework but are not addressed by existing proximal causal inference methods.}
    \label{fig:dag-special-cases}
\end{figure}

\textbf{1. Multi-proxy setting.} We consider a causal inference framework with multiple proxies
where classification of proxy variables is not required. Instead,  we need at least three proxy variables  
that are conditionally independent given the latent confounder. Specifically, consider \Cref{fig:our_dag}. In many modern causal inference problems, treatment assignment and outcomes are jointly influenced by latent factors that are not directly observed, but instead manifest through multiple noisy proxy measurements. When these proxies arise from distinct data-generating mechanisms, it is often reasonable to model them as conditionally independent given the latent confounder. 
For instance:
\begin{itemize}
\item \textbf{E-commerce:} Consider the problem of estimating the causal effect of different discount rates on an online platform. A user's latent purchase intent or preference type is reflected through:
clickstream data (product page visits, dwell time); transactional history (cart abandonment rates, past purchase frequency); search queries (specific keywords vs. broad category browsing); and customer support interactions (chat logs, return requests).
\item \textbf{Healthcare:} A patient's disease severity or subtype (unobserved) is simultaneously reflected in medical imaging (e.g., radiomic features extracted from chest X-rays or CT scans), {high-frequency vitals} (e.g., heart rate variability), and clinical notes. Here, the treatment may correspond to the dosage of a drug, and the outcome may correspond to recovery.
\item \textbf{Genomics:} In biological studies, a cell’s latent state—such as its differentiation stage or activation status—acts as an unobserved confounder. The {treatment} may represent an experimental perturbation, such as gene knockout, drug exposure, or environmental stimulus, while the {outcome} may correspond to a phenotypic response, such as cell survival. This latent state gives rise to multiple proxy measurements, including transcriptomic profiles, chromatin accessibility, protein abundance, and epigenetic marks.
\end{itemize}
 Our causal estimation method proceeds in three stages: 
First, we recover the posterior distribution of the latent confounder for each unit using tensor decomposition. 
Second, we use these recovered posteriors to estimate a cluster-specific treatment model.
And finally, we estimate the cluster-specific outcome model. 

\textbf{2. Multi-treatment setting.}
We consider a framework without proxy, but with multiple (at least three) conditionally independent treatments. The ``Deconfounder” \citep{wang2019blessings} advocated the use of multiple simultaneous treatments to infer latent structure, which utilizes standard factor models to estimate a substitute for the unobserved confounder from multiple assigned causes. However, their assumptions and the validity of their identifiability results are debated \citep{ogburn2019comment}. While the conceptual idea is similar to that in \cite{wang2019blessings}, we derive a new identifiability result by reformulating the multiple-treatment setting as an identifiable mixture learning problem, assuming that the latent confounder $U$ takes finitely many values. To be precise, we consider the DAG in \Cref{fig:multi-tr-dag}. 

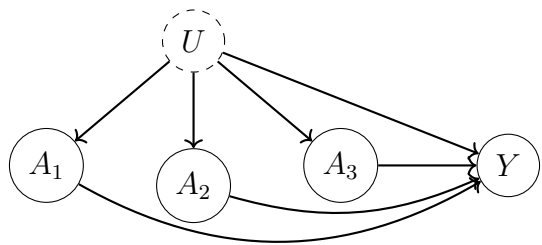
\begin{figure}[ht]
        \centering
\begin{tikzpicture}[
    scale=0.8, every node/.style={transform shape}, 
    node distance=1cm and 1.3cm,
    every node/.style={draw, circle, minimum size=0.8cm},
    latent/.style={draw, circle, dashed},
    obs/.style={draw, circle},
    arrow/.style={->, thick}
]

\node[latent] (U) {$U$};
\node[obs, below left=of U]  (A1) {$A_1$};
\node[obs, below=of U]       (A2) {$A_2$};
\node[obs, below right=of U] (A3) {$A_3$};
\node[obs, right =of A3] (Y) {$Y$};

\draw[arrow] (U) -- (A1);
\draw[arrow] (U) -- (A2);
\draw[arrow] (U) -- (A3);
\draw[arrow] (U) -- (Y);
\draw[arrow] (A1) to [bend right=30] (Y);
\draw[arrow] (A2) to [bend right=20] (Y);
\draw[arrow] (A3) -- (Y);

\end{tikzpicture}
\caption{Our multi-treatment DAG}
    \label{fig:multi-tr-dag}
\end{figure}

\paragraph{Contributions.} In summary, our key contributions are:
\begin{enumerate} 
\item  A multi-proxy causal framework (as depicted in \Cref{fig:our_dag}) that differs from existing proximal causal inference methods in that it does not require partitioning proxies into \textit{negative control exposures} and \textit{negative control outcomes}. It utilizes any collection of three or more proxies that are conditionally independent given the latent confounder.

\item The identifiability of the model parameters and hence of causal estimands, CATE, and ATE
is established in \Cref{thm:identifiability} for our multi-proxy model, under some parametric assumptions on the treatment and outcome model. Further, an efficient estimation procedure using tensor decomposition to compute the posteriors, followed by a three-stage regression.  

\item  We extend our multi-proxy identifiability in \Cref{thm:identifiability} to a non-parametric setting in \Cref{thm:identifiability-gen}. Developing algorithms for estimating causal effects within this nonparametric framework, as well as establishing finite-sample guarantees using nonparametric methods such as kernel ridge regression, is left for future work.

\item The identifiability of the model parameters in the multi-treatment setting (\Cref{fig:multi-tr-dag}) and hence of causal estimands, CATE, and ATE is established in \Cref{thm:identifiability-no-proxy}, with an efficient estimation procedure in \Cref{alg:multi_treatment_estimation}.

\item A nonasymptotic error bound for the nonparametric estimation of the latent posterior distributions via tensor decomposition is provided. Further, finite-sample error bounds for the subsequent stages of the estimation procedure are provided, 
culminating in a non-asymptotic bound (and hence, a non-asymptotic confidence interval) for the causal effect in \Cref{thm:ate-error}.

\end{enumerate}
In many real-world settings, the latent confounder can be well approximated by a discrete distribution. We demonstrate the effectiveness of our method on the BWGHT dataset \citep{wooldridge2010econometric}, where it recovers a trivalent confounder corresponding to three socioeconomic classes and estimates the causal effects of maternal smoking on infant birth weight.

\medskip
\noindent{\bf Related work.} Our work intersects with two primary bodies of literature: causal inference with unobserved confounding and mixture learning -- specifically, the method of moments via tensor decomposition for latent 
variable models.

\smallskip
\noindent{\em Causal Inference with Unobserved Confounders.}
The standard approach to causal inference assumes unconfoundedness, requiring all confounders to be observed \citep{rosenbaum1983central}. When confounders are unobserved, identification typically relies on the instrumental variables \citep{angrist1996identification} or sensitivity analysis \citep{rosenbaum2002observational}. Recently, there has been significant interest in leveraging \textit{proxy variables} (or, negative controls) to identify causal effects: \cite{miao2018identifying} and \cite{tchetgen2020introduction} established non-parametric identification results using proxies, framing the problem as the solution to some integral equations. \cite{mastouri2021proximal} use Reproducing Kernel Hilbert Spaces (RKHS) in the proximal causal inference framework. While powerful, these methods often face statistical challenges related to the ill-posed nature of inverse problems, as discussed before. Several recent works have applied tensor decompositions to causal problems. \cite{wang2019blessings} proposed the ``Deconfounder," which uses factor models to adjust for multiple causes, though its identification conditions have been debated \citep{ogburn2019comment,ogburn2020counterexamples}. \cite{amjad2018robust, kallus2018causal, athey2021matrix} simultaneously introduced matrix factorization to handle confounding in panel data with sparse observations. This connects directly to the Synthetic Control Method \citep{abadie2010synthetic}, which constructs counterfactuals as convex combinations of untreated units. Recent work has reinterpreted and extended synthetic controls in different ways: \cite{amjad2018robust} proposed ``Robust Synthetic Control" via singular value thresholding; \cite{agarwal2020synthetic} introduced ``Synthetic Interventions" using principal component regression for multi-unit interventions; \cite{athey2021matrix,agarwal2023causal} employed matrix completion frameworks for estimating causal effects in panel data; \cite{arkhangelsky2021synthetic} proposed ``Synthetic Difference-in-Differences," bridging synthetic control methods with fixed-effects models to provide doubly-robust estimation.

\smallskip
\noindent{\em Causal inference and mixture learning.} We believe that this connection between causal inference and finite mixture models is fundamental, best illustrated by the classic Simpson’s Paradox \cite{blyth1972simpson,wagner1982simpson,pearl2011simpson}. In this setting, the population is composed of distinct latent subpopulations (e.g., high-risk vs. low-risk patients), each with its own treatment mechanism and outcome response. When these groups are aggregated, the true causal effect is obscured—or even reversed—because the observed data is a weighted mixture of these heterogeneous distributions. Thus, the problem of identifying the causal effect in the presence of unobserved categorical confounding is mathematically equivalent to the problem of identifying the components of a finite mixture model. 
In the classical framework, if the confounder $U$ is observed, causal identification is achieved via the ``backdoor criterion" \citep{pearl2009causality}, which stratifies the population to re-weight the heterogeneous treatment effects. We address the problem of unobserved $U$ by treating the problem as one of \textit{latent mixture identification}. We believe this work establishes a concrete bridge between mixture learning and causal inference, opening the door to new methodological and theoretical developments at their intersection.

\smallskip
\noindent{\em Mixture Learning and Tensor Decomposition.}
Our identification strategy draws heavily from the literature on the method of moments for latent variable models. Unlike Expectation-Maximization (EM) algorithms, which are prone to local optima and lack finite-sample guarantees, \cite{allman2009identifiability} established the foundational identifiability of latent class models given three conditionally independent views, based on \cite{kruskal1977three}.   \cite{anandkumar2012method} and \cite{hsu2012spectral} extended these results to provide efficient algorithms for parameter estimation in multi-view models, e.g., Hidden Markov models, topic models, and Gaussian mixture models. The \textit{Robust Tensor Power Method} by \cite{anandkumar2014tensor} provides global convergence guarantees and finite-sample error bounds under mild non-degeneracy conditions. 



\medskip
\noindent{\bf Organization.} The rest of the paper is organized as follows. \Cref{sec:background} provides a brief overview of existing results that we later use in our algorithm and analysis. \Cref{sec:identifiability} introduces our initial causal framework with multiple proxies and establishes the identifiability of the causal parameters, which we later generalize in \Cref{sec:gen-identifiability}. In \Cref{sec:alg}, we present an estimation algorithm. In \Cref{sec:multi-treatment}, we establish identifiability and estimation of causal effects in the multi-treatment setting. \Cref{sec:theory} analyzes the finite sample errors in estimation. Finally, we demonstrate the empirical effectiveness of our method through simulation and real data experiments in \Cref{sec:expt}. We conclude the paper in \Cref{sec:conclusion}. Omitted proofs and additional results are provided in the Appendix.


\section{Background}
\label{sec:background}
\subsection{Identifiability of Nonparametric Finite Mixture Models}
\label{app:nonparametric-mixtures}

We summarize classical identifiability results for nonparametric finite mixtures
of product measures, following \cite{allman2009identifiability}.

Let $d_1,\dots,d_p$ be positive integers with $p \ge 3$, and let
$\mathbb P_{v,u}$ be probability measures on $\mathbb{R}^{d_v}$.
Define $d_Z = \sum_{v=1}^p d_v$, and consider the mixture distribution on
$\mathbb{R}^d$ given by
\begin{equation}
\label{eq:block-mixture}
\mathbb{P}
=
\sum_{u=1}^K \pi_u \prod_{v=1}^p \mathbb P_{v,u}.
\end{equation}

\begin{theorem}[\cite{allman2009identifiability}]
\label{thm:block-identifiability}
Let $\mathbb{P}$ be of the form \eqref{eq:block-mixture} and $\pi_u>0,$ for all $u=1,\cdots,K$. Suppose that for every
$v \in \{1,\dots,p\}$, the measures
$\{\mathbb P_{v,u}\}_{u=1}^K$ on $\mathbb{R}^{d_v}$ are linearly independent.
If $p \ge 3$, then the parameters
$\{\pi_u, \mathbb P_{v,u}\}_{u=1,\dots,K;\,v=1,\dots,p}$
are strictly identifiable from $\mathbb{P}$, up to permutations.
\end{theorem}

\subsection{Nonparametric Estimation of Latent Variable Models}
\label{subsec:tensor-method}

Given $n$ observations $\{(z_1^i, z_2^i, z_3^i)\}_{i \in [n]}$ drawn \textit{i.i.d.} from a multi-view latent variable model of the form \eqref{eq:block-mixture}, \cite{song2014nonparametric} introduced a kernel-based spectral method to learn $\{\pi_u, \mathbb P_{v,u}\}_{u=1,\dots,K;\,v=1,\dots,p}$, by embedding distributions into Reproducing Kernel Hilbert Spaces (RKHS) and using tensor power methods to recover parameters efficiently.

Consider an RKHS on $\mathcal{X}$ with a kernel $k(x,x^\prime)$.  Alternatively, $k(x, \cdot)$ can be viewed as an
implicit feature map $\rho(x)$ where  $k(x,x^\prime)=\langle \rho(x),\rho(x^\prime)\rangle$. Popular kernel functions on $\mathbb{R}^d$ include the Gaussian RBF kernel and the Laplace kernel.

For simplicity, we assume $p=3$ and the views are symmetric, i.e., the conditional distributions $\mathbb{P}_{1,u}=\mathbb{P}_{2,u}=\mathbb{P}_{3,u}$ for each $u$ (we denote them as $\mathbb{P}_{u}$). We embed each $\mathbb{P}_{u}$ into the RKHS as
\begin{equation*} \mu_{Z|U=k}=\int_{\mathcal{X}}\rho(z)\mathbb{P}_k(dz).
\end{equation*}
Conceptually, we can collect
these into a matrix (with potentially an infinite number of rows)
\begin{equation*}
 C_{Z|U} := (\mu_{Z|U=1},\cdots, \mu_{Z|U=K}).
\end{equation*}
Once we have the conditional embedding $\mu_{Z|U=k}$, we have the density $p_{k}$ corresponding to $\mathbb{P}_{k}$ by performing an inner product $p_{k}(z)=\langle\rho(z),\mu_{Z|U=k}\rangle$.

Although the empirical kernel embeddings can be infinite-dimensional, we can carry out the decomposition using just the kernel matrices. We denote the implicit feature matrix by
\begin{align*}
\Sigma &:= (\rho(z_1^1), \dots, \rho(z_1^n), \rho(z_2^1), \dots, \rho(z_2^n)), \\
\Lambda &:= (\rho(z_2^1), \dots, \rho(z_2^n), \rho(z_1^1), \dots, \rho(z_1^n)),
\end{align*}
and the corresponding kernel matrix by $\Omega = \Sigma^\top \Sigma$ and $L = \Lambda^\top \Lambda$ respectively. 

\textbf{Step 1.} We perform a kernel eigenvalue decomposition of the empirical 2nd order embedding
\[
  \widehat{C}_{Z_1 Z_2} := \frac{1}{2n} \sum_{i=1}^n \left( \rho(z_1^i) \otimes \rho(z_2^i) + \rho(z_2^i) \otimes \rho(z_1^i) \right),
  \]
which can be expressed succinctly as $\widehat{C}_{Z_1 Z_2} = \frac{1}{2n} \Sigma \Lambda^\top$. Its leading $K$ eigenvectors $\widehat{\mathcal{U}}_K = (\widehat{u}_1, \dots, \widehat{u}_K)$ lie in the span of the columns of $\Sigma$, \textit{i.e.}, $\widehat{\mathcal{U}}_K = \Sigma (\gamma_1, \dots, \gamma_K)$ with $\gamma_i \in \mathbb{R}^{2n}$. Then we can transform the eigenvalue decomposition problem for an infinite-dimensional matrix to a problem involving finite-dimensional kernel matrices,
\begin{align*}
\widehat{C}_{Z_1 Z_2} \widehat{C}_{Z_1 Z_2}^\top u = \widehat{\lambda}^2 u \ 
&\Rightarrow \ \frac{1}{4n^2} \Sigma \Lambda^\top \Lambda \Sigma^\top \Sigma \gamma = \widehat{\lambda}^2 \Sigma \gamma \\
&\Rightarrow \ \frac{1}{4n^2} \Omega L \Omega \gamma = \widehat{\lambda}^2 \Omega \gamma.
\end{align*}
Let the Cholesky decomposition of $\Omega$ be $R^\top R$. Then by redefining $\widetilde{\gamma} = R \gamma$, we solve the problem
\begin{equation} 
\frac{1}{4n^2} R L R^\top \widetilde{\gamma} = \widehat{\lambda}^2 \widetilde{\gamma}, \ \ \text{and obtain } \gamma = R^\dagger \widetilde{\gamma}.
\end{equation}
The resulting eigenvectors satisfy $(\hat U^{(i)})^\top \hat u_{i'} = \gamma_i^\top
\Sigma^\top \Sigma \gamma_{i'} = \gamma_i^\top \Omega \gamma_{i'} = 
\widetilde{\gamma}_i^\top \widetilde{\gamma}_{i'} = \delta_{ii'}$. 

\textbf{Step 2.}
Define the whitening operator $\widehat{\mathcal{W}} := \widehat{\mathcal{U}}_K \widehat{S}_K^{-1/2}$ and apply it on the empirical 3rd order embedding
\begin{align*}
    \widehat{C}_{Z_1 Z_2 Z_3} &:= \frac{1}{3n} \sum_{i=1}^n (\rho(z_1^i) \otimes \rho(z_2^i) \otimes \rho(z_3^i) + \rho(z_3^i) \otimes \rho(z_1^i) \otimes \rho(z_2^i) + \rho(z_2^i) \otimes \rho(z_3^i) \otimes \rho(z_1^i))
\end{align*}
to compute  
\begin{align*}
    \widehat{\mathcal{T}}
    &= \frac{1}{3n} \sum_{i=1}^n (\xi(z_1^i) \otimes \xi(z_2^i) \otimes \xi(z_3^i) + \xi(z_3^i) \otimes \xi(z_1^i) \otimes \xi(z_2^i) + \xi(z_2^i) \otimes \xi(z_3^i) \otimes \xi(z_1^i)),
\end{align*}
where
\[
    \xi(z) := \widehat{\mathcal{W}}^\top\rho(z)= \widehat{S}_K^{-1/2} (\gamma_1, \dots, \gamma_K)^\top \Sigma^\top \rho(z) \ \in \ \mathbb{R}^K.
\]
\textbf{Step 3.} We run the tensor power method \cite{anandkumar2014tensor} on the finite dimension tensor $\widehat{\mathcal{T}}$ to obtain its leading $K$ eigenvectors $\widehat{M} := (\widehat{v}_1, \dots, \widehat{v}_K)$, and the corresponding eigenvalues $\widehat{\lambda} := (\widehat{\lambda}_1, \dots, \widehat{\lambda}_K)$.

\textbf{Step 4.} The estimates of the conditional embeddings are
\[
    (\widehat{\mu}_{Z|U=1}, \dots, \widehat{\mu}_{Z|U=K}) =  \Sigma(\gamma_1, \dots, \gamma_K) \widehat{S}_K^{1/2} \widehat{M},
\]
and the estimates of the class probabilities are
\[(\hat\pi_1, \dots,\hat\pi_K)=(\widehat{\lambda}_1^{-2}, \dots, \widehat{\lambda}_K^{-2}).\]
Let $\rho_0 := \sup_{x \in \mathcal{X}} k(x, x)$, $\|\cdot\|$ be the Hilbert-Schmidt norm, $\pi_{\min} := \min_{i \in [K]} \pi_i$ and $\sigma_K(\mathcal{C}_{Z_1 Z_2})$ be the $K$-th largest singular value of $$\mathcal{C}_{Z_1 Z_2}:=\sum_{u=1}^K\pi_u\cdot\mu_{Z|U=u}\otimes\mu_{Z|U=u}.$$

\section{Multi-proxy Causal Inference Framework}
\label{sec:identifiability}

The latent variable $U$ is \emph{unobserved}, categorical, and takes values in $\{1,\dots,K\}$.
The proxy variables $Z=(Z_1^\top,Z_2^\top,Z_3^\top)^\top$ are observed. 
In our setup, treatment $A$ and outcome $Y$ are both continuous. Our goal is to estimate the average treatment effect (ATE), 
$$\tau(a):=\mathbb E^{do(A=a)}[Y].$$
Here $\mathbb E^{do(A=a)}$ denotes the expectation under the interventional distribution under treatment $A = a$.
We are also interested in estimating the conditional average treatment effect (CATE), $$\tau_u(a,z)=\mathbb E^{do(A=a)}[Y\mid U=u,Z=z].$$

In this section, we state our model assumptions and focus on the identifiability of these model parameters and causal effects in the population limit. The transition from identification to estimation in the finite-sample regime is detailed in subsequent sections. Our identification strategy rests upon the following assumptions.
\begin{assumption}[Latent Conditional Independence]
\label{ass:indep}
The proxy variables $Z_v\in\mathbb R^d$, for $v=1,2,3$ are conditionally independent given the latent state $U$:$Z_1 \perp Z_2 \perp Z_3 \mid U.$\end{assumption}
This is the standard ``Naive Bayes" or ``Multi-view" assumption. 

\begin{assumption}[Linear independence of measures]
\label{ass:lin-ind-z}
Let $\mathbb P_{v,u}$ be the conditional distribution of $Z_v,$ for $v \in \{1, 2, 3\}$ given the latent state $U=u$. Then, we assume that the measures $\{\mathbb P_{v,u}\}_{1 \le u\le K}$ are linearly independent for each $v=1,2,3.$
\end{assumption}
\Cref{ass:indep,ass:lin-ind-z} ensure identifiability and the uniqueness of the tensor factorization of the proxies \cite{allman2009identifiability,song2014nonparametric}. It says that the clusters are mathematically distinguishable. 

\begin{assumption}
\label{ass:positivity-1}
The class weights $\pi_u=P(U=u)>0$, for all $u=1,\cdots,K.$
\end{assumption}

\begin{assumption}[Treatment model]
\label{ass:treatment-model-cont}
The treatment model is parameterized by $[\alpha_1^\top, \cdots, \alpha_K^\top]^\top\in\mathbb{R}^{L\times K},$ and $ [\sigma_1^2, \cdots, \sigma_K^2]\in\mathbb{R}_{\geq 0}^K$ and 
$$A\mid (Z=z,U=u)\sim H(\alpha_u^\top \phi(z),\sigma^2_u),$$ where $\phi=(\phi_1,\cdots,\phi_{L})^\top:\mathbb{R}^{3d}\to\mathbb R^L$ is some known function, and $H(\mu,s)$ is a known distribution family parameterized by mean $\mu$ and variance $s$.
\end{assumption}
\begin{assumption}[Outcome model]
\label{ass:outcome-model-cont}
The outcome model is parameterized by $[\beta_1^\top, \cdots, \beta_K^\top]^\top\in\mathbb{R}^{M\times K}$:
 $$\mathbb{E}[Y \mid A=a, Z=z, U=u] = \beta_u^\top \psi(a, z),$$ where $\psi=(\psi_1,\cdots,\psi_{M})^\top :\mathbb{R}\times\mathbb{R}^{3d}\to\mathbb R^M$ is some known function.
\end{assumption}
Define the functions $f_u(z)=\prod_{v=1}^3 f_{v,u}(z_v)$ and $e_u(a, z) = h(a;\alpha_u^\top \phi(z),\sigma^2_u)$, where $h(\cdot ;\alpha_u^\top \phi(z),\sigma^2_u)$ is the density of the distribution $H(\alpha_u^\top \phi(z),\sigma^2_u)$ and $f_{v,u}$ is the density of $\mathbb P_{v,u}$.
We also define the vectors:
\begin{equation}
\label{eq:basis}
 {W}^*(z)= \begin{bmatrix} f_1(z)  \\ \vdots \\ f_K(z) \end{bmatrix}, \quad \Phi^*(z) = \begin{bmatrix} f_1(z) \phi(z) \\ \vdots \\ f_K(z) \phi(z) \end{bmatrix}, \quad\Psi^*(a, z) = \begin{bmatrix} f_1(z)e_1(a, z) \psi(a, z) \\ \vdots \\ f_K(z)e_K(a, z) \psi(a, z) \end{bmatrix}.
\end{equation}

\begin{assumption}
\label{ass:invertibility-cont}
The matrices $\mathbb{E}[{W}^*(Z){W}^*(Z)^\top], \mathbb{E}[\Phi^*(Z)\Phi^*(Z)^\top]$ and $\mathbb{E}[\Psi^*(A, Z)\Psi^*(A, Z)^\top]$ are all positive definite (PD).
\end{assumption}

We provide a sufficient condition for \Cref{ass:invertibility-cont} using the following well-known result.

\begin{lemma}
\label{lem:ae}
Let $f: \mathbb{R}^{3d} \to \mathbb{R}$ be a real continuous function. If $Z$ has full support over $\mathbb{R}^{3d}$ (i.e., every open set in $\mathbb{R}^{3d}$ has positive probability), then $\mathbb P[f(Z)=0]=1$ implies $f$ is identically zero.
\end{lemma}

When $Z$ and $(A,Z)$ have full supports over $\mathbb{R}^{3d}$ and $\mathbb{R}\times\mathbb{R}^{3d}$ respectively, then the lemma shows that:
\begin{itemize}
    \item If the set of functions $\{ f_i(z)\}_{i=1,\cdots,K}$ are continuous and linearly independent over $\mathbb{R}^{3d}$, then\\
    $\mathbb{E}[{W}^*(Z){W}^*(Z)^\top]$ is PD.
     \item If the set of functions $\{ f_i(z)\phi_j(z)\}_{i=1,\cdots,K,j=1,\cdots,L}$ are continuous and linearly independent over $\mathbb{R}^{3d}$, then $\mathbb{E}[\Phi^*(Z)\Phi^*(Z)^\top]$ is PD.
    \item If the set of functions $\{ f_i(z)e_i(a,z)\psi_j(a,z)\}_{i=1,\cdots,K,j=1,\cdots,M}$ are continuous and linearly independent over $\mathbb{R}\times\mathbb{R}^{3d}$, then $\mathbb{E}[\Psi^*(A, Z)\Psi^*(A, Z)^\top]$ is PD.
\end{itemize}

\subsection{One concrete model example}
\label{concrete-example}
\label{subsec:example-cont}
Consider a population with $K=2$ latent subtypes with equal prevalence ($\pi_u = 0.5,$ for $u=1,2$).
\begin{itemize}
    \item \textbf{Proxies:} The three views $Z=(Z_1,Z_2,Z_3)^\top$ are generated conditional on $U\in\{1,2\}$ with distinct centers:
    \[ Z | U=u \sim \mathcal{N}(\boldsymbol{\mu}_u,\sigma^2 \mathbf{I}_3), \quad  \]
    $\text{with } \boldsymbol{\mu}_u = [\mu_{u,1}, \mu_{u,2}, \mu_{u,3}]^\top$ such that $\mu_{1,v}\neq\mu_{2,v}$, for each $v=1,2,3.$
    \item \textbf{Treatment:} Treatment depends on the latent $U$ and the feature vector $\phi(z) = [1, z^\top]^\top$.
    \[ A | Z=z, U=u \sim \mathcal{N}(\alpha_u^\top \phi(z), \sigma_A^2),\quad\alpha_1\neq\alpha_2. \]
    \item \textbf{Outcome:} Outcome depends on the latent state, treatment, and all proxies.
    \[ Y | A=a, Z=z, U=u \sim \mathcal{N}(\beta_u^\top \Psi(a,z), \sigma_Y^2), \quad \text{with } \Psi(a,z) = [1, a, z^\top]^\top. \]
\end{itemize}
\Cref{ass:indep,ass:positivity-1,ass:treatment-model-cont,ass:outcome-model-cont} hold directly from construction. Also, it is known that Gaussian densities with distinct means are linearly independent \cite{yakowitz1968identifiability}. Hence, $\mu_{1,v}\neq\mu_{2,v}$ ensures that \Cref{ass:lin-ind-z} holds. The following proposition verifies \Cref{ass:invertibility-cont}; the proof is provided in the appendix.
\begin{proposition}
\label{prop:pd-verify}
Under the above setup, the matrices $\mathbb{E}[W^*(Z)W^*(Z)^\top]$, $\mathbb{E}[\Phi^*(Z)\Phi^*(Z)^\top]$ and\\ $\mathbb{E}[\Psi^*(A,Z)\Psi^*(A,Z)^\top]$ are positive definite.
\end{proposition}

Next, we show that the structural parameters of the causal model are uniquely determined by the observable joint distribution of $(Z_1, Z_2, Z_3, A, Y)$ up to permutation.

\subsection{Identification in the multi-proxy model}
\begin{theorem}
\label{thm:identifiability}
Under \Cref{ass:indep,ass:lin-ind-z,ass:positivity-1,ass:treatment-model-cont,ass:outcome-model-cont,ass:invertibility-cont}, both the treatment model parameters $\{\alpha_u,\sigma_u^2\}_{u=1}^K$ and outcome model parameters $\{\beta_u\}_{u=1}^K$ are uniquely identifiable (up to permutation) from the joint distribution of $(Z_1, Z_2, Z_3, A, Y)$. Moreover, the average treatment effect, 
$\tau(a)=\mathbb E^{do(A=a)}[Y]$ is uniquely identifiable, and the conditional average treatment effect, $\tau_u(a,z)=\mathbb E^{do(A=a)}[Y\mid U=u,Z=z]$ is uniquely identifiable up to permutations.
\end{theorem}

\begin{proof} We provide the proof in a few steps as follows.

\textbf{Step 1.} Under \Cref{ass:indep,ass:lin-ind-z,ass:positivity-1}, Theorem 8 of \cite{allman2009identifiability} (see \Cref{app:nonparametric-mixtures}) guarantees that the parameters $\pi_u$ and the distributions $\mathbb P_{v,u}$ are uniquely identified (up to permutation). Let $f_{v,u}$ be the density corresponding to $\mathbb P_{v,u}$. With the densities and priors identified, the posterior membership probability is uniquely determined by Bayes' rule:
\begin{equation}
\label{eq:post-mem-prob-cont}
    w_u(Z)  = P(U=u \mid Z)=\frac{\pi_u \prod_{v=1}^3 f_{v,u}(Z_v)}{\sum_{k=1}^K \pi_k \prod_{v=1}^3 f_{v,k}(Z_v)},
\end{equation}
here we use the notation $Z=(Z_1,Z_2,Z_3).$

\textbf{Step 2.} Using \Cref{ass:treatment-model-cont},
$$\mathbb{E}[A \mid Z] =\sum_{u=1}^KP(U=u \mid Z) \mathbb{E}[A \mid Z,U=u]= \sum_{u=1}^K w_u(Z) \cdot \alpha_u^\top \phi(Z).$$
 Define the vectors $\boldsymbol{\alpha} \in \mathbb{R}^{K \cdot L}$ and $\Phi(Z) \in \mathbb{R}^{K \cdot L}$ as:
\begin{equation}
    \label{eq:params-cont}\boldsymbol{\alpha} = \begin{bmatrix} \alpha_1 \\ \vdots \\ \alpha_K \end{bmatrix}, \quad \Phi(Z) = \begin{bmatrix} w_1(Z) \phi(Z) \\ \vdots \\ w_K(Z) \phi(Z) \end{bmatrix}.
\end{equation}
We have, $\mathbb{E}[\Phi(Z)A]=\mathbb{E}[\Phi(Z)\mathbb{E}[A\mid Z]]=\mathbb{E}[\Phi(Z)\Phi(Z)^\top] \boldsymbol{\alpha}.$
Since $\mathbb{E}[\Phi(Z)\Phi(Z)^\top]$ is PD (which follows from \Cref{ass:invertibility-cont}, since rows of $\Phi(Z)$ are constant multiples of the rows of $\Phi^*(Z)$), $$\boldsymbol{\alpha}=(\mathbb{E}[\Phi(Z)\Phi(Z)^\top])^{-1}\mathbb{E}[\Phi(Z)A].$$
Thus, we have proved that the parameters $\alpha_u$ are uniquely recovered (up to permutation).
By the law of total variance:
$\text{Var}(A|Z) = \mathbb{E}_{U|Z}[\text{Var}(A|Z, U)] + \text{Var}_{U|Z}(\mathbb{E}[A|Z, U]).$
For our model, this expands to:
  \begin{equation}
 \text{Var}(A|Z) = \sum_{u=1}^K w_u(Z)\sigma_u^2 + \underbrace{\sum_{u=1}^K w_u(Z)(\alpha_u^\top \phi(Z))^2 - \left(\sum_{u=1}^K w_u(Z)\cdot\alpha_u^\top \phi(Z)\right)^2}_{=: C(Z)}.
\end{equation}
 We are left with the linear equation for the unknown variances:
$W(Z)^\top {\boldsymbol{\sigma}} = \text{Var}(A|Z) - C(Z),$
where $W(z)=[w_1(z),\cdots,w_K(z)]^\top,{\boldsymbol{\sigma}}=[\sigma_1^2,\cdots,\sigma_K^2]^\top.$ Taking expectation on both sides, we get $\mathbb{E}[W(Z)W(Z)^\top ]{\boldsymbol{\sigma}} =  \mathbb E[W(Z)(\text{Var}(A\mid Z)-C(Z))]$.
Since $ \mathbb{E} \left[ W(Z) W(Z)^\top \right]$ is positive definite  (using  \Cref{ass:invertibility-cont}, since rows of $W(Z)$ are constant multiples of the rows of $W^*(Z)$), we get the unique solution $${\boldsymbol{\sigma}}=(\mathbb{E}[W(Z)W(Z)^\top ])^{-1}\mathbb E[W(Z)(\text{Var}(A\mid Z)-C(Z))].$$ 
With $\alpha_u, \sigma_u$ identified, we can identify (up to permutation) the conditional density of $A$ given $U,Z$:
$$e_u(a,z) :=  h(a;\alpha_u^\top \phi(z),\sigma^2_u).$$
\textbf{Step 3.} 
We now construct the correct weights for the outcome regression. We need $\tilde{w}_u(a, z) := P(U=u \mid A=a, Z=z)$. By Bayes' rule:$$\tilde{w}_u(a, z) = \frac{f(a\mid U=u, Z=z) P(U=u \mid Z=z)}{\sum_{k} f(a\mid U=k, Z=z) P(U=k \mid Z=z)}.$$
Substituting the identified quantities from Stages 1 and 2:
\begin{equation}
\label{eq:weight-cont}
    \tilde{w}_u(a, z) = \frac{e_u(a,z) w_u(z)}{\sum_{k} e_k(a,z) w_k(z)}.
\end{equation}
So, $\tilde{w}_u(a, z)$ is now identifiable (up to permutation).

\textbf{Step 4.} Finally, we identify the outcome model parameters. We observe $$\mathbb{E}[Y \mid A=a, Z=z] = \sum_{u=1}^K P(U=u \mid A=a, Z=z)\cdot\mathbb{E}[Y \mid A=a, Z=z, U=u].$$
Substituting the outcome model in \Cref{ass:outcome-model-cont}:$$\mathbb{E}[Y \mid A=a, Z=z] = \sum_{u=1}^K \tilde{w}_u(a, z) (\beta_u^\top \psi(a, z)).$$
We regress $Y$ on the features weighted by the \textit{updated} posteriors $\tilde{w}_u(a, z)$, as explained below.

Define the vectors $\boldsymbol{\beta} \in \mathbb{R}^{K \cdot M}$ and $\Psi(A, Z) \in \mathbb{R}^{K \cdot M}$ as:
\begin{equation}
\label{eq:params-2-cont}
    \boldsymbol{\beta} = \begin{bmatrix} \beta_1 \\ \vdots \\ \beta_K \end{bmatrix}, \quad
\Psi(A, Z) = \begin{bmatrix} \tilde{w}_1(A, Z) \psi(A, Z) \\ \vdots \\ \tilde{w}_K(A, Z) \psi(A, Z) \end{bmatrix}.
\end{equation}
The mixture equation simplifies to a standard linear model structure: $\mathbb{E}[Y \mid A, Z] = \boldsymbol{\beta}^\top \Psi(A, Z).$
We have, $ \mathbb{E}[\Psi(A, Z)Y]=\mathbb{E}[\Psi(A, Z)\mathbb{E}[Y\mid A,Z]]=\mathbb{E}[\Psi(A, Z)\Psi(A, Z)^\top] \boldsymbol{\beta}.$
The explicit closed-form solution is given by:
$$\boldsymbol{\beta} = \left( \mathbb{E} \left[ \Psi(A,Z) \Psi(A,Z)^\top \right] \right)^{-1} \mathbb{E}\left[ \Psi(A,Z) Y \right].$$
This solution exists and is unique because $ \mathbb{E}\left[ \Psi(A,Z) \Psi(A,Z)^\top \right]$ is positive definite (using  \Cref{ass:invertibility-cont}, since rows of $\Psi(A,Z)$ are constant multiples of the rows of $\Psi^*(A,Z)$).
Thus, $\{\beta_u\}_u$ is uniquely identified (up to permutation).

\textbf{Step 5.} Now we conclude how the causal effects can be identified. 
Since the set $(U,Z)$ is the parent set for $A$ in our DAG, it works as a valid adjustment set for the average treatment effect (ATE). Therefore,
\begin{align*}
    \tau(a)= \mathbb E^{do(A=a)}[Y]
    &= \sum_{u=1}^K\pi_u \int \mathbb E[Y \mid A=a,  Z=z, U=u]~d\mathbb P_{1,u}(z_1)d\mathbb P_{2,u}(z_2)d\mathbb P_{3,u}(z_3)\\
    &=\sum_{u=1}^K\pi_u \int \beta_u^\top\psi(a,z)~d\mathbb P_{1,u}(z_1)d\mathbb P_{2,u}(z_2)d\mathbb P_{3,u}(z_3). 
\end{align*}
Since $\{\pi_u, \mathbb P_{1,u}, \mathbb P_{2,u}, \mathbb P_{3,u}\}_{u=1}^K$ and $\{\beta_u\}_{u=1}^K$ are uniquely identified up to permutations, as shown in stages 1 and 4 respectively, $\tau(a)$ is uniquely identifiable.
However, the \emph{conditional average treatment effect}:
\begin{align*}
   \tau_u(a,z)&= \mathbb E^{do(A=a)}[Y \mid U=u,Z=z]=\mathbb E[Y \mid A=a, U=u,Z=z]=\beta_u^\top\psi(a,z)
\end{align*}
   is uniquely identified only up to permutations.
\end{proof}


\begin{remark}
Note that the feature vectors $\phi(z)$ and $\psi(a,z)$ may be infinite-dimensional, provided that the corresponding coefficient vectors are sparse so that the effective feature representations remain finite-dimensional and satisfy the above assumptions. Importantly, identifiability does not require prior knowledge of this effective feature space. We further generalize these assumptions in \Cref{sec:gen-identifiability}.
\end{remark}
\section{Tensor-decomposition-based causal estimation for \\ multi-proxy setting}
\label{sec:alg}
Here is the formal methodological description of the tensor-based CATE estimation framework. We propose a three-stage estimation procedure that leverages multi-view spectral learning to recover latent confounding structures, followed by regressions to estimate the parameters within each cluster. We observe i.i.d.\ samples $\{(Z_{1}^{(i)},Z_{2}^{(i)},Z_{3}^{(i)},A^{(i)},Y^{(i)})\}_{i=1}^n$, generated according to the directed acyclic graph (DAG) in \Cref{fig:our_dag},  where $Z^{(i)}_1, Z^{(i)}_2, Z^{(i)}_3 \in \mathbb{R}^d$ are three conditionally independent proxies, $A^{(i)} \in\mathbb{R}$ is the treatment assignment, $Y^{(i)} \in \mathbb{R}$ is the outcome.\\\\
\textbf{Step 1: Unsupervised Tensor Learning.} 
\begin{itemize}

    \item Robust Tensor Decomposition: Apply the method of moment/robust tensor power method \citep{anandkumar2014tensor,song2014nonparametric}, as described in \Cref{subsec:tensor-method}, to extract the estimated factor densities $\hat{f}_{v,u}$ and mixing weights $\hat{\pi}_u$.

    \item Compute posterior weights
    : $\hat{w}_{u}(Z^{(i)})=\frac{\hat\pi_u \prod_{v=1}^3 \hat f_{v,u}(Z_v^{(i)})}{\sum_{k=1}^K \hat\pi_k \prod_{v=1}^3 \hat f_{v,k}(Z_v^{(i)})}.$
\end{itemize}
\textbf{Step 2: Treatment model estimation.}
\begin{itemize}
    \item First, construct the estimated feature vector
    \begin{equation*}
        \hat\Phi(Z^{(i)}) = [ \hat w_1(Z^{(i)}) \phi(Z^{(i)})^\top, \cdots, \hat w_K(Z^{(i)}) \phi(Z^{(i)})^\top]^\top.
    \end{equation*}
    Then, regress $A$ on $\hat\Phi(Z)$ to compute: $$\boldsymbol{\hat\alpha}= \left(\sum_{i=1}^n\hat\Phi(Z^{(i)})\hat\Phi(Z^{(i)})^\top\right)^{-1}\left(\sum_{i=1}^n\hat\Phi(Z^{(i)})A^{(i)}\right).$$
    \item Define $\hat C(Z^{(i)})=\sum_{u=1}^K \hat w_u(Z^{(i)})(\hat\alpha_u^\top \phi(Z^{(i)}))^2 - \left(\sum_{u=1}^K\hat w_u(Z^{(i)})\cdot\hat\alpha_u^\top \phi(Z^{(i)})\right)^2$ and\\
    $\hat W(Z^{(i)})=[\hat w_1(Z^{(i)}),\cdots,\hat w_K(Z^{(i)})]^\top$ and compute:$$\hat{\boldsymbol{\sigma}}= \left(\sum_{i=1}^n \hat{{W}}(Z^{(i)}) \hat{{W}}(Z^{(i)})^\top\right)^{-1}\left(\sum_{i=1}^n \hat{{W}}(Z^{(i)})\hat C(Z^{(i)})\right).$$ 
\end{itemize}
\textbf{Step 3: Causal Outcome Estimation.}
\begin{itemize}
    \item Update Posterior Weights: Compute the weights $\hat{\tilde{w}}_{u}$ that account for the treatment assignment:
  $$\hat{\tilde{w}}_{u}(A^{(i)},Z^{(i)})=\frac{\hat e_u(A^{(i)},Z^{(i)}) \hat w_u(Z^{(i)})}{\sum_{k} \hat e_k(A^{(i)},Z^{(i)}) \hat w_k(Z^{(i)})},$$
  where $\hat e_u(a,z)=h(a;\hat\alpha_u^\top \phi(z),\hat\sigma^2_u).$
    \item Weighted Outcome Regression: Construct:$$\hat{\Psi}(A^{(i)}, Z^{(i)}) = \Big[ \hat{\tilde{w}}_{1}(A^{(i)},Z^{(i)}) \psi(A^{(i)}, Z^{(i)})^\top, \dots, \hat{\tilde{w}}_{K}(A^{(i)},Z^{(i)}) \psi(A^{(i)}, Z^{(i)})^\top \Big]^\top.$$  
    Estimate the causal parameters $\hat{\boldsymbol{\beta}} = [\hat{\beta}_1^\top, \dots, \hat{\beta}_K^\top]^\top$ by regressing $Y$ on above:$$\hat{\boldsymbol{\beta}} = \left( \sum_{i=1}^n \hat{\Psi}(A^{(i)}, Z^{(i)}) \hat{\Psi}(A^{(i)}, Z^{(i)})^\top \right)^{-1} \sum_{i=1}^n \hat{\Psi}(A^{(i)}, Z^{(i)}) Y^{(i)}.$$
    \item The ATE estimate is $$\hat\tau(a)=\sum_{u=1}^K\hat\pi_u \int \hat\beta_u^\top\psi(a,z)~\hat f_{1,u}(z_1)\hat f_{2,u}(z_2)\hat f_{3,u}(z_3)dz_1dz_2dz_3.$$
And the CATE estimates are $\{ \hat\beta_u^\top\psi(a,z)\}_{u=1}^K$.
\end{itemize}

\begin{remark}
    One can also employ an alternative mixture learning method (e.g., EM algorithm) to compute the posteriors $\hat w_u$, followed by employing other regression models (e.g., generalized linear model). We do not provide global convergence guarantees for those settings (that might require a different set of assumptions), and we leave it for future work. 
\end{remark}

\begin{remark}
We assume the number of latent clusters  $K$ to be known in theory. However, in practice, if $K$ is unknown, one can perform singular value decomposition (SVD) on the empirical second-order embedding, $\hat{C}_{Z_1 Z_2}$ to determine $K$. In the population limit, this matrix has exactly $K$ non-zero singular values; in finite samples, $K$ can be selected based on the spectral gap (the point where singular values drop sharply).
\end{remark}

\section{A generalization of the multi-proxy identifiability}
\label{sec:gen-identifiability}
In this section, we generalize the identifiability guarantee in \Cref{sec:identifiability} to a broader class of treatment and outcome models, without imposing any specific parametric assumptions. 

\begin{assumption}[Treatment model]
\label{ass:treatment-model-rkhs}
The conditional density of treatment $A$ given $(Z, U=u)$, denoted by $e_u(A, Z)$, is such that for each  $u \in \{1, \dots, K\}$, the function $(a,z)\mapsto e_u(a, z)$ lies in a function class $\mathcal{H}_1$ such that for any set of functions $\{e_u\}_{u=1}^K\in \mathcal{H}_1^K$ and $\{e_u^\prime\}_{u=1}^K\in \mathcal{H}_1^K$,
\[ \sum_{u=1}^K \pi_u{f}_u(Z)  e_u(A,Z)= \sum_{u=1}^K \pi_u{f}_u(Z)  e_u^\prime(A,Z) \text{ almost surely }\implies e_u\equiv e_u^\prime, \forall u.\]
\end{assumption}
\begin{assumption}[Outcome model]
\label{ass:outcome-model-rkhs} We assume for each  $u \in \{1, \dots, K\}$, there exists a function $(a,z)\mapsto g_u(a, z)$ that lies in a function class $\mathcal{H}_2$ such that
$$\mathbb{E}[Y \mid A=a, Z=z, U=u] = g_u(a, z),$$
and for any set of functions $\{g_u\}_{u=1}^K\in \mathcal{H}_2^K$ and $\{g_u^\prime\}_{u=1}^K\in \mathcal{H}_2^K$,
\[ \sum_{u=1}^K \pi_u{f}_u(Z) e_u(A, Z) g_u(A,Z) = \sum_{u=1}^K \pi_u{f}_u(Z)e_u(A, Z) g_u^\prime(A,Z)  \text{ almost surely }\implies g_u\equiv g_u^\prime, \forall u.\]
\end{assumption}

Indeed, \Cref{ass:treatment-model-cont,ass:outcome-model-cont,ass:invertibility-cont} arise as special cases of the above assumptions. 
Next, we give an example that satisfies these more general conditions.
\paragraph{One concrete example.} Consider the same proxy and treatment models as in \Cref{concrete-example} along with the following outcome model class
\[\mathcal{H}_2=\{g:g(a,z)=\sum_{j=0}^\infty\beta_{0,j}a^j+\sum_{i=0}^3\sum_{j=1}^\infty\beta_{i,j}z_i^j; \beta_{i,j}\in\mathbb{R}~ \forall i,j, \beta_{0}=[\beta_{0,1},\beta_{0,2},\cdots]^T \text{ is s-sparse}\},\]
for any fixed $s\in\mathbb{N}$. 
\begin{proposition}
\label{prop:example-gen}
The above model satisfies \Cref{ass:treatment-model-rkhs,ass:outcome-model-rkhs}.
\end{proposition}

We now state the identifiability guarantee under the new assumptions.

\begin{theorem}
\label{thm:identifiability-gen}
Under \Cref{ass:indep,ass:lin-ind-z,ass:positivity-1,ass:treatment-model-rkhs,ass:outcome-model-rkhs} the average treatment effect, 
$\tau(a)=\mathbb E^{do(A=a)}[Y]$ is uniquely identifiable, and the conditional average treatment effect, $\tau_u(a,z)=\mathbb E^{do(A=a)}[Y\mid U=u,Z=z]$ is uniquely identifiable up to permutations.
\end{theorem}

\section{Multi-treatment Causal Inference Framework}
\label{sec:multi-treatment}

The latent variable $U$ is \emph{unobserved}, categorical, and takes values in $\{1,\dots,K\}$, as before.
In this setup, there are no proxy variables; rather, we observe multiple (at least three) treatments $A$ and outcome $Y$, according to the directed acyclic graph (DAG) in \Cref{fig:multi-tr-dag}. For simplicity, we assume only three treatments, $A=(A_1^\top,A_2^\top,A_3^\top)^\top$. As before, our goal is to estimate the average treatment effect (ATE), 
$\tau(a):=\mathbb E^{do(A=a)}[Y]$, as well as the conditional average treatment effect (CATE), $\tau_u(a)=\mathbb E^{do(A=a)}[Y\mid U=u].$
  We need the following assumptions for the identifiability of the causal effect in our model.
\begin{assumption}[Latent Conditional Independence]
\label{ass:indep-tr}
The treatment variables $A_v\in\mathbb R^d$, for $v=1,2,3$ are conditionally independent given the latent state $U$: $A_1 \perp A_2 \perp A_3 \mid U.$\end{assumption}
The conditional independence assumption is also present in prior works in the multi-treatment setting \citep{wang2019blessings}.
\begin{assumption}[Linear independence of measures]
\label{ass:lin-ind-tr}
Let $\mathbb Q_{v,u}$ be the conditional distribution of $A_v,$ for $v \in \{1, 2, 3\}$ given the latent state $U=u$. Then, we assume that the measures $\{\mathbb Q_{v,u}\}_{1 \le u\le K}$ are linearly independent for each $v=1,2,3.$
  \end{assumption}
\Cref{ass:indep-tr,ass:lin-ind-tr} ensure identifiability and the uniqueness of the tensor factorization of the treatments \cite{allman2009identifiability,song2014nonparametric}. It says that the clusters are mathematically distinguishable. 
\begin{assumption}[Outcome model]
\label{ass:outcome-model-tr}
The outcome model is parameterized by $[\gamma_1^\top, \cdots, \gamma_K^\top]^\top\in\mathbb{R}^{M\times K}$: 
$$\mathbb{E}[Y \mid A=a, U=u] = \gamma_u^\top \xi(a),$$ where $\xi=(\xi_1,\cdots,\xi_{M})^\top :\mathbb{R}^{3d}\to\mathbb R^M$ is a known function. We use the notation $A=(A_1^\top,A_2^\top,A_3^\top)^\top.$
  \end{assumption}
Define the functions $r_u(a)=\prod_{v=1}^3 g_{v,u}(a_v)$, where  $g_{v,u}$ is the density of $\mathbb Q_{v,u}$.
We also define the vector:
  \begin{equation}
\Xi^*(a) = [r_1(a) \xi(a)^\top, \cdots, r_K(a) \xi(a)^\top]^\top.
\end{equation}

\begin{assumption}
\label{ass:invertibility-tr}
The matrix $\mathbb{E}[\Xi^*(A)\Xi^*(A)^\top]$ is positive definite (PD).
\end{assumption}
Then, we have the following identifiability guarantee.
\begin{theorem}
\label{thm:identifiability-no-proxy}
Under \Cref{ass:indep-tr,ass:lin-ind-tr,ass:positivity-1,ass:outcome-model-tr,ass:invertibility-tr}, the outcome model parameters $\{\gamma_u\}_{u=1}^K$ are uniquely identifiable (up to permutation) from the joint distribution of $(A, Y)$. Moreover, the average treatment effect, 
$\tau(a)=\mathbb E^{do(A=a)}[Y]$ is uniquely identifiable, and the conditional average treatment effect, $\tau_u(a)=\mathbb E^{do(A=a)}[Y\mid U=u]$ is uniquely identifiable up to permutations.
\end{theorem}

The proof follows similar arguments as in the proof of \Cref{thm:identifiability} and is provided in the appendix. 
Moreover, a tensor decomposition–based estimation procedure, analogous to \Cref{sec:alg} for the multi-proxy setting, can be employed for the multi-treatment setting, as summarized in \Cref{alg:multi_treatment_estimation}.

\begin{algorithm}[ht]
\caption{Estimation of Causal Effects via Multi-Treatment Tensor Decomposition}
\label{alg:multi_treatment_estimation}
\begin{algorithmic}[1]
\Require Dataset $\mathcal{D} = \{(A_{1i}, A_{2i}, A_{3i}, Y_i)\}_{i=1}^n$, number of latent states $K$, outcome feature map $\xi(\cdot)$.

\State Apply tensor decomposition method \cite{song2014nonparametric,anandkumar2014tensor} to estimate the latent prior probabilities $\hat{\pi}_u$ and densities $\hat{g}_{v,u}(a) = \widehat{\mathbb{P}}(A_v = a \mid U=u)$.
\State Compute the empirical posterior membership probabilities for each sample $i \in [n]$:
  \[
    \hat{w}_u(A_i) = \frac{\hat{\pi}_u \prod_{v=1}^3 \hat{g}_{v,u}(A_{vi})}{\sum_{k=1}^K \hat{\pi}_k \prod_{v=1}^3 \hat{g}_{v,k}(A_{vi})}, \quad \forall u \in [K].
    \]

\State Construct the composite feature vector $\widehat{\Xi}_i$ for each sample $i$:
  \[
    \widehat{\Xi}_i = \begin{bmatrix}
    \hat{w}_1(A_i) \xi(A_i) \\
    \vdots \\
    \hat{w}_K(A_i) \xi(A_i)
    \end{bmatrix}.
    \]
\State Estimate the parameters $\hat{\bm{\gamma}} = [\hat{\gamma}_1^\top, \dots, \hat{\gamma}_K^\top]^\top$ via ordinary least squares:
  \[
    \hat{\bm{\gamma}} = \left( \sum_{i=1}^n \widehat{\Xi}_i \widehat{\Xi}_i^\top \right)^{-1} \sum_{i=1}^n \widehat{\Xi}_i Y_i.
    \]

\State Estimate the CATE for latent state $u$: $\hat{\tau}_u(a) = \hat{\gamma}_u^\top \xi(a).$
  \State Estimate the ATE for target treatment combination $a$: $\hat{\tau}(a) = \sum_{u=1}^K \hat{\pi}_u \hat{\gamma}_u^\top \xi(a).$
  \end{algorithmic}
\end{algorithm}

\section{Non-asymptotic guarantees for the model parameters}
\label{sec:theory}
In this section, we provide bounds on the estimation errors of ATE and the model parameters for the estimation procedure described in \Cref{sec:alg}.
The estimation error for the causal parameter $\boldsymbol{\beta}$ accumulates through the sequential stages of the algorithm. To rigorously quantify this error propagation, we decompose our finite-sample analysis into distinct steps corresponding to each estimation phase. We first introduce some notations.

\paragraph{Notations.} The Euclidean norm of a vector $x$ is denoted by $\|x\|$. The induced spectral norm of a matrix $A$ is
denoted by $\|A\|$, i.e. $\|A\| := \sup\{\|Ax\|: \|x\| = 1\}.$ For any symmetric and positive semidefinite matrix $M$, let $\|x\|_M$ denote the norm of a vector $x$ defined by
\(\|x\|_M=\sqrt{x^\top M x}\) and $\lambda_{min}(M)$ denote the minimum eigenvalue of $M$.

We will show that under regularity conditions,
 for any target treatment $a$, and any $\delta \in (0,1)$, for sufficiently large $n$, with probability at least $1 - \delta$, the absolute error is bounded by:
\[
|\hat{\tau}(a) - \tau(a)| \le \mathcal{O}\left( \log(1/\delta)\cdot n^{-\frac{b}{2b+7d}} \right),
\]
where $b$ is some smoothness parameter for the density of $Z$.
\subsection{Nonparametric estimation error for weights $w$}

For clarity of exposition, we focus on the symmetric view setting as done in  \cite{song2014nonparametric}, that is, we assume the conditional feature distributions are identical across the three views: $f_{1,u} = f_{2,u} = f_{3,u}$ for all latent components $u \in \{1, \dots, K\}$. This assumption is made for notational and presentational convenience; our approach extends directly to the general asymmetric case in which the view-specific distributions may differ.

\begin{definition}
    Let \( \mathcal{X} \subset \mathbb{R}^d \) be a compact space. For any multi-index \( r = (r_1, \ldots, r_d) \), with \( r_i \in \mathbb{N} \), define \( |r| = \sum_i r_i \), and let 
$D^r = \frac{\partial^{|r|}}{\partial x_1^{r_1} \cdots \partial x_d^{r_d}}.$
The Hölder class \( \mathcal{H}(b, l) \) is the set of functions \( f \in L^2(\mathcal{X}) \) satisfying
$|D^r f(x) - D^r f(y)| \leq l \|x - y\|^{b - |r|}$
for all multi-indices \( r \) such that \( |r| \leq \lfloor b \rfloor \), and for all \( x, y \in \mathcal{X} \).
Moreover, define the bounded Hölder class \( \mathcal{H}(b, l, m_0, m_1) \) to be
$\left\{ f \in \mathcal{H}(b, l) : m_0 < f < m_1 \right\}$.
\end{definition}

Assuming density to be in the bounded Hölder class \( \mathcal{H}(b, l, m_0, m_1) \) is standard in nonparametric density functional estimation.

\begin{theorem}
\label{thm:posterior-1-error}
Let $f_{v,u}\in \mathcal{H}(b, l, m_0, m_1)$, for all $v,u$ and fix $\delta\in(0,1)$. Consider a bounded kernel $k(.,.)$  of order $\floor{b}$, with bandwidth $s=cn^{-\frac{1}{4+7d}}$ and supremum value $\rho_0=c^{\prime}s^{-d}$, for some constant $c,c^{\prime}$, satisfying $\int |x-y|^b |k(x, y)| dy < \infty$, for any $x$. Let $N_u(z) = \pi_u \prod_{v=1}^3 f_{v,u}(z_v)$ and $\hat{N}_u(z) = \hat\pi_u \prod_{v=1}^3 \hat f_{v,u}(z_v)$. Then, under the assumptions of \Cref{thm:song-et-al}, there exists a permutation $\sigma$ such that for large enough $n$,
with probability $1-\delta$,
\begin{align*}
     \max_{u=1,\cdots,K}\sup_z| \hat{ w}_{\sigma(u)}(z) - w_{u}(z) |&\leq   \frac{2K}{m_0^3}\times\left[\max_{u} \sup_z|\hat{N}_{\sigma(u)}(z) - N_u(z) |\right]\\
     &\leq\Bigg(\frac{18Km_1^2(c^\prime )^{2.5}(8(cn^{-\frac{1}{2b+7d}})^{d}+5c^{\prime})\sqrt{\log \frac{\delta}{8}}}{m_0^3c^{3.5d}\sigma_K^{1.5}(\mathcal{C}_{Z_1 Z_2})} +c^{\prime\prime}\Bigg)n^{-\frac{b}{2b+7d}},
\end{align*}
for some constant $c^{\prime\prime}>0$, depending only on the kernel and $b, l$.
\end{theorem}

From now on, \textbf{for clarity of exposition, we assume $\sigma$ to be the identity permutation.}
\subsubsection{Error bound for estimating $\alpha$}
We need the following assumptions:
\begin{enumerate}
    \item  With probability $1 - \delta/4$, $\max_{i=1,\cdots,n}\max_{u=1,\cdots,K}|\hat{w}_u(Z^{(i)}) - w_u(Z^{(i)})| \le \ew $. (An explicit expression for $\epsilon_{n,4\delta,w}$ is derived in \Cref{thm:posterior-1-error}.)
    \item The vector $\Sigma_\Phi^{-1}\Phi(Z)$ is sub-Gaussian with variance proxy $\rho_\phi^2$.
    \item
 $A$ and $\|\phi(Z) \|$ are sub-Gaussians with proxy variance $\sigma^2_A$ and $\sigma^2_\phi$ 
 respectively. 
    \item The noise $\epsilon_1(Z):=A-\boldsymbol{\alpha}^\top\Phi(Z)$  is conditionally sub-Gaussian, i.e., for all $\lambda\in\mathbb{R}$, almost
surely:  $\mathbb{E}(\exp(\lambda\epsilon_1(Z))\mid Z)\leq\exp\left(\frac{\lambda^2\sigma^2_{1,noise}}{2}\right).$
\end{enumerate}
We have already derived an expression for $\ew$ in the previous subsection. Now we shall derive an error bound for $\hat\alpha$ in terms of $\ew$.
To state our result, we first define the following for all $\delta\in(0,1)$:
\begin{align*}
    &C_{A,n}:=\E\left(A^2\right)+ \max \left( \sqrt{\frac{32\sigma^2_A \log(4/\delta)}{n}}, \frac{32\sigma^2_A \log(4/\delta)}{n} \right),\\
    &C_{\phi,n}:= \E\left(\| \phi(Z) \|^2\right)+ \max \Bigg( \sqrt{\frac{32\sigma^2_\phi \log(4/\delta)}{n}}, \frac{32\sigma^2_\phi \log(4/\delta)}{n} \Bigg).
\end{align*}

\begin{lemma}
\label{thm:alpha-error-bound}
Let $\delta \in (0,1)$. Under  \Cref{ass:treatment-model-cont,ass:outcome-model-cont,ass:invertibility-cont} and the above assumptions, for large enough $n$, with probability at least $1-2\delta$,
    \begin{align*}
       \|{\boldsymbol{\alpha}} - \hat{\boldsymbol{\alpha}}\|\leq & \frac{1}{2(\lam-\delta_{G,n})} {\left[1+ \frac{3  K C_{\phi,n}}{\lam-\delta_{G,n}- 3 \ew K C_{\phi,n}} \right]} (C_{\phi,n}+C_{A,n})\ew\\&+\frac{\sigma_{1,noise}}{\sqrt{n(1-\eta_{1,n,\delta})\lambda_\Phi}}\sqrt{LK + 2\sqrt{LK \log(3/\delta)}+2\log(3/\delta)+\frac{2\log(3/\delta)}{1-\eta_{1,n,\delta}}},
    \end{align*}
    where $\lambda_0=\lambda_{\min}(\Sigma_\phi)$, ${\Sigma_\Phi}=\mathbb{E}[{\Phi}(Z) {\Phi}(Z)^\top]$, ${\mu_\Phi}=\mathbb{E}[{\Phi}(Z)]$,  $\delta_{G,n}=\frac{LK}{n}+(1+ 8\rho_\Phi \|\mu_\Phi\|_2)\sqrt{\frac{LK}{n}}+\max\left( \sqrt{\frac{\ln(8c_1/\delta)}{c_2 n}}, \frac{\ln(8c_1/\delta)}{c_2 n} \right)+8\rho_\Phi \|\mu_\Phi\|_2 \cdot \sqrt{\frac{ \log(8/\delta)}{n}}$, and $ \eta_{1,n,\delta}= \frac{LK}{n}+(1+ 8 \rho_\Phi \| \Sigma^{-1/2}_\Phi\mu_\Phi\|_2)\sqrt{\frac{LK}{n}}$\\
    $+\max\left( \sqrt{\frac{\ln(3c_1/\delta)}{c_2 n}}, \frac{\ln(3c_1/\delta)}{c_2 n} \right)+8\rho_\Phi \| \Sigma^{-1/2}_\Phi\mu_\Phi\|_2 \sqrt{\frac{ \log(3/\delta)}{n}}$, for some universal constants $c_1,c_2$.
\end{lemma}

\subsection{Error bound for estimating $\sigma$}

We need the following assumptions:
\begin{enumerate}
    \item  With prob. $1 - \delta/4$, $\max_{i=1,\cdots,n}\max_{u=1,\cdots,K}|\hat{w}_u(Z^{(i)}) - w_u(Z^{(i)})| \le \ew $
    \item  With prob. $1 - \delta/4$,  $ \|{\boldsymbol{\alpha}} - \hat{\boldsymbol{\alpha}}\|\leq \ea$
    \item
 $\|\phi(Z) \|$ is sub-Gaussian with proxy variance $\sigma^2_\phi$. 
 \item The noise $\epsilon_2(Z):=C(Z)-\boldsymbol{\sigma}^\top W(Z)$ is conditionally subgaussian, i.e., for all $\lambda\in\mathbb{R}$, almost
surely: $\mathbb{E}(\exp(\lambda\epsilon_2(Z))\mid Z)\leq\exp\left(\frac{\lambda^2\sigma^2_{2,noise}}{2}\right).$
\end{enumerate}
We have already derived an expression for $\ew$ and $\ea$ in the previous subsections. Now we shall derive an error bound for $\hsigma$ in terms of $\ew$ and $\ea$.

\begin{lemma}
\label{thm:sigma-est-error}
Let $\delta \in (0,1)$. Under \Cref{ass:treatment-model-cont,ass:outcome-model-cont,ass:invertibility-cont} and the above assumptions, for large enough $n$,  with probability at least $1-2\delta$,
\begin{align}
 \nonumber  | \hsigma - \boldsymbol{\sigma^2}| \le&\frac{\sqrt{K}C_{\phi,n}}{\lambda_W-\delta_{W,n}}\left( B_\alpha^2 (3K+2K^2+1)\ew  +   (3K+2+B_\alpha+2\ea)\ea\right)\\ 
 \nonumber&+\frac{3 K^{3/2}\left( B_\alpha^2 (3K+2K^2)\ew  +   (3K+2+B_\alpha+2\ea)\ea+B_\alpha\right)C_{\phi,n}}{(\lambda_W-\delta_{W,n})(\lambda_W-\delta_{W,n}- 3K \ew)}\cdot\ew\\
 &+\frac{\sigma_{2,noise}}{\sqrt{n(1-\eta_{2,n,\delta})\lambda_W}}\sqrt{K + 2\sqrt{K \log(3/\delta)}+2\log(3/\delta)+\frac{2\log(3/\delta)}{1-\eta_{2,n,\delta}}},
\end{align}    
where $\delta_W=\lambda_{\min}(\Sigma_W)$, ${\Sigma_W}=\mathbb{E}[{W}(Z) {W}(Z)^\top]$, $B_\alpha=\max_{k=1,\cdots,K}\|{\alpha}_k\|$, $ \eta_{2,n,\delta}= (1+ 4\sqrt{K}  \| \Sigma^{-1/2}_W\mu_W\|_2)\sqrt{\frac{K}{n}}$ $+\max\left( \sqrt{\frac{\ln(3c_3/\delta)}{c_4 n}}, \frac{\ln(3c_3/\delta)}{c_4n} \right)+4\sqrt{K} \| \Sigma^{-1/2}_W\mu_W\|_2 \sqrt{\frac{ \log(3/\delta)}{n}}+\frac{K}{n},$ and  $\delta_{W,n}=\frac{d}{n}+\sqrt{\frac{d}{n}}+\max\left( \sqrt{\frac{\ln(4c_3/\delta)}{c_4 n}}, \frac{\ln(4c_3/\delta)}{c_4 n} \right)$, for some universal constants $c_3,c_4$.
\end{lemma}

\subsection{Error in estimating $\tilde{w}$}
We assume the following.
\begin{enumerate}
  \item  With prob. $1 - \delta/4$, $\displaystyle\max_{i=1,\cdots,n}\max_{u=1,\cdots,K}|\hat{w}_u(Z^{(i)}) - w_u(Z^{(i)})| \le \ew $.
 \item  With prob. $1 - \delta/4$,  $ \|{\boldsymbol{\alpha}} - \hat{\boldsymbol{\alpha}}\|\leq \ea$, and $\ea\to 0$, as $n\to\infty$.
      \item  With prob. $1 - \delta/4$,  $ \|{\boldsymbol{\sigma}^2} - \hat{\boldsymbol{\sigma}^2}\|\leq \es$, and $\es\to 0$, as $n\to\infty$.
    \item
 $A$ and $\|\phi(Z) \|$ are sub-Gaussians with proxy variance $\sigma^2_A$ and $\sigma^2_\phi$ 
 respectively. 
\end{enumerate}
We have already derived an expression for $\ew$, $\ea$, and $\es$ in the previous subsections. Now we shall derive an error bound for $\hat{\tilde w}$ in terms of $\ew$, $\ea$ and $\es$. We first
define the constants:
\[M_A=\mathbb{E}(|A|)+ \sigma_A \sqrt{2 \log(16n / \delta)},\quad M_\phi=\mathbb{E}(\|\phi(Z)\|)+ \sigma_\phi \sqrt{2 \log(16n d / \delta)} .\]
\begin{lemma}
\label{thm:wtilde-esr-error}
Under the above assumptions, for large enough $n$, with probability $1-\delta$,
\begin{align*}
  &\max_{i=1,\cdots,n}\max_{u=1,\cdots,K} | \hat{\tilde w}_u((A^{(i)},Z^{(i)}) -\tilde w_u((A^{(i)},Z^{(i)}) |\\
  &\leq 2K\ew+2K(1+\ew) \Bigg[\frac{\ea}{B_\sigma}M_\phi(2M_A+(\ea+2B_\alpha)M_\phi)+\frac{2(M_A^2+(\ea+B_\alpha)M_\phi^2) \es}{B_\sigma(B_\sigma-\es)}\Bigg],
\end{align*}
where $B_\alpha=\max_{k}\|{\alpha}_k\|$ and $B_\sigma=\max_{k}{\sigma}_k^2$.
\end{lemma}
\subsection{Error in estimating $\beta$}
We now derive an error bound for $\hat{\boldsymbol{\beta}}$ in terms of the error in estimating $\hat{\tilde w}$, which we have derived in the previous subsection. We need to assume that
\begin{enumerate}
    \item  With probability $1 - \delta/4$, 
$\max_{i=1,\cdots,n}\max_{u=1,\cdots,K} | \hat{\tilde w}_u((A^{(i)},Z^{(i)}) -\tilde w_u((A^{(i)},Z^{(i)}) |\leq \ewtilde$.
    \item The vector $\Sigma_\Psi^{-1}\Psi(A,Z)$ is sub-Gaussian with proxy variance $\rho_\Psi^2$.
    \item $Y$ and $\|\psi(A,Z) \|$ are sub-Gaussians with proxy variance $\sigma^2_Y$ and $\sigma^2_\psi$ respectively. 
  \item The noise $\epsilon_3(A,Z):=Y-\boldsymbol{\beta}^\top \Psi(A,Z)$ is such that for all $\lambda\in\mathbb{R}$, $\mathbb{E}(\exp(\lambda\epsilon_3(A,Z))\mid A,Z)\leq\exp\left(\frac{\lambda^2\sigma^2_{3,noise}}{2}\right)$, almost
surely.
  \end{enumerate}

Define the constants for each $\delta$:
$$C_{Y,n}= \E\left(\| Y \|^2\right)+ \max \left( \sqrt{\frac{32\sigma^2_Y \log(4/\delta)}{n}}, \frac{32\sigma^2_Y \log(4/\delta)}{n} \right),$$
$$C_{\psi,n}= \E\left(\| \psi(A,Z) \|^2\right)+ \max \left( \sqrt{\frac{32\sigma^2_\psi\log(4/\delta)}{n}}, \frac{32\sigma^2_\psi \log(4/\delta)}{n} \right).$$

\begin{theorem}
\label{thm:beta-error}
    Let $\delta \in (0,1)$. 
Under \Cref{ass:treatment-model-cont,ass:outcome-model-cont,ass:invertibility-cont} and the above assumptions, for large enough $n$, with probability at least $1-2\delta$,
    \begin{align*}
       \|{\boldsymbol{\beta}} - \hat{\boldsymbol{\beta
    }}\|\leq & \frac{\sqrt{K}}{2(\lambda_\Psi-\delta_{H,n})} {\left[1+ \frac{3  K C_{\psi,n}}{\lambda_\Psi-\delta_{H,n}- 3 \ewtilde K C_{\psi,n}} \right]} (C_{\psi,n}+C_{Y,n})\ewtilde\\&+\frac{\sigma_{3,noise}}{\sqrt{n(1-\eta_{n,\delta})\lambda_\Psi}}\sqrt{MK + 2\sqrt{MK \log(3/\delta)}+2\log(3/\delta)+\frac{2\log(3/\delta)}{1-\eta_{n,\delta}}},
    \end{align*}
    where $\lambda_\Psi=\lambda_{\min}(\Sigma_\Psi)$,  ${\Sigma_\Psi}=\mathbb{E}[{\Psi}(A,Z) {\Psi}(A,Z)^\top]$, ${\mu_\Psi}=\mathbb{E}[{\Psi}(A,Z)]$,  $\delta_{H,n}=\frac{MK}{n}+(1+ 8\rho_\Psi \|\mu_\Psi\|_2)\sqrt{\frac{MK}{n}}+\max\left( \sqrt{\frac{\ln(8c_1/\delta)}{c_2 n}}, \frac{\ln(8c_1/\delta)}{c_2 n} \right)+8\rho_\Psi \|\mu_\Psi\|_2 \cdot \sqrt{\frac{ \log(8/\delta)}{n}}$, and $ \eta_{n,\delta}= \frac{MK}{n}+(1+ 8 \rho_\Psi \| \Sigma^{-1/2}_\Psi\mu_\Psi\|_2)\sqrt{\frac{MK}{n}}$\\$+\max\left( \sqrt{\frac{\ln(3c_5/\delta)}{c_6 n}}, \frac{\ln(3c_5/\delta)}{c_6 n} \right)+8\rho_\Psi \| \Sigma^{-1/2}_\Psi\mu_\Psi\|_2 \sqrt{\frac{ \log(3/\delta)}{n}}$, for some universal constants $c_5,c_6$.
\end{theorem}

\subsection{Error in estimating ATE}
We now derive an error bound for the estimation of ATE. We assume the following:
\begin{enumerate}
    \item For each fixed $a$, $\psi(a,z)$ is a continuous function of $z$. Since we already assumed that each $Z_v$ takes values in a compact space $\mathcal{X}$, there exists $M_a>0$ such that $\sup_{z\in \mathcal{X}^3}|\psi(a,z)|\leq M_a$.
    \item With prob. $1 - \delta/2$,  $ \|{\boldsymbol{\beta}} - \hat{\boldsymbol{\beta}}\|\leq \epsilon_{n,\delta,\beta}$.
    \item With prob. $1 - \delta/2$,  $\max_{u} \sup_z|\hat{N}_{\sigma(u)}(z) - N_u(z) |\leq \epsilon_{n,\delta,N}$.
\end{enumerate}
\begin{theorem}
\label{thm:ate-error}
Let $\delta \in (0,1)$ and $V$ be the finite Lebesgue volume of the compact space $\mathcal{X}^3$. 
Under the above assumptions, for large enough $n$, with probability at least $1-\delta$    \begin{align*}|\hat\tau(a)-\tau(a)| &\leq M_a \left( \epsilon_{n,\delta,\beta} + V \sqrt{K} (\|\beta\|_2 + \epsilon_{n,\delta,\beta}) \epsilon_{n,\delta,N} \right).
\end{align*}
\end{theorem}
From the results derived in the previous subsections, it follows that $ \epsilon_{n,\delta,\beta}=\mathcal{O}\left( \log(1/\delta)\cdot n^{-\frac{b}{2b+7d}} \right)$ and hence, from the above theorem, 
$|\hat{\tau}(a) - \tau(a)| \le \mathcal{O}\left( \log(1/\delta)\cdot n^{-\frac{b}{2b+7d}} \right)$, with probability at least $1-\delta$.

\textbf{Confidence interval for ATE:} Note that the above bound yields a non-asymptotic $1-\delta$ confidence interval for the $\tau(a)$:
\[\hat\tau(a)\pm M_a \left( \epsilon_{n,\delta,\beta} + V \sqrt{K} (\|\beta\|_2 + \epsilon_{n,\delta,\beta}) \epsilon_{n,\delta,N} \right).\]

\section{Experimental Results}
\label{sec:expt}

\subsection{Simulation: Multi-proxy}
We simulate $n$ i.i.d.\ samples, where each sample $i\in[n]$ consists of three proxies $Z^{(i)}_{1}, Z^{(i)}_{2}, Z^{(i)}_{3} \in \mathbb{R}^3$, treatment $A^{(i)} \in \mathbb{R}$, and outcome $Y^{(i)} \in \mathbb{R}$. The data is generated from a mixture of $K=3$ latent clusters, denoted by the latent variable $U^{(i)} \in \{1, 2, 3\}$, which is drawn from $U^{(i)} \sim \text{Categorical}(\pi), \pi = [0.33, 0.33, 0.34].$
Conditioned on $U^{(i)} = k$, the proxies are generated from multivariate Gaussian:
$$Z_v^{(i)} \mid U^{(i)} \sim \mathcal{N}\left(\mu_{ U^{(i)}}^{(v)}, \sigma^2_{ U^{(i)}}\mathbf{I}\right), \quad v \in \{1, 2, 3\},$$
where $\sigma_1=\sigma_2=\sigma_3 = 0.8$ and
\begin{align*}
    & \mu_{1}^{(1)} = [-3, 0, 0]^\top, \quad \mu_{2}^{(1)} = [0, 3, 0]^\top, \quad \mu_{3}^{(1)} = [3, -1, 2]^\top,\\
    &  \mu_{1}^{(2)} = [0, -3, 1]^\top, \quad \mu_{2}^{(2)} = [-2, 0, 3]^\top, \quad \mu_{3}^{(2)} = [3, 0, 0]^\top\\
    & \mu_{1}^{(3)} = [3, 1, -1]^\top, \quad \mu_{2}^{(3)} = [1, 0, 3]^\top, \quad \mu_{3}^{(3)} = [0, -2, 0]^\top.
\end{align*}
The treatment $A^{(i)} \in \mathbb{R}$ is generated conditionally on the proxies $Z^{(i)}$ and the latent $U^{(i)}$:
\begin{equation*}
    A^{(i)} \mid Z^{(i)}, U^{(i)}  \sim \mathcal{N}\left( {\alpha}_{U^{(i)},1} Z^{(i)}_{1,1}+{\alpha}_{U^{(i)},2} Z^{(i)}_{1,2}+{\alpha}_{U^{(i)},3}Z^{(i)}_{1,3}, \sigma_{U^{(i)}}^2 \right),
\end{equation*}
where the treatment noise is $\sigma_1^2= 0.6,\sigma_2^2=1,\sigma_3^2=0.8$ and the  parameters ${\alpha}_k=[{\alpha}_{k,1},{\alpha}_{k,2},{\alpha}_{k,3}]^\top$ are
\begin{equation*}
   {\alpha}_1 = [ 1.0 , 0.5 , -0.5]^\top, \quad
{\alpha}_2 = [-0.5 , 1.0 , 0.5]^\top, \quad
  {\alpha}_3 = [ 0.5 ,-0.5 , 1.0]^\top.
\end{equation*}
The outcome $Y^{(i)} \in \mathbb{R}$ is generated via the following linear model:
\begin{equation*}
    Y^{(i)} \mid A^{(i)},Z^{(i)}, U^{(i)} \sim \mathcal{N}\left(\beta_{U^{(i)},1} + {\beta}_{U^{(i)},2} A^{(i)} + {\beta}_{U^{(i)},3}^\top Z^{(i)}_{1,1} +{\beta}_{U^{(i)},4}^\top Z^{(i)}_{1,2} +
    {\beta}_{U^{(i)},5}^\top Z^{(i)}_{1,3}, 1\right),
\end{equation*}
where ${\beta}_k = [\beta_{k,1}, \beta_{k,2}, {\beta}_{k,3},{\beta}_{k,4},{\beta}_{k,5}]^\top$ are set as:
\begin{equation*}
    {\beta}_1 = [ 1, \mathbf{2.5}, 0.5, 0.5, 0.5]^\top, \quad
    {\beta}_2 = [ 5, \mathbf{-1}, -0.5, 0.5, -0.5]^\top, \quad
    {\beta}_3 = [2 , \mathbf{4.0}, 1.0, -1.0, 1.0].
\end{equation*}
In this case, the CATE difference is:
$\hat \tau_u(a_1,z)-\hat \tau_u(a_2,z)=\hat\beta_u^\top\psi(a_1,z)-\hat\beta_u^\top\psi(a_2,z)=\beta_{u,2}(a_1-a_2),$
which depends only on the parameter $\beta_{u,2}.$
So, the target CATE coefficients to be recovered are $\beta_{1,2} = 2.5$, $\beta_{2,2} = -1$, and $\beta_{3,2} = 4.0$.
We evaluate the proposed estimator across varying sample sizes, performing $100$ independent trials for each sample size. \Cref{fig:cate} shows the convergence of the estimated coefficient $\hat{\beta}_{u,2}$ as sample size increases.

\begin{figure}[!h]
    \centering
    \includegraphics[width=0.85\linewidth]{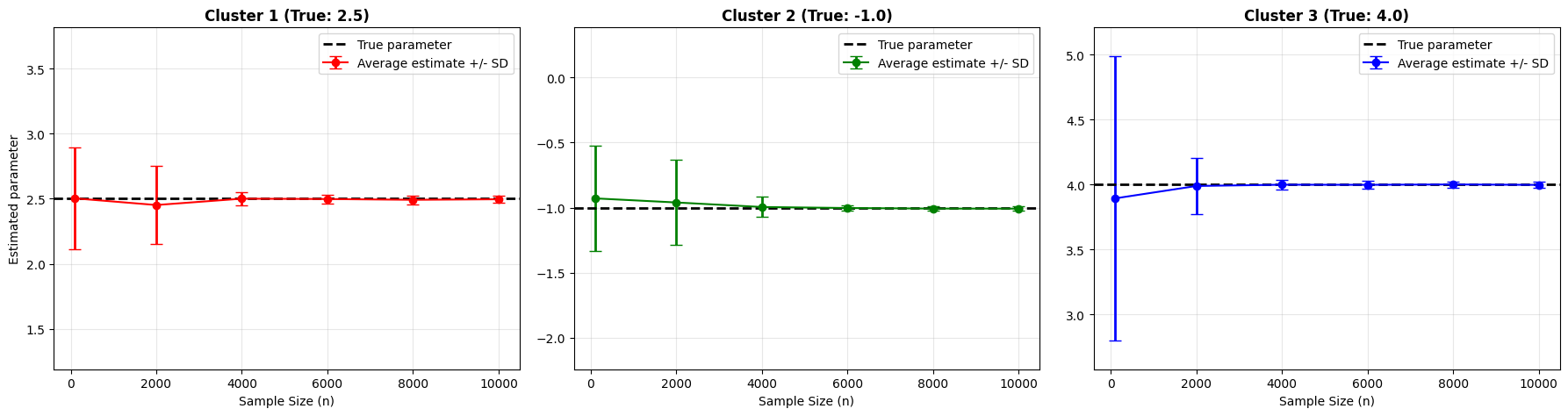}
    \caption{Results with $100$ independent runs. In the vertical axis, we plot $\hat{\beta}_{u,2}$, for $u=1,2,3.$}
    \label{fig:cate}
\end{figure}

\subsection{Simulation: Multi-treatment}
We simulate $n$ independent and identically distributed samples, where each sample $i\in[n]$ consists of three conditionally independent treatments $A^{(i)}_{1}, A^{(i)}_{2}, A^{(i)}_{3} \in \mathbb{R}$ and a continuous outcome $Y^{(i)} \in \mathbb{R}$. The data is generated from a mixture of $K=2$ latent clusters, denoted by the latent variable $U^{(i)} \in \{1, 2\}$, where $U^{(i)} \sim \text{Categorical}(\pi),  \pi = [0.5, 0.5].$
Conditioned on the latent state $U^{(i)} = k$, we independently sample three discrete treatment variables $A_1^{(i)}, A_2^{(i)}, A_3^{(i)} \in \{0, 1, 2, 3, 4\}$. The cluster-specific emission probabilities, $\mathbb{P}(A_v \mid U=k)$ for $v \in \{1,2,3\}$, are generated such that the resulting $5 \times 2$ probability matrices maintain full column rank. Because the sum of the Kruskal ranks of these matrices strictly satisfies the identifiability threshold ($2+2+2 \ge 2(2)+2$), the third-order moment tensor of the joint treatment distribution is guaranteed to be uniquely decomposable.
Finally, the continuous outcome $Y^{(i)} \in \mathbb{R}$ is generated according to the linear structural model:$$Y^{(i)} = \beta_{U^{(i)}, 0} + \beta_{U^{(i)}, 1}A_1^{(i)} + \beta_{U^{(i)}, 2}A_2^{(i)} + \beta_{U^{(i)}, 3}A_3^{(i)} + \epsilon^{(i)},$$where the  parameters are set to $\beta_{1} = [1.0, 0.5, 2.5, -0.5]^\top$ and $\beta_{2} = [-1.0, 1.5, -1.0, 0.8]^\top$, mapping the treatments to the outcome. The noise is drawn as $\epsilon^{(i)} \sim \mathcal{N}(0, 1)$.
We evaluate the proposed estimator across varying sample sizes, performing $100$ independent trials for each sample size. \Cref{fig:cate-multitr} shows the convergence of the estimated coefficient $\|\hat{\beta}_{u}\|$ as sample size increases.
\begin{figure}[ht]
    \centering
    \includegraphics[width=0.6\linewidth]{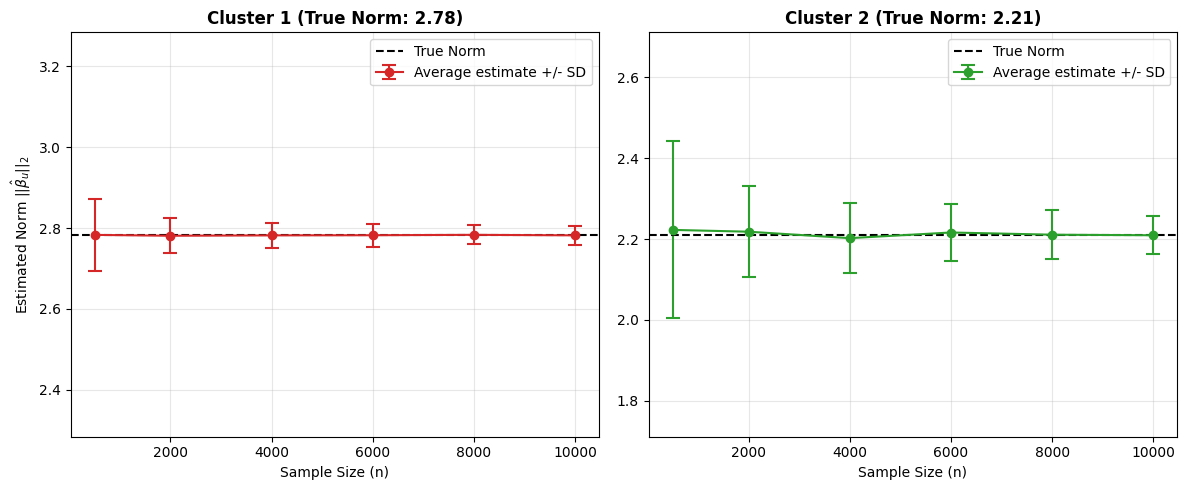}
    \caption{Results with $100$ independent runs. In the vertical axis, we plot $\|\hat{\beta}_{u}\|$, for $u=1,2.$}
    \label{fig:cate-multitr}
\end{figure}

\subsection{Real data}
We evaluate our method on the \texttt{bwght} dataset \citep{wooldridge2010econometric} to estimate the average treatment effect (ATE) of maternal smoking ($A$) on infant birth weight ($Y$). The primary identification challenge is the unmeasured confounder $U$: the latent family socio-economic status. We consider three proxies --- family income ($Z_1$), father's years of education ($Z_2$), and mother's years of education ($Z_3$). We embed the proxies into an RKHS via a Gaussian RBF kernel.
To determine the cardinality of the unobserved confounder $U$, we analyzed the spectrum of the RKHS cross-covariance operator $\hat{C}_{Z_1,Z_2}$. The scree plot of the singular values exhibits a rank-3 structure and suggests using $K=3$. This tripartite latent structure aligns with epidemiological consensus, capturing a high-risk disadvantaged class, a discordant middle-tier class, and a highly advantaged protective class. To interpret the recovered latent space, we assigned subjects to discrete states using maximum a posteriori (MAP) estimation on the recovered posteriors, $\arg\max_k w_k(Z_i)$. \Cref{fig:proxy-density} shows the empirical distribution of the proxies across these assignments, demonstrating that the model successfully reconstructs the underlying socio-economic stratification. Robust power iterations are then used to recover the discrete priors and continuous observation densities. We then fit a weighted Poisson regression model for $A$ with all available features except parity and gender of the child (which clearly does not affect $A$) and a linear regression model $Y$, with all available features. For the outcome model, regression coefficients corresponding to treatment $A$ for the three classes are $-0.0224$,$-0.0273$ and $-0.0241$ (standard deviations $0.01666$, $0.0281$ and $0.0168$) respectively. These estimates indicate a negative causal effect of maternal smoking on birth weight across all latent socio-economic strata. This finding reinforces prior results (e.g., \cite{mullahy1997instrumental}) on smoking-induced fetal growth restriction.

\begin{figure}[ht]
    \centering
    \includegraphics[width=0.8\linewidth]{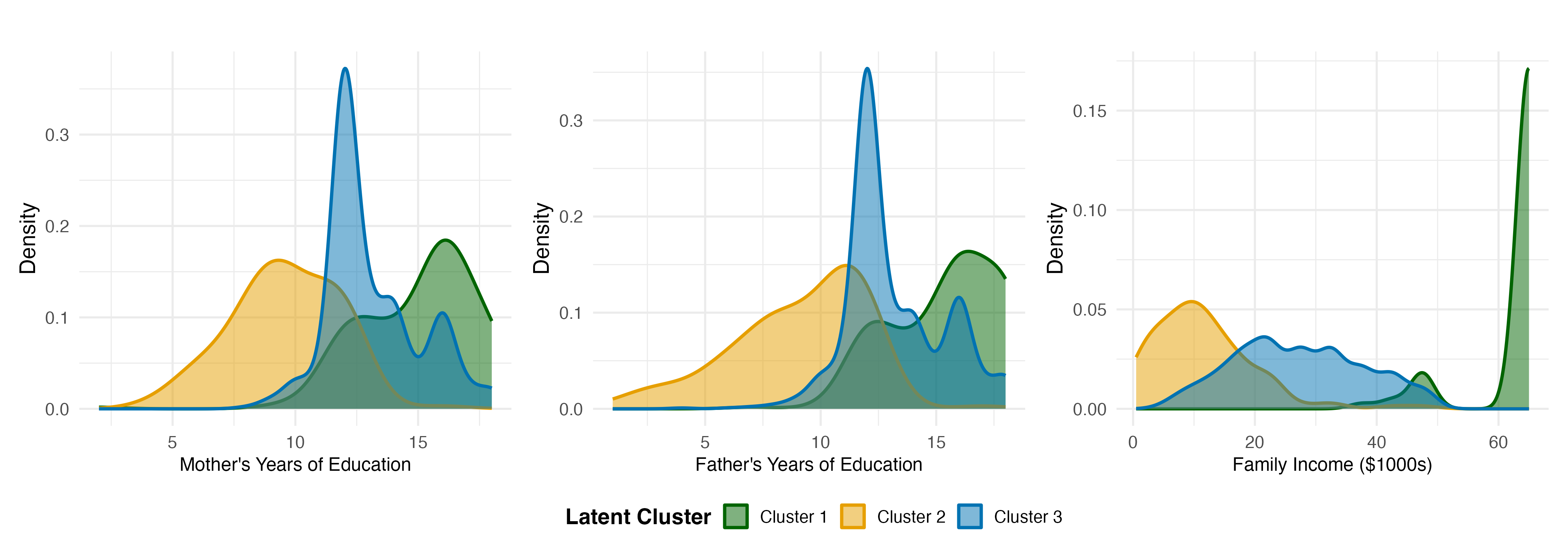}
    \caption{Empirical density of different proxies across the recovered latent classes, where the class assignments were computed via MAP estimation. The separation of the proxy densities shows that the non-parametric decomposition successfully reconstructs the underlying socio-economic heterogeneity.}
    \label{fig:proxy-density}
\end{figure}

\section{Conclusion and Future Directions}
\label{sec:conclusion}

In this work, we addressed the challenging problem of causal effect estimation in the presence of unobserved confounding by leveraging multi-view data. Unlike existing proximal causal inference methods that require rigid classification of proxies into negative control exposures and outcomes, our framework allows for a flexible usage of three or more available proxies. By formulating the problem via mixture learning with tensor decomposition, we avoided the ill-posed inverse problems typical of integral equation approaches, instead providing a computationally stable estimation procedure with rigorous finite-sample error guarantees.
While our current results are promising, several avenues for future research remain.
While we treat the proxy distributions nonparametrically, our current estimation algorithm assumes parametric forms for the treatment assignment mechanism and the outcome model. A natural extension would be to generalize the estimation framework to the \textit{fully nonparametric} setting.

\subsection*{Acknowledgement}
The authors thank Alberto Abadie and Yixin Wang for their helpful feedback.
This work was supported in part by the NSF FODSI project (Award No.\ 2022448) and a research project on causal inference supported by Generali.

\bibliography{ref}
\bibliographystyle{apalike}

\newpage
\appendix

\section{Auxiliary Lemmas}
\begin{lemma}
\label{lem:matrx-conc}
Let $\{x_i\}_{i=1}^n \subset \mathbb{R}^d$ be i.i.d. random vectors with mean $\mathbb{E}[x_i] = \mu$. Assume that the centered vectors $\tilde{x}_i = x_i - \mu$ are $\sigma$-sub-Gaussian. Define  $\Sigma = \mathbb{E}[x_ix_i^\top]$ and $\widehat{\Sigma} = \frac{1}{n} \sum_{i=1}^n x_ix_i^\top$. Then, for any $\delta > 0$, there exist universal constants $c_1, c_2 > 0$ such that with probability at least $1 -2\delta$, we have:$$\|\widehat{\Sigma} - \Sigma\|_2 \leq \frac{d}{n}+(1+ 8 \sigma \|\mu\|_2)\sqrt{\frac{d}{n}}+\max\left( \sqrt{\frac{\ln(c_1/\delta)}{c_2 n}}, \frac{\ln(c_1/\delta)}{c_2 n} \right)+8\sigma \|\mu\|_2 \cdot \sqrt{\frac{ \log(1/\delta)}{n}}.$$
\end{lemma}

\begin{proof}
We define the error matrix $\Delta := \widehat{\Sigma} - \Sigma$. Substituting the identity $x_i = \tilde{x}_i + \mu$ into the empirical sum yields:$$\Delta = \left( \frac{1}{n} \sum_{i=1}^n (\tilde{x}_i + \mu)(\tilde{x}_i + \mu)^\top \right) - (\text{Cov}(x) + \mu\mu^\top)$$Expanding the product and canceling the $\mu\mu^\top$ terms, we obtain:$$\Delta = \underbrace{\left( \frac{1}{n} \sum_{i=1}^n \tilde{x}_i \tilde{x}_i^\top - \text{Cov}(x) \right)}_{T_1} + \underbrace{\frac{1}{n} \sum_{i=1}^n \tilde{x}_i \mu^\top + \frac{1}{n} \sum_{i=1}^n \mu \tilde{x}_i^\top}_{T_2}$$By the triangle inequality, $\|\Delta\|_2 \leq \|T_1\|_2 + \|T_2\|_2$.

The term $T_1$ represents the concentration of the centered sample covariance matrix. Since $\tilde{x}_i$ is zero-mean and $\sigma$-sub-Gaussian, by (Theorem 6.5 of \cite{wainwright2019high}), we have:
$$\|T_1\|_2 \leq  \frac{d}{n}+\sqrt{\frac{d}{n}}+\max\left( \sqrt{\frac{\ln(c_1/\delta)}{c_2 n}}, \frac{\ln(c_1/\delta)}{c_2 n} \right),$$
with probability at least $1 -\delta$.

Let $u = \frac{1}{n} \sum_{i=1}^n \tilde{x}_i$ be the empirical mean of the centered vectors. The term $T_2$ simplifies to the rank-two matrix $u\mu^\top + \mu u^\top$. Its operator norm is bounded by:
$$\|T_2\|_2 \leq 2 \|u\|_2 \|\mu\|_2.$$
Because $\tilde{x}_i$ is $\sigma$-sub-Gaussian and zero-mean, $u$ is sub-Gaussian with parameter $\sigma^2/n$. Using the concentration of the $\ell_2$-norm for sub-Gaussian vectors, we have with probability at least $1 - \delta$:
$$\|u\|_2 \leq 4 \sigma \sqrt{\frac{d}{n}}+2\sigma \sqrt{\frac{\log(1/\delta)}{n}}.$$
Thus, $T_2 \leq 8 \sigma \|\mu\|_2 \sqrt{\frac{d}{n}}+4\sigma \|\mu\|_2 \sqrt{\frac{\log(1/\delta)}{n}}$.

Combining the bounds for $T_1$ and $T_2$ via a union bound, we conclude that:$$\|\Delta\|_2 \leq \frac{d}{n}+(1+ 8 \sigma \|\mu\|_2)\sqrt{\frac{d}{n}}+\max\left( \sqrt{\frac{\ln(c_1/\delta)}{c_2 n}}, \frac{\ln(c_1/\delta)}{c_2 n} \right)+8\sigma \|\mu\|_2 \cdot \sqrt{\frac{ \log(1/\delta)}{n}},$$
with probability $1-2\delta.$
\end{proof}

\begin{lemma}
\label{lem:lin-reg}
Consider the linear regression $\mathbb{E}[Y|X] = X^\top \gamma$. Let $X_i \in \mathbb{R}^d$ be i.i.d. random vectors with a non-zero mean $\mathbb{E}[X_i] = \mu$ and $\mathbb{E}[X_iX_i^\top] = \Sigma$.
Assume that
\begin{enumerate}
    \item The noise $\epsilon(X):=Y-\gamma^\top X$ is such that for all $\lambda\in\mathbb{R}$, almost
surely: $$\mathbb{E}(\exp(\lambda\epsilon(X)\mid X)\leq\exp\left(\frac{\lambda^2\sigma^2_{noise}}{2}\right).$$
\item The vector $\Sigma^{-1/2}(X-\mu)$ is sub-Gaussian, i.e., for all $\alpha\in\mathbb R^d$:
$$\mathbb{E}(\exp(\alpha^\top\Sigma^{-1/2}(X-\mu))\leq\exp\left(\frac{\|\alpha\|^2\sigma^2}{2}\right).$$
\end{enumerate}
Then, for sufficiently large $n$, with probability $1-3\delta$,
\begin{equation}
\|\hat{\gamma}_{\text{ols}} - \gamma\|^2 \leq \frac{\sigma^2}{n\lambda_{\min}({\Sigma})(1-\eta_{n,\delta})}\left(d + 2\sqrt{d \log(1/\delta)}+2\log(1/\delta)+\frac{2\log(1/\delta)}{1-\eta_{n,\delta}}\right),
\end{equation}
where $ \eta_{n,\delta}= \frac{d}{n}+(1+ 8 \sigma \| \Sigma^{-1/2}\mu\|_2)\sqrt{\frac{d}{n}}+\max\left( \sqrt{\frac{\ln(c_1/\delta)}{c_2 n}}, \frac{\ln(c_1/\delta)}{c_2 n} \right)+8\sigma \| \Sigma^{-1/2}\mu\|_2 \cdot \sqrt{\frac{ \log(1/\delta)}{n}},$ and $c_1, c_2 > 0$ are universal constants.

\end{lemma}
\begin{proof}
The OLS estimator satisfies the empirical normal equations $\hat{\Sigma}\hat{\gamma}_{\text{ols}} = \frac{1}{n}\sum_{i=1}^n[XY]$. Substituting $Y = X^\top \gamma + \epsilon(X)$:$$\hat{\Sigma}\hat{\gamma}_{\text{ols}} = \frac{1}{n}\sum_{i=1}^n[X(X^\top \gamma + \epsilon(X))] = \hat{\Sigma} \gamma + \frac{1}{n}\sum_{i=1}^n[X \epsilon(X)].$$
Rearranging gives the error identity:$$\hat{\gamma}_{\text{ols}} - \gamma = \hat{\Sigma}^{-1} \frac{1}{n}\sum_{i=1}^n[X \epsilon(X)].$$
  This identity holds for any design matrix $\mathbf{X}$, regardless of whether it is centered or shifted by a mean $\mu$.

Substitute the OLS identity into the expression:$$\Sigma^{1/2}(\hat{\gamma}_{\text{ols}} - \gamma) = \Sigma^{1/2} \hat{\Sigma}^{-1} \frac{1}{n}\sum_{i=1}^nX_i \epsilon(X_i).$$By inserting the identity $I = \hat{\Sigma}^{-1/2} \hat{\Sigma}^{1/2}$, we obtain:$$\Sigma^{1/2}(\hat{\gamma}_{\text{ols}} - \gamma) = \left( \Sigma^{1/2} \hat{\Sigma}^{-1/2} \right) \frac{1}{n}\sum_{i=1}^n[\hat{\Sigma}^{-1/2} X_i \epsilon(X_i)].$$
  
  Using the standard inequality $\|Az\| \leq \|A\| \|z\|$, we have:$$\|\Sigma^{1/2}(\hat{\gamma}_{\text{ols}} - \gamma)\|^2 \leq \|\Sigma^{1/2} \hat{\Sigma}^{-1/2}\|^2 \cdot \|\frac{1}{n}\sum_{i=1}^n[\hat{\Sigma}^{-1/2} X_i \epsilon(X_i)]\|^2.$$
  Since $\|\Sigma^{1/2} \hat{\Sigma}^{-1/2}\|^2 = \|\Sigma^{1/2} \hat{\Sigma}^{-1} \Sigma^{1/2}\|$, the conclusion follows:
  $$\|\hat{\gamma}_{\text{ols}} - \gamma\|_{\Sigma}^2 \leq \|\Sigma^{1/2} \hat{\Sigma}^{-1} \Sigma^{1/2}\| \cdot \|\frac{1}{n}\sum_{i=1}^n[\hat{\Sigma}^{-1/2} X_i\epsilon(X_i)]\|^2.$$
  
  Define  $\widetilde{\Sigma} = \frac{1}{n} \sum Z_i Z_i^\top$ where $Z_i = \Sigma^{-1/2}X_i$. So, $\widetilde{\Sigma} = \Sigma^{-1/2}\hat{\Sigma}\Sigma^{-1/2}$.

Therefore, $\|\Sigma^{1/2} \hat{\Sigma}^{-1} \Sigma^{1/2}\| =\|\widetilde{\Sigma}^{-1}\|=\frac{1}{\lambda_{\min}(\widetilde{\Sigma})}$.
By Weyl's inequality,
$$\lambda_{\min}(\widetilde{\Sigma}) \geq 1 - \|\widetilde{\Sigma} - I\|.$$
By \Cref{lem:matrx-conc} and condition 2, with probability $1-\delta$, $  \|\widetilde{\Sigma} - I\|\leq \eta_{n,\delta}$, where
\begin{align*}
  \eta_{n,\delta}= & \frac{d}{n}+(1+ 8 \sigma \| \Sigma^{-1/2}\mu\|_2)\sqrt{\frac{d}{n}}+\max\left( \sqrt{\frac{\ln(c_1/\delta)}{c_2 n}}, \frac{\ln(c_1/\delta)}{c_2 n} \right)+8\sigma \| \Sigma^{-1/2}\mu\|_2 \cdot \sqrt{\frac{ \log(1/\delta)}{n}}.
\end{align*}
Therefore, with probability $1-\delta$,
\begin{equation}
\label{eq:first-term}
    \|\Sigma^{1/2} \hat{\Sigma}^{-1} \Sigma^{1/2}\| \leq\frac{1}{1-\eta_{n,\delta}}.
\end{equation}
\begin{align*}
    \left\|\frac{1}{n}\sum_{i=1}^n[\hat{\Sigma}^{-1/2} X_i \epsilon(X_i)]\right\|^2&\leq\left\|\frac{1}{n}\sum_{i=1}^n[\hat{\Sigma}^{-1/2} (X_i-\mu) \epsilon(X_i)]\right\|^2+\|\hat{\Sigma}^{-1/2} \mu \|^2\left(\frac{1}{n}\sum_{i=1}^n\epsilon(X_i)\right)^2.
\end{align*}
Since $X_i-\mu$ are zero-mean vectors, by Lemma 5 of \cite{hsu2011analysis} and our condition 1, with probability $1-\delta$,
\begin{equation*}
    \left\|\frac{1}{n}\sum_{i=1}^n[\hat{\Sigma}^{-1/2} (X_i-\mu) \epsilon(X_i)]\right\|^2\leq\frac{\sigma_{noise}^2}{n}\left(d + 2\sqrt{d \log(1/\delta)}+2\log(1/\delta)\right).
\end{equation*}
With probability at least $1-\delta$, we have the following sub-gaussian bound:$$\left|\frac{1}{n}\sum_{i=1}^n \epsilon(X_i)\right| \leq \sigma_{noise} \sqrt{\frac{2\log(2/\delta)}{n}}.$$
Since $\mu\mu^\top \preceq \Sigma$, it follows that $\mu^\top \Sigma^{-1} \mu \leq 1$. Therefore,
$$\|\hat{\Sigma}^{-1/2} \mu \|^2 = \mu^\top \hat{\Sigma}^{-1} \mu \leq \|\Sigma^{1/2} \hat{\Sigma}^{-1} \Sigma^{1/2}\| \cdot (\mu^\top \Sigma^{-1} \mu)\leq \|\Sigma^{1/2} \hat{\Sigma}^{-1} \Sigma^{1/2}\|.$$
Finally, combining all the bounds, with probability $1-3\delta$,
\begin{equation}
    \|\hat{\gamma}_{\text{ols}} - \gamma\|_{\Sigma}^2 \leq \frac{\sigma_{noise}^2}{n(1-\eta_{n,\delta})}\left(d + 2\sqrt{d \log(1/\delta)}+2\log(1/\delta)+\frac{2\log(1/\delta)}{1-\eta_{n,\delta}}\right).
\end{equation}
Hence, with probability $1-3\delta$,
\begin{equation}
\|\hat{\gamma}_{\text{ols}} - \gamma\|^2 \leq \frac{\sigma_{noise}^2}{n\lambda_{\min}({\Sigma})(1-\eta_{n,\delta})}\left(d + 2\sqrt{d \log(1/\delta)}+2\log(1/\delta)+\frac{2\log(1/\delta)}{1-\eta_{n,\delta}}\right).
\end{equation}
\end{proof}

\begin{lemma}
\label{lem:reg-error-2}
Let $\hat{\gamma}_{\text{ols}}$ and $\hat{\gamma}$ be the OLS estimates derived from samples $(X_i, Y_i)_{i=1}^n$ and perturbed samples $(\hat{X}_i, \hat{Y}_i)_{i=1}^n$, respectively: $$ \hat{\gamma}_{\text{ols}}=\D^{-1} \left( \frac{1}{n} \sum_{i=1}^n X_i Y_i \right),~~\hat{\gamma}=\hat\D^{-1} \left( \frac{1}{n} \sum_{i=1}^n \hat X_i \hat Y_i \right),$$
where $\D= \frac{1}{n} \sum_{i=1}^n X_i X_i^\top,$ and $\hat\D= \frac{1}{n} \sum_{i=1}^n\hat X_i\hat X_i^\top$.
Let  $\mathcal{A},\mathcal{B},\mathcal{C}$ and $\mathcal{D}$ denote the events 
  $\| \Sigma - \D \|\leq  \epsilon_{n,\delta,1},$ $ \| \hD - \D \| \le \epsilon_{n,\delta,2},$
  $\left\| \frac{1}{n} \sum_{i=1}^n (\hat X_i \hat Y_i - X_i Y_i) \right\|\le \epsilon_{n,\delta,3}$
  and $\frac{1}{n} \sum_{i=1}^n \|\hat X_i\|| \hat Y_i|\le B_{n,\delta}$ respectively, such that $\epsilon_{n,\delta,1},\epsilon_{n,\delta,2},\epsilon_{n,\delta,3}\to 0,$ as $n\to\infty$. Then, for sufficiently large $n$,
on the event $\mathcal{A}\cap\mathcal{B}\cap\mathcal{C}\cap\mathcal{D}$,
\begin{align}
\nonumber  | \hat{\gamma} - \hat{\gamma}_{\text{ols}}| \le&\frac{\epsilon_{n,\delta,3}}{\lambda-\epsilon_{n,\delta,1}}+\frac{B_{n,\delta}\cdot\epsilon_{n,\delta,2}}{(\lambda-\epsilon_{n,\delta,1})(\lambda-\epsilon_{n,\delta,1}- \epsilon_{n,\delta,2})},
\end{align}  
where $\Sigma=\mathbb E(XX^\top)$ and $\lambda=\lambda_{\min}(\Sigma)$.
\end{lemma}

\begin{proof}

We decompose the error as:
  \begin{align*}
\hat{\gamma} - \hat{\gamma}_{\text{ols}} &= \hD^{-1} \left( \frac{1}{n} \sum \hat X_i \hat Y_i \right) - \D^{-1} \left( \frac{1}{n} \sum X_i Y_i \right) \\
&= \underbrace{\D^{-1} \left( \frac{1}{n} \sum (\hat X_i \hat Y_i - X_i Y_i) \right)}_{T_1 } + \underbrace{(\hD^{-1} - \D^{-1}) \left( \frac{1}{n} \sum \hat X_i \hat Y_i \right)}_{T_2}.
\end{align*}

\paragraph{Step 1: Inverse Matrix Stability}

From Weyl's inequality:
\[ \lambda_{\min}(\D) \ge  \lambda_{\min}(\Sigma_W)- \| \Sigma_W - \D \|. \]
Under $\mathcal{E}_{D}$, we have $\lambda_{\min}(\D) \geq \lambda-\epsilon_{n,\delta,1}$.
 From Weyl's inequality:
  \[ \lambda_{\min}(\hD) \ge \lambda_{\min}(\D) - \| \hD - \D \|. \]
We also have $\| \hD - \D \| \le  \epsilon_{n,\delta,2} $.
Therefore,
  \[ \lambda_{\min}(\hD) \ge \lambda-\epsilon_{n,\delta,1}- \epsilon_{n,\delta,2}. \]
Thus, for sufficiently large $n$, i.e., when $\lambda-\epsilon_{n,\delta,1}- \epsilon_{n,\delta,2}>0$, 
$$\|\D^{-1}\| =\lambda_{\max}(\D^{-1})\le \frac{1}{\lambda-\epsilon_{n,\delta,2}},~~\|\hD^{-1}\| =\lambda_{\max}(\hD^{-1})\le \frac{1}{\lambda-\epsilon_{n,\delta,1}- \epsilon_{n,\delta,2} }.$$

  Using the identity $\hD^{-1} - \D^{-1} = -\hD^{-1} (\Delta \mathbf{D}) \D^{-1}$, we have the following on $\mathcal{E}_{D}\cap\mathcal{E}_{W}$:
  \begin{align*}
\| \hD^{-1} - \D^{-1} \| &\le \| \hD^{-1} \| \| \Delta \mathbf{D} \| \| \D^{-1} \| \\
&\le\frac{\epsilon_{n,\delta,2}}{(\lambda-\epsilon_{n,\delta,2})(\lambda-\epsilon_{n,\delta,2}- \epsilon_{n,\delta,1})}.
\end{align*}

\paragraph{Step 2: Bounding Term $T_1$}

The term $T_1$ represents the projection of the feature and target errors.
We bound the norm of the average:
  \begin{align*}
\| T_1 \|
  &\le \| \D^{-1} \| \left\| \frac{1}{n} \sum_{i=1}^n (\hat X_i \hat Y_i - X_i Y_i) \right\| \\
&\le \frac{\epsilon_{n,\delta,3}}{\lambda-\epsilon_{n,\delta,1}}.
\end{align*}

\paragraph{Step 3: Bounding Term $T_2$}

On $\mathcal{E}_{W}\cap\mathcal{E}_{\alpha}\cap\mathcal{E}_{\phi}\cap\mathcal{E}_{D}$,
\begin{align*}
\| T_2 \| &=\Bigg\| (\hD^{-1} - \D^{-1}) \left( \frac{1}{n} \sum_{i=1}^n \hat X_i \hat Y_i \right)\Bigg\|_2\\
&\le \| \hD^{-1} - \D^{-1} \| \left( \frac{1}{n} \sum_{i=1}^n \|\hat X_i\|| \hat Y_i| \right) \\
&\le  \frac{\epsilon_{n,\delta,2}}{(\lambda-\epsilon_{n,\delta,1})(\lambda-\epsilon_{n,\delta,1}- \epsilon_{n,\delta,2})}B_{n,\delta}.
\end{align*}

Combining the bounds on $T_1$ and $T_2$,    we derive that
with probability at least $1 - \delta$:
  \begin{align}
\label{bdd:sigma-err-1}
\nonumber  | \hat{\gamma} - \hat{\gamma}_{\text{ols}}| \le&\frac{\epsilon_{n,\delta,3}}{\lambda-\epsilon_{n,\delta,1}}+\frac{\epsilon_{n,\delta,2}}{(\lambda-\epsilon_{n,\delta,1})(\lambda-\epsilon_{n,\delta,1}- \epsilon_{n,\delta,2})}B_{n,\delta}.
\end{align}

Combining \eqref{bdd:sigma-err-1} and \eqref{bdd:sigma-err-2}, we get the desired bound.
\end{proof}

\section{Omitted Proofs}

\begin{proof}[Proof of Proposition \ref{prop:pd-verify}]
\textbf{To show: $\mathbb{E}[W^*(Z)W^*(Z)^\top]$ is positive definite.}
Suppose there exist constants $\mathbf{c}=(c_1, c_2)\neq(0,0)$ such that $\mathbb{E}[\mathbf{c}^\top W^*(Z)W^*(Z)^\top\mathbf{c}]=0$. This implies that  $\mathbb{E}[(\mathbf{c}^\top W^*(Z))^2]=0$, i.e.,
\[ \mathbb{P}\left(\mathbf{c}^\top W^*(Z) = 0 \right)= \mathbb{P}\left( c_1 f_1(Z) + c_2 f_2(Z) = 0 \right) =1. \]
Conditioning on the event, we have $c_1 f_1(Z) = -c_2 f_2(Z)$. Taking the logarithm of the ratio:
\[ \ln\left( \frac{f_1(Z)}{f_2(Z)} \right) = \ln\left( -\frac{c_2}{c_1} \right) \implies Z^\top(\boldsymbol{\mu}_1 - \boldsymbol{\mu}_2)  = \text{constant.} \]
Since $\boldsymbol{\mu}_1 \neq \boldsymbol{\mu}_2$,
this equation defines a hyperplane (an affine subspace of dimension $d-1$) in $\mathbb{R}^d$.
The Lebesgue measure of a hyperplane in $\mathbb{R}^d$ is zero, and $Z$ has Lebesgue density in $\mathbb R^d$. Thus, $\mathbb{P}(Z^\top(\boldsymbol{\mu}_1 - \boldsymbol{\mu}_2)  = \text{constant}) = 0$.
This contradicts the assumption that the probability is $1$. Therefore, we must have $c_1=c_2=0$. Thus, $\mathbb{E}[W^*(Z)W^*(Z)^\top]$ is positive definite.

\textbf{To show: $\mathbb{E}[\Phi^*(Z)\Phi^*(Z)^\top]$ is positive definite.}

Suppose there exist constant vector $\mathbf{d}\neq\mathbf{0}$ such that $\mathbb{E}[\mathbf{d}^\top\Phi^*(Z)\Phi^*(Z)^\top\mathbf{d}]=0$. This implies that  $\mathbb{E}[(\mathbf{d}^\top \Phi^*(Z))^2]=0$, i.e.,
\[ \mathbb{P}\left(\mathbf{d}^\top \Phi^*(Z) = 0 \right)=1. \]
Define the function $f(z) = \mathbf{d}^\top \Phi^*(z)$. Since $f$ is continuous function, by Lemma \ref{lem:ae}, the above equation implies that $f(z)$ must be the zero function (identically zero for all $z \in \mathbb{R}^d$), i.e.,
\[ f(z) = f_1(z) P_1(z) + f_2(z) P_2(z) \equiv 0, \]
where $P_u(z) = \mathbf{d}_u^\top \phi(z)$ and $\mathbf{d}=(\mathbf{d}_1^\top,\mathbf{d}_2^\top)^\top$. Dividing by $f_2(z)$:
\[ g(z) := e^{z^\top (\boldsymbol{\mu}_2 - \boldsymbol{\mu}_1) + C} P_1(z) + P_2(z) \equiv 0. \]
Take the limit along the ray $z = t(\boldsymbol{\mu}_2 - \boldsymbol{\mu}_1)$ as $t \to \infty$. The term $e^{t \|\boldsymbol{\mu}_2 - \boldsymbol{\mu}_1\|^2}$ grows faster than any polynomial $P_2(z)$ with $z = t(\boldsymbol{\mu}_2 - \boldsymbol{\mu}_1)$. For the expression to remain identically zero for all $t$, the coefficient of the exponential, $P_1(z)$, must vanish identically.
\[ P_1(z) = \mathbf{d}_1^\top \phi(z) \equiv 0 \implies \mathbf{d}_1 = \mathbf{0}, \]
$\text{since the functions } \phi_1(z)=1,\phi_2(z)=z_1,\phi_3(z)=z_2,\phi_4(z)=z_3 \text{ are linearly independent}.$
Now $P_1(z) \equiv 0$ implies $P_2(z) \equiv 0$, which implies $\mathbf{d}_2 = \mathbf{0},$ using the same argument. So, we have
$\mathbf{d}=\mathbf{0}$, which is a contradiction. Therefore, $\mathbb{E}[\Phi^*(Z)\Phi^*(Z)^\top]$ is positive definite.

\textbf{To show: $\mathbb{E}[\Psi^*(A,Z)\Psi^*(A,Z)^\top]$ is positive definite}

Suppose there exists a constant vector $\mathbf{k} \neq \mathbf{0}$ such that $\mathbf{k}^\top \mathbb{E}[\Psi^*(A, Z)\Psi^*(A, Z)^\top] \mathbf{k} = 0$. Since the matrix is symmetric and positive semi-definite, this implies the quadratic form is zero, which is equivalent to the random variable $\mathbf{k}^\top \Psi^*(A, Z)$ being zero almost surely:
\[ \mathbb{P}\left( \mathbf{k}^\top \Psi^*(A, Z) = 0 \right) = 1. \]
We analyze the event $E = \{ (a, z) \in \mathbb{R} \times \mathbb{R}^3 : \mathbf{k}^\top \Psi^*(a, z) = 0 \}$. By the law of total expectation, we can write the probability as an integral over the marginal density of $Z$:
\[ \mathbb{P}(E) = \int_{\mathbb{R}^3} \mathbb{P}\left( \mathbf{k}^\top \Psi^*(A, z) = 0 \mid Z=z \right) f_Z(z) dz = 1. \]
For this integral to equal 1, the inner probability term must equal 1 for almost all $z$ (with respect to Lebesgue measure on $\mathbb{R}^3$).
Let $\mathcal{Z}_{good} = \{ z \in \mathbb{R}^3 : \mathbb{P}(\mathbf{k}^\top \Psi^*(A, z) = 0 \mid Z=z) = 1 \}$.
We proceed in two steps.

\paragraph{Step 1:}
The treatment means are given by linear functions $\mu_{A,1}(z) = \alpha_1^\top \phi(z)$ and $\mu_{A,2}(z) = \alpha_2^\top \phi(z)$. By the model design, the parameter vectors are distinct, i.e., $\Delta \alpha = \alpha_1 - \alpha_2 \neq \mathbf{0}$.
Consider the set of $z$ where the means coincide:
\[ \mathcal{Z}_{equal} = \{ z \in \mathbb{R}^3 : \alpha_1^\top \phi(z) = \alpha_2^\top \phi(z) \}. \]
Substituting $\phi(z) = [1, z^\top]^\top$, the condition becomes:
\[ (\alpha_{1,0} - \alpha_{2,0}) + (\tilde{\alpha}_1 - \tilde{\alpha}_2)^\top z = 0. \]
Since $\Delta \alpha \neq \mathbf{0}$, this equation defines an affine hyperplane $H$ in $\mathbb{R}^3$ (a set of dimension strictly less than 3). The Lebesgue measure of a hyperplane in $\mathbb{R}^3$ is zero.
Since the distribution of $Z$ is absolutely continuous with respect to Lebesgue measure, $\mathbb P[Z \in \mathcal{Z}_{equal}]=0$.
Thus, for almost all $z \in \mathbb{R}^3$, we have $\mu_{A,1}(z) \neq \mu_{A,2}(z)$. 
\paragraph{Step 2:}
Let us restrict our analysis to a fixed $z \in \mathcal{Z}_{good} \setminus \mathcal{Z}_{equal}$.
The condition $\mathbf{k}^\top \Psi^*(A, z) = 0$ almost surely (w.r.t $A$) implies that the function
\[ h_z(A) = f_1(z)e_1(A,z) Q_1(A) + f_2(z)e_2(A,z) Q_2(A) \]
must vanish on a set of measure 1 in $\mathbb{R}$. Since $h_z(A)$ is a sum of products of Gaussian densities and polynomials, it is a \textit{continuous function} on $\mathbb{R}$. Thus, by Lemma \ref{lem:ae}, $h_z(a) \equiv 0$ for all $a \in \mathbb{R}$.

Here, $e_u(A,z) = \frac{1}{\sqrt{2\pi}} e^{-\frac{1}{2}(A - \mu_u)^2}$ and $Q_u(A) = f_u(z)\mathbf{k}_u^\top [1, A, z^\top]$ is a polynomial in $A$ of degree at most 1.
We assume for contradiction that the polynomials are not both zero. The identity $h_z(a) \equiv 0$ implies:
\[ e^{-\frac{1}{2}(a - \mu_1)^2} Q_1(a) = - e^{-\frac{1}{2}(a- \mu_2)^2} Q_2(a). \]
Assume without loss of generality that $Q_2(a)$ is not identically zero. We can rearrange the equation:
\[ \left| \frac{Q_1(a)}{Q_2(a)} \right| = \frac{e^{-\frac{1}{2}(a- \mu_2)^2}}{e^{-\frac{1}{2}(a - \mu_1)^2}} = \exp\left( -\frac{1}{2}\left[ 2a(\mu_1 - \mu_2) + (\mu_2^2 - \mu_1^2)\right]\right).\]
We now examine the growth rates as $a \to \infty$.
The left-hand side is a ratio of polynomials, which behaves asymptotically like $a^d$ for some integer $d$. The right-hand side is an exponential function $e^{\lambda a}$ with rate $\lambda = \mu_2 - \mu_1$.
Since $\mu_1 \neq \mu_2$, $\lambda \neq 0$.
\begin{itemize}
    \item If $\lambda > 0$ (i.e., $\mu_2 > \mu_1$), then as $a \to \infty$, the RHS grows exponentially to $\infty$, while the LHS grows at most polynomially. This is a contradiction.
    \item If $\lambda < 0$ (i.e., $\mu_2 < \mu_1$), then as $a \to -\infty$, let $a= -x$ where $x \to \infty$. The term $\lambda a = (-\lambda) x$ becomes positive exponential growth, again contradicting the polynomial growth of the LHS.
\end{itemize}
The only way to avoid contradiction is if the assumption that $Q_2(a)$ is not zero is false. Thus, $Q_2(a) \equiv 0$. Substituting this back into the original equation gives $e_1 Q_1 \equiv 0 \implies Q_1(a) \equiv 0$.

We have shown that for almost all $z$, the polynomials $Q_u(a) = k_{u,a} a + k_{u,0}(z)$ must be identically zero. Therefore,
\[ k_{u,A} = 0 \quad \text{and} \quad k_{u,0}(z) = k_{u, \text{int}} + \mathbf{k}_{u, z}^\top z = 0, \text{ for almost all } z. \]
Since it is continuous, $k_{u,0}(z)$ must be identically zero (Lemma \ref{lem:ae}); and hence the vector coefficients $\mathbf{k}_{u, z}$ must be zero.
Thus, the entire vector $\mathbf{k}$ must be zero, which contradicts the premise that $\mathbf{k} \neq \mathbf{0}$.
Therefore, the matrix is positive definite.
\end{proof}

\begin{proof}[Proof of Proposition \ref{prop:example-gen}]
 \textbf{Part 1.}   Note that under the assumed treatment model,
    the conditional density of $A$ at $A=a$ given $Z=z$ and $U=u$ belongs to the class
    \[\mathcal{H}_1=\left\{e: e(a,z)=\frac{1}{\sqrt{2\pi}\sigma_A}\exp{\left(-\frac{1}{2\sigma_A^2}(a-\alpha^\top\phi(z))^2\right)};\sigma_A>0,\alpha\in\mathbb{R}^4\right\},\]
    where $\phi(z)=[1,z^\top]^\top$.
  Suppose, there exists a set of functions $\{e_u\}_{u=1}^2\in \mathcal{H}_1^2$ and $\{e_u^\prime\}_{u=1}^2\in \mathcal{H}_1^2$, such that $\sum_{u=1}^2 \pi_u{f}_u(Z)  e_u(A,Z)= \sum_{u=1}^2 \pi_u{f}_u(Z)  e_u^\prime(A,Z)$ almost surely.
Then, it directly follows from Proposition \ref{prop:pd-verify} and \Cref{thm:identifiability} that $e_u\equiv e_u^\prime, \forall u.$

\textbf{Part 2.} Suppose there exists a set of functions $\{g_u\}_{u=1}^2\in \mathcal{H}_2^2$ and $\{g_u^\prime\}_{u=1}^2\in \mathcal{H}_2^2$, such that $\sum_{u=1}^2 \pi_u{f}_u(Z)  e_u(A,Z)g^\prime_u(A,Z)= \sum_{u=1}^2 \pi_u{f}_u(Z)  e_u(A,Z)g_u(A,Z)$ almost surely. Let, $g_u(a,z)=\sum_{j=0}^\infty\beta_{0,j,u}a^j+\sum_{i=0}^3\sum_{j=1}^\infty\beta_{i,j,u}z_i^j$ and $g_u^\prime(a,z)=\sum_{j=0}^\infty\beta_{0,j,u}^\prime a^j+\sum_{i=0}^3\sum_{j=1}^\infty\beta_{i,j,u}^\prime z_i^j$. Define, $k_{0,j,u}=\beta_{i,j,u}-\beta_{0,j,u}^\prime$ and 
\[
\zeta(a,z)=\sum_{u=1}^2 \pi_u{f}_u(z)  e_u(a,z)(\sum_{j=0}^\infty k_{0,j,u}a^j+\sum_{i=0}^3\sum_{j=1}^\infty k_{i,j,u}z_i^j).
\]
Suppose, for the sake of contradiction, that the vector $\mathbf{k}\in\mathbb{R}^\infty$, which consists of all $k_{i,j,u}$, is non-zero. 
Then we have
\[ \mathbb{P}\left(\zeta(A,Z) = 0 \right) = 1.\]
We analyze the event $E = \{ (a, z) \in \mathbb{R} \times \mathbb{R}^3 : \zeta(a, z) = 0 \}$. By the law of total expectation, we can write the probability as an integral over the marginal density of $Z$:
\[ \mathbb{P}(E) = \int_{\mathbb{R}^3} \mathbb{P}\left( \zeta(A, z) = 0 \mid Z=z \right) f_Z(z) dz = 1. \]
For this integral to equal 1, the inner probability term must equal 1 for almost all $z$ (with respect to Lebesgue measure on $\mathbb{R}^3$).
Let $\mathcal{Z}_{good} = \{ z \in \mathbb{R}^3 : \mathbb{P}(\zeta(A, z) = 0 \mid Z=z) = 1 \}$.
We proceed in two steps.

\paragraph{Step 1:}
The treatment means are given by linear functions $\mu_{A,1}(z) = \alpha_1^\top \phi(z)$ and $\mu_{A,2}(z) = \alpha_2^\top \phi(z)$. By the model design, the parameter vectors are distinct, i.e., $\Delta \alpha = \alpha_1 - \alpha_2 \neq \mathbf{0}$.
Consider the set of $z$ where the means coincide:
\[ \mathcal{Z}_{equal} = \{ z \in \mathbb{R}^3 : \alpha_1^\top \phi(z) = \alpha_2^\top \phi(z) \}. \]
Substituting $\phi(z) = [1, z^\top]^\top$, the condition becomes:
\[ (\alpha_{1,0} - \alpha_{2,0}) + (\tilde{\alpha}_1 - \tilde{\alpha}_2)^\top z = 0. \]
Since $\Delta \alpha \neq \mathbf{0}$, this equation defines an affine hyperplane $H$ in $\mathbb{R}^3$ (a set of dimension strictly less than 3). The Lebesgue measure of a hyperplane in $\mathbb{R}^3$ is zero.
Since the distribution of $Z$ is absolutely continuous with respect to Lebesgue measure, $\mathbb P[Z \in \mathcal{Z}_{equal}]=0$.
Thus, for almost all $z \in \mathbb{R}^3$, we have $\mu_{A,1}(z) \neq \mu_{A,2}(z)$. 
\paragraph{Step 2:}
Let us restrict our analysis to a fixed $z \in \mathcal{Z}_{good} \setminus \mathcal{Z}_{equal}$.
The condition $\zeta(A, z) = 0$ almost surely (w.r.t $A$) implies that the function
\[ h_z(A) = e_1(A,z) Q_{1,z}(A) + e_2(A,z) Q_{2,z}(A) \]
must vanish on a set of measure 1 in $\mathbb{R}$. Since $h_z(A)$ is a sum of products of Gaussian densities and polynomials, it is a \textit{continuous function} on $\mathbb{R}$. Thus, by Lemma \ref{lem:ae}, $h_z(a) \equiv 0$ for all $a \in \mathbb{R}$.
Here, $e_u(A,z) = \frac{1}{\sqrt{2\pi}} e^{-\frac{1}{2}(A - \mu_u)^2}$ and $Q_{u,z}(A) = \pi_uf_u(z)(\sum_{j=0}^\infty k_{0,j,u}A^j+\sum_{i=0}^3\sum_{j=1}^\infty k_{i,j,u}z_i^j)$ is a polynomial in $A$ of finite degree because $k_{0,u}=\beta_{0,u}-\beta^\prime_{0,u}$ is at most $2s$ sparse and $s$ is finite.
We assume for contradiction that the polynomials are not both zero. The identity $h_z(a) \equiv 0$ implies:
\[ e^{-\frac{1}{2}(a - \mu_1)^2} Q_{1,z}(a) = - e^{-\frac{1}{2}(a- \mu_2)^2} Q_{2,z}(a). \]
Assume without loss of generality that $Q_{2,z}(a)$ is not identically zero. We can rearrange the equation:
\[ \left| \frac{Q_{1,z}(a)}{Q_{2,z}(a)} \right| = \frac{e^{-\frac{1}{2}(a- \mu_2)^2}}{e^{-\frac{1}{2}(a - \mu_1)^2}} = \exp\left( -\frac{1}{2}\left[ 2a(\mu_1 - \mu_2) + (\mu_2^2 - \mu_1^2)\right]\right).\]
We now examine the growth rates as $a \to \infty$.
The left-hand side is a ratio of polynomials, which behaves asymptotically like $a^d$ for some integer $d$. The right-hand side is an exponential function $e^{\lambda a}$ with rate $\lambda = \mu_2 - \mu_1$.
Since $\mu_1 \neq \mu_2$, $\lambda \neq 0$.
\begin{itemize}
    \item If $\lambda > 0$ (i.e., $\mu_2 > \mu_1$), then as $a \to \infty$, the RHS grows exponentially to $\infty$, while the LHS grows at most polynomially. This is a contradiction.
    \item If $\lambda < 0$ (i.e., $\mu_2 < \mu_1$), then as $a \to -\infty$, let $a= -x$ where $x \to \infty$. The term $\lambda a = (-\lambda) x$ becomes positive exponential growth, again contradicting the polynomial growth of the LHS.
\end{itemize}
The only way to avoid contradiction is if the assumption that $Q_{2,z}(a)$ is not zero is false. Thus, $Q_{2,z}(a) \equiv 0$. Substituting this back into the original equation gives $e_1 Q_{1,z} \equiv 0 \implies Q_{1,z}(a) \equiv 0$.

We have shown that for almost all $z$, the polynomials $Q_{u,z}(a) = \pi_uf_u(z)(\sum_{j=0}^\infty k_{0,j,u}a^j+\sum_{i=0}^3\sum_{j=1}^\infty k_{i,j,u}z_i^j)$ must be identically zero. Therefore,
\[ k_{0,j,u} = 0, \forall j,u \quad \text{and} \quad \sum_{i=0}^3\sum_{j=1}^\infty k_{i,j,u}z_i^j = 0, \text{ for almost all } z. \]
Since it is continuous, $\sum_{i=0}^3\sum_{j=1}^\infty k_{i,j,u}z_i^j$ must be identically zero (Lemma \ref{lem:ae}); and hence the vector coefficients $\mathbf{k}_{i,j,u}$ must be zero.
Thus, the entire vector $\mathbf{k}$ must be zero, which contradicts the premise that $\mathbf{k} \neq \mathbf{0}$.
Therefore, $g_u\equiv g_u^\prime,$ for all $u=1,2$.
\end{proof}

\begin{proof}[Proof of \Cref{thm:identifiability-gen}] We provide the proof in a few steps as follows.

\textbf{Step 1.} Under \Cref{ass:indep,ass:lin-ind-z,ass:positivity-1}, Theorem 8 of \cite{allman2009identifiability} (see \Cref{app:nonparametric-mixtures}) guarantees that the parameters $\pi_u$ and the distributions $\mathbb P_{v,u}$ are uniquely identified (up to permutation). Let $f_{v,u}$ be the density corresponding to $\mathbb P_{v,u}$. With the densities and priors identified, the posterior membership probability is uniquely determined by Bayes' rule:
\begin{equation}
    w_u(Z)  = P(U=u \mid Z)=\frac{\pi_u \prod_{v=1}^3 f_{v,u}(Z_v)}{\sum_{k=1}^K \pi_k \prod_{v=1}^3 f_{v,k}(Z_v)},
\end{equation}
here we use the notation $Z=(Z_1,Z_2,Z_3).$

\textbf{Step 2.} We begin with the observed conditional density of treatment $A$ given proxies $Z$, which is a mixture over the $K$ latent states:
$$\mathbb P(A\mid Z) = \sum_{u=1}^K w_u(Z) e_u(A, Z).$$
Note that $\mathbb P(A\mid Z)$, is  unique almost surely and $ \sum_{u=1}^K w_u(Z) e_u(A, Z) = \sum_{u=1}^K w_u(Z) e_u^\prime(A, Z)$ can be equivalently written as $\sum_{u=1}^K \pi_uf_u(Z) e_u(A, Z) = \sum_{u=1}^K\pi_u f_u(Z) e_u^\prime(A, Z)$. Hence,
by \Cref{ass:treatment-model-rkhs}, each $\{e_u\}_{u=1}^K\in \mathcal{H}_1^K$ is uniquely identified.

\textbf{Step 3: Identification of the Outcome Functional.}
Using the latent-specific treatment densities $e_u(a, z)$ and the proxy-based posterior $w_u(z)$ identified in Steps 1 and 2, we compute the posterior weights via Bayes' Rule:
$$\tilde{w}_u(a, z) = P(U=u \mid A=a, Z=z) = \frac{e_u(a, z) w_u(z)}{\sum_{k=1}^K e_k(a, z) w_k(z)}.$$
Therefore, by \Cref{ass:outcome-model-rkhs}, $$\mathbb{E}[Y \mid A, Z] = \sum_{u=1}^K \tilde{w}_u(A, Z) h_u(A, Z).$$
Note that $\mathbb{E}[Y|A, Z]$ is  unique almost surely and $ \sum_{u=1}^K \tilde w_u(A,Z) h_u(A, Z) = \sum_{u=1}^K  \tilde w_u(A,Z) h_u^\prime(A, Z)$ can be equivalently written as $\sum_{u=1}^K \pi_uf_u(Z) e_u(A, Z)h_u(A, Z) = \sum_{u=1}^K\pi_u f_u(Z) e_u^\prime(A, Z)h_u^\prime(A, Z)$. Hence, by \Cref{ass:outcome-model-rkhs},  $\{h_u\}_{u=1}^K\in \mathcal{H}_2^K$ is uniquely identified.

\textbf{Step 4.} Now we conclude how the causal estimands can be identified. 
Since the set $(U,Z)$ is the parent set for $A$ in our DAG, it works as a valid adjustment set for the average treatment effect (ATE). Therefore,
\begin{align*}
    \tau(a)&= \mathbb E^{do(A=a)}[Y]\\
    &= \sum_{u=1}^K\pi_u \int \mathbb E[Y \mid A=a,  Z=z, U=u]~d\mathbb P_{1,u}(z_1)d\mathbb P_{2,u}(z_2)d\mathbb P_{3,u}(z_3)\\
    &=\sum_{u=1}^K\pi_u \int h_u(a,z)~d\mathbb P_{1,u}(z_1)d\mathbb P_{2,u}(z_2)d\mathbb P_{3,u}(z_3). 
\end{align*}
Since $\{\pi_u, \mathbb P_{1,u}, \mathbb P_{2,u}, \mathbb P_{3,u}\}_{u=1}^K$ and $\{h_u\}_{u=1}^K$ are uniquely identified up to permutations, as shown in stages 1 and 3 respectively, $\tau(a)$ is uniquely identifiable.
However, the \emph{conditional average treatment effect} (CATE):
\begin{align*}
   \tau_u(a,z)&= \mathbb E^{do(A=a)}[Y \mid U=u,Z=z]=\mathbb E[Y \mid A=a, U=u,Z=z]=h_u(a,z)
\end{align*}
   is uniquely identified only up to permutations.
\end{proof}

\begin{proof}[Proof of \Cref{thm:identifiability-no-proxy}]
We establish the theorem in three steps as follows. 

\noindent \textbf{Step 1.} Under \Cref{ass:indep,ass:lin-ind-z,ass:positivity-1}, Theorem 8 of \cite{allman2009identifiability} (see \Cref{app:nonparametric-mixtures}) guarantees that the parameters $\pi_u$ and the distributions $\mathbb Q_{v,u}$ are uniquely identified (up to permutation). Let $f_{v,u}$ be the density corresponding to $\mathbb Q_{v,u}$. With the densities and priors identified, the posterior membership probability is uniquely determined by Bayes' rule:
\begin{equation}
\label{eq:post-mem-prob-tr}
    w_u(A)  = P(U=u \mid A)=\frac{\pi_u \prod_{v=1}^3 g_{v,u}(A_v)}{\sum_{k=1}^K \pi_k \prod_{v=1}^3 g_{v,k}(A_v)}.
\end{equation}

\textbf{Step 2.} Next, we identify the outcome model parameters. We observe $$\mathbb{E}[Y \mid A=a] = \sum_{u=1}^K P(U=u \mid A=a)\cdot\mathbb{E}[Y \mid A=a,  U=u].$$
Substituting the outcome model in \Cref{ass:outcome-model-tr}:$$\mathbb{E}[Y \mid A=a] = \sum_{u=1}^K {w}_u(a) (\beta_u^\top \xi(a)).$$
Define the vectors $\boldsymbol{\beta} \in \mathbb{R}^{K \cdot M}$ and $\Xi(A) \in \mathbb{R}^{K \cdot M}$ as:
\begin{equation}
\label{eq:params-2-tr}
    \boldsymbol{\beta} = \begin{bmatrix} \beta_1 \\ \vdots \\ \beta_K \end{bmatrix}, \quad
\Xi(A) = \begin{bmatrix} {w}_1(A) \xi(A) \\ \vdots \\ {w}_K(A) \xi(A) \end{bmatrix}.
\end{equation}
The mixture equation simplifies to a standard linear model structure:$$\mathbb{E}[Y \mid A] = \boldsymbol{\beta}^\top \Xi(A).$$
We have, $$ \mathbb{E}[\Xi(A)Y]=\mathbb{E}[\Xi(A)\mathbb{E}[Y\mid A]]=\mathbb{E}[\Xi(A)\Xi(A)^\top] \boldsymbol{\beta}.$$
The explicit closed-form solution is given by:
$$\boldsymbol{\beta} = \left( \mathbb{E} \left[ \Xi(A) \Xi(A)^\top \right] \right)^{-1} \mathbb{E}\left[ \Xi(A) Y \right].$$
This solution exists and is unique because $ \mathbb{E}\left[ \Xi(A) \Xi(A)^\top \right]$ is positive definite (using  \Cref{ass:invertibility-tr}, since rows of $\Xi(a)$ are constant multiples of the rows of $\Xi^*(a)$).
Thus, $\{\beta_u\}_u$ is uniquely identified (up to permutation).

\textbf{Step 3.} Finally, we conclude with the identification of causal estimands of interest. Since the set $U$ is the parent set for $A$ in our DAG, it works as a valid adjustment set for the average treatment effect (ATE). Therefore,
\begin{align*}
    \tau(a)= \mathbb E^{do(A=a)}[Y ]= \sum_{u=1}^K\pi_u\mathbb E[Y \mid A=a, U=u]=\sum_{u=1}^K\pi_u  \beta_u^\top\xi(a). 
\end{align*}
Since $\{\pi_u\}_{u=1}^K$ and $\{\beta_u\}_{u=1}^K$ are uniquely identified up to permutations, as shown in stages 1 and 2 respectively, $\tau(a)$ is uniquely identifiable.
However, the \emph{conditional average treatment effect} (CATE):
\begin{align*}
   \tau_u(a)&= \mathbb E^{do(A=a)}[Y \mid U=u]=\mathbb E[Y \mid A=a, U=u]=\beta_u^\top\xi(a)
\end{align*}
   is uniquely identified only up to permutations.
\end{proof}

Now we will prove \Cref{thm:posterior-1-error}.
The main ingredient is the following result.
\begin{theorem}[\cite{song2014nonparametric}]
\label{thm:song-et-al}
\textit{Pick any $\delta \in (0,1)$. When the number of samples $n$ satisfies}
\[
n > \frac{\theta \rho^2_0 \log \frac{\delta}{2}}{\sigma_K^2(\mathcal{C}_{Z_1 Z_2})}, \quad \theta := \max \left( \frac{C_3 K^2 \rho_0}{\sigma_K(\mathcal{C}_{Z_1 Z_2})}, \frac{C_4 K^{2/3}}{\pi_{\min}^{1/3}} \right),
\]
\textit{for some constants $C_3, C_4 > 0$, and the number of iterations $N$ and the number of random initialization vectors $L$ (drawn uniformly on the sphere $\mathcal{S}^{K-1}$) satisfy}
\[
N \geq C_2 \cdot \left( \log(K) + \log \log \left( \frac{1}{\sqrt{\pi_{\min}} \epsilon_T} \right) \right),
\]
\textit{for constant $C_2 > 0$ and $L = \text{poly}(K) \log(1/\delta)$, the robust power method in \cite{anandkumar2014tensor} yields eigen-pairs $(\widehat{\lambda}_i, \widehat{v}_i)$ such that there exists a permutation $\eta$, with probability $1 - 4\delta$, we have}
\[
\|\pi_j^{-1/2} \mu_{Z|U=j} - \widehat{\lambda}_{\eta(j)} \hat\mu_{Z|U=\eta(j)}\| \leq 8\epsilon_T \cdot \pi_j^{-1/2},
\]
\[
|\pi_j^{-1/2} - \widehat{\lambda}_{\eta(j)}| \leq 5\epsilon_T, \quad \forall j \in [K],
\]
\textit{and}
\[
\left\| \mathcal{T} - \sum_{j=1}^K \widehat{\lambda}_j \widehat{\phi}_j^{\otimes 3} \right\| \leq 55\epsilon_T,
\]
\textit{where $\epsilon_T$ is the tensor perturbation bound}
\[
\epsilon_T := \|\widehat{\mathcal{T}} - \mathcal{T}\| \leq \frac{8\rho^{1.5}_0 \sqrt{\log \frac{\delta}{2}}}{\sqrt{n} \sigma_K^{1.5}(\mathcal{C}_{Z_1 Z_2})} + \frac{512\sqrt{2}\rho^3_0 \left(\log \frac{\delta}{2}\right)^{1.5}}{n^{1.5} \sigma_K^3(\mathcal{C}_{Z_1 Z_2}) \sqrt{\pi_{\min}}}.
\]
\end{theorem}

\begin{proof}[Proof of \Cref{thm:posterior-1-error}]

Bound the error $| \hat{w}_{\sigma(u)}(z) -  w_u(z) |$, where the posterior weight is defined as:
$$ w_u(z) = \frac{N_u(z)}{D(z)}.$$
Define, $N_u(z) = \pi_u\prod_{v=1}^3f_{v,u}(z_v)$,  $\hat{N}_u(z) = \hat \pi_u\prod_{v=1}^3\hat f_{v,u}(z_v)$, $D(z) = \sum_{k=1}^K N_k(z)$ and $\hat{D}(z) = \sum_{k=1}^K \hat{N}_k(z)$. Now,
\begin{equation}
\label{eq:delta-nu}
    | \hat{N}_{\sigma(u)}(z) - N_u(z) | \leq m_1^3|\hat{\pi}_{\sigma(u)}-\pi_u|+m_1^2\sum_{v=1}^3|\hat f_{v,{\sigma(u)}}(z_v)- f_{v,u}(z_v)|.
\end{equation}
Now,
\begin{align*}
   | \hat{w}_{\sigma(u)}(z) - w_u(z) |&=\frac{\hat N_{\sigma(u)}(z)}{\hat D(z)}-\frac{N_u(z)}{D(z)}\\
    &\le \frac{|\hat{N}_{\sigma(u)}(z) - N_u(z)|}{{D}(z)} + \frac{\hat N_{\sigma(u)}(z)}{\hat D(z)} \frac{|\hat{D}(z) - D(z)|}{{D}(z)}\\
    &\le  \frac{| \hat{N}_{\sigma(u)}(z) - N_u(z) | +\sum_{k=1}^K| \hat{N}_{\sigma(k)}(z) - N_k(z) | }{{D}(z)}\\
       &\le \frac{2}{m_0^3} \sum_{k=1}^K| \hat{N}_{\sigma(k)}(z) - N_k(z) | 
\end{align*}
Therefore, we have
\begin{equation}
\label{eq:bound-dw}
   \max_{u=1,\cdots,K}\sup_z | \hat{w}_{\sigma(u)}(z) - w_u(z) |\leq \frac{2}{m_0^3} \sum_{k=1}^K \sup_z| \hat{N}_{\sigma(k)}(z) - N_k(z) |.
\end{equation}
It follows from \Cref{thm:song-et-al} that under the assumptions of the theorem, for large enough $n$, with probability $1-\delta,$
\[ |\pi_u^{-1/2} - \hat\pi_{\sigma(u)}^{-1/2}| \leq \frac{40\rho^{1.5}_0 \sqrt{\log \frac{\delta}{2}}}{\sqrt{n} \sigma_K^{1.5}(\mathcal{C}_{Z_1 Z_2})} + \frac{2560\sqrt{2}\rho^3_0 \left(\log \frac{\delta}{2}\right)^{1.5}}{n^{1.5} \sigma_K^3(\mathcal{C}_{Z_1 Z_2}) \sqrt{\pi_{\min}}}.\]
and
\begin{align*}
    \|\mu_{Z|U=u} - \hat\mu_{Z|U=\sigma(u)}\|&\leq\| \mu_{Z|U=j} - \pi_j^{1/2}\widehat{\lambda}_{\eta(j)} \hat\mu_{Z|U=\eta(j)}\|+\|(\pi_j^{1/2}\widehat{\lambda}_{\eta(j)}-1) \hat\mu_{Z|U=\eta(j)}\|\\
    & \leq(8+5\pi_j^{1/2}\|\hat\mu_{Z|U=\eta(j)}\|)\left(\frac{8\rho^{1.5}_0 \sqrt{\log \frac{\delta}{2}}}{\sqrt{n} \sigma_K^{1.5}(\mathcal{C}_{Z_1 Z_2})} + \frac{512\sqrt{2}\rho^3_0 \left(\log \frac{\delta}{2}\right)^{1.5}}{n^{1.5} \sigma_K^3(\mathcal{C}_{Z_1 Z_2}) \sqrt{\pi_{\min}}}\right)\\
&\leq(8+5\rho_0)\left(\frac{8\rho^{1.5}_0 \sqrt{\log \frac{\delta}{8}}}{\sqrt{n} \sigma_K^{1.5}(\mathcal{C}_{Z_1 Z_2})} + \frac{512\sqrt{2}\rho^3_0 \left(\log \frac{\delta}{8}\right)^{1.5}}{n^{1.5} \sigma_K^3(\mathcal{C}_{Z_1 Z_2}) \sqrt{\pi_{\min}}}\right).
\end{align*}
Let, $\tilde f_u(x)=\langle\mu_{Z|u},\rho(x)\rangle$. So, the estimated density $\hat f_u(x)=\langle\hat\mu_{Z|u},\rho(x)\rangle$ has error
\begin{align*}
    |\hat f_{\sigma(u)}(x)-f_u(x)|&\leq|\hat f_{\sigma(u)}(x)-\tilde f_u(x)|+|\tilde f_u(x)-f_u(x)|.
\end{align*}
Now, $\rho_0=c^\prime s^{-d}$ and
\begin{align*}
    \sup_{x}|\hat f_{\sigma(u)}(x)-\tilde f_u(x)|&=\sup_{x}\langle\hat\mu_{Z|{\sigma(u)}}-\mu_{Z|u},\rho(x)\rangle\\
    &\leq \rho_0\|\hat\mu_{Z|{\sigma(u)}}-\mu_{Z|u}\|\\
    &\leq c^\prime s^{-d}(8+5c^\prime s^{-d})\left(\frac{8(c^\prime s^{-d})^{1.5}\sqrt{\log \frac{\delta}{8}}}{\sqrt{n} \sigma_K^{1.5}(\mathcal{C}_{Z_1 Z_2})} + \frac{512\sqrt{2}(c^\prime s^{-d})^3 \left(\log \frac{\delta}{8}\right)^{1.5}}{n^{1.5} \sigma_K^3(\mathcal{C}_{Z_1 Z_2}) \sqrt{\pi_{\min}}}\right).
\end{align*}

Proposition 1.2 of \cite{tsybakov2008nonparametric} shows that when the true density belongs to $\mathcal{H}(b, l, m_0, m_1)$ and the kernel is of order $\floor{b}$, the bias term is bounded as
\begin{equation}
    \sup_x|\tilde f_u(x)-f_u(x)|\leq c^{\prime\prime} s^b,
\end{equation}
where $s$ is the bandwidth of the kernel and $c^{\prime\prime} $ is a constant depending only on the kernel and $b, l$. Since $s=cn^{-\frac{1}{2b+7d}}$, we have
\begin{align*}
   \sup_x |\hat f_{\sigma(u)}(x)-f_u(x)|
   &\leq\left(\frac{8(c^\prime )^{2.5}(8(cn^{-\frac{1}{2b+7d}})^{d}+5c^{\prime})\sqrt{\log \frac{\delta}{8}}}{c^{3.5d}\sigma_K^{1.5}(\mathcal{C}_{Z_1 Z_2})} +c^{\prime\prime}\right)n^{-\frac{b}{2b+7d}}\\
   &+\frac{512\sqrt{2}(c^\prime )^{4}(8(cn^{-\frac{1}{2b+7d}})^{d}+5c^{\prime})\left(\log \frac{\delta}{8}\right)^{1.5}}{c^{5d}\sigma_K^{3}(\mathcal{C}_{Z_1 Z_2})\sqrt{\pi_{\min}}} \cdot n^{-\frac{6b+11d}{2b+7d}},
\end{align*}
and
\begin{align*}
  |\pi_u - \hat\pi_{{\sigma(u)}}|&=\sqrt{\pi_u\hat\pi_{{\sigma(u)}}}\cdot |\pi_u^{1/2} + \hat\pi_{{\sigma(u)}}^{1/2}|\cdot  |\pi_u^{-1/2} - \hat\pi_{\sigma(u)}^{-1/2}|\\
  &\leq \frac{80(c^\prime s^{-d})^{1.5} \sqrt{\log \frac{\delta}{8}}}{\sqrt{n} \sigma_K^{1.5}(\mathcal{C}_{Z_1 Z_2})} + \frac{5120\sqrt{2}(c^\prime s^{-d})^3 \left(\log \frac{\delta}{8}\right)^{1.5}}{n^{1.5} \sigma_K^3(\mathcal{C}_{Z_1 Z_2}) \sqrt{\pi_{\min}}}\\
  &=\frac{80(c^\prime /c^d)^{1.5} \sqrt{\log \frac{\delta}{8}}}{ \sigma_K^{1.5}(\mathcal{C}_{Z_1 Z_2})}n^{-\frac{b+2d}{2b+7d}} + \frac{5120\sqrt{2}(c^\prime /c^{d})^3 \left(\log \frac{\delta}{8}\right)^{1.5}}{\sigma_K^3(\mathcal{C}_{Z_1 Z_2}) \sqrt{\pi_{\min}}}n^{-\frac{3(2b+5d)}{2(2b+7d)}}.
\end{align*}
Therefore, from \eqref{eq:delta-nu},
\begin{align*}
 | \hat{N}_{\sigma(u)}(z) - N_u(z) | &\leq  \frac{80m_1^3(c^\prime /c^d)^{1.5} \sqrt{\log \frac{\delta}{8}}}{ \sigma_K^{1.5}(\mathcal{C}_{Z_1 Z_2})}n^{-\frac{b+2d}{2b+7d}} + \frac{5120\sqrt{2}m_1^3(c^\prime /c^{d})^3 \left(\log \frac{\delta}{8}\right)^{1.5}}{\sigma_K^3(\mathcal{C}_{Z_1 Z_2}) \sqrt{\pi_{\min}}}n^{-\frac{3(2b+5d)}{2(2b+7d)}}\\
  &+ \left(\frac{8m_1^2(c^\prime )^{2.5}(8(cn^{-\frac{1}{2b+7d}})^{d}+5c^{\prime})\sqrt{\log \frac{\delta}{8}}}{c^{3.5d}\sigma_K^{1.5}(\mathcal{C}_{Z_1 Z_2})} +c^{\prime\prime}\right)n^{-\frac{b}{2b+7d}}\\
   &+\frac{512\sqrt{2}m_1^2(c^\prime )^{4}(8(cn^{-\frac{1}{2b+7d}})^{d}+5c^{\prime})\left(\log \frac{\delta}{8}\right)^{1.5}}{c^{5d}\sigma_K^{3}(\mathcal{C}_{Z_1 Z_2})\sqrt{\pi_{\min}}} \cdot n^{-\frac{6b+11d}{2b+7d}}.
\end{align*}

For large $n$, the third term dominates. So, for sufficiently large $n$,
\[ \max_{u} \sup_z| \hat{N}_{\sigma(u)}(z) - N_u(z) | \leq \left(\frac{9m_1^2(c^\prime )^{2.5}(8(cn^{-\frac{1}{2b+7d}})^{d}+5c^{\prime})\sqrt{\log \frac{\delta}{8}}}{c^{3.5d}\sigma_K^{1.5}(\mathcal{C}_{Z_1 Z_2})} +c^{\prime\prime}\right)n^{-\frac{b}{2b+7d}}.\]
Substituting the above inequality in \eqref{eq:bound-dw}, we get
\[ \max_{u=1,\cdots,K}\sup_z| \hat{w}_{\sigma(u)}(z) - w_u(z) |\leq\left(\frac{18Km_1^2(c^\prime )^{2.5}(8(cn^{-\frac{1}{2b+7d}})^{d}+5c^{\prime})\sqrt{\log \frac{\delta}{8}}}{m_0^3c^{3.5d}\sigma_K^{1.5}(\mathcal{C}_{Z_1 Z_2})} +c^{\prime\prime}\right)n^{-\frac{b}{2b+7d}}.\]
\end{proof}

\begin{proof}[Proof of \Cref{thm:alpha-error-bound}]
We use $\hat{\Phi}_i$ and ${\Phi}_i$ as the shorthand for $\hat{\Phi}(Z^{(i)})$ and ${\Phi}(Z^{(i)})$ respectively.
Let $\hat{\mathbf{G}} = \frac{1}{n} \sum_{i=1}^n \hat{\Phi}_i \hat{\Phi}_i^\top$ and ${\mathbf{G}} = \frac{1}{n} \sum_{i=1}^n {\Phi}_i {\Phi}_i^\top$. We decompose the error as 
$$\|\hat{\boldsymbol{\alpha}} - \boldsymbol{\alpha}\|_{2}\leq\|\hat{\boldsymbol{\alpha}} - \hat{\boldsymbol{\alpha}}^*\|_{2}+\|{\boldsymbol{\alpha}} - \hat{\boldsymbol{\alpha}}^*\|_{2},$$
where $\hat{\boldsymbol{\alpha}}^*={\mathbf{G}}^{-1} \left( \frac{1}{n} \sum_{i=1}^n {\Phi}_i A^{(i)} \right)$.

We define four events:
\begin{itemize}
    \item $\mathcal{E}_{W}$: $\max_{i=1,\cdots,n}\max_{k=1,\cdots,K}|\hat{w}_u(Z^{(i)}) - w_u(Z^{(i)})| \le \ew $ (this has probability $1-\delta/4$ by assumption)
    \item $\mathcal{E}_{A}$: $\frac{ 1}{n}\sum_{i=1}^n  |A^{(i)}|^2\leq  C_{A,n}$, where $C_{A,n}=\E\left(A^2\right)+ \max \left( \sqrt{\frac{32\sigma^2_A \log(4/\delta)}{n}}, \frac{32\sigma^2_A \log(4/\delta)}{n} \right)$
    (this has probability $1-\delta/4$ by sub-exponential bound, since $A$ is $\sigma^2_A$-sub-gaussian implies $A^2$ is $(16\sigma^2_A,16\sigma^2_A)$-sub-exponential.)
    \item $\mathcal{E}_{\phi}$: $\frac{ 1}{n}\sum_{i=1}^n  \| \phi(Z^{(i)}) \|^2\leq  C_{\phi,n}$, where $C_{\phi,n}= \E\left(\| \phi(Z) \|^2\right)+ \max \left( \sqrt{\frac{32\sigma^2_\phi \log(4/\delta)}{n}}, \frac{32\sigma^2_\phi \log(4/\delta)}{n} \right)$
    (this has probability $1-\delta/4$ by sub-exponential bound, since $\|\phi(Z^{(i)}) \|$ is $\sigma^2_\phi$-sub-gaussian implies $\|\phi(Z^{(i)}) \|^2$ is $(16\sigma^2_\phi,16\sigma^2_\phi)$-sub-exponential.)
    \item $\mathcal{E}_{G}$: $\| \Sigma_\phi - \G \|\leq  \delta_{G,n}$, where $\delta_{G,n}=\frac{LK}{n}+(1+ 8\rho_\Phi \|\mu_\Phi\|_2)\sqrt{\frac{LK}{n}}+\max\left( \sqrt{\frac{\ln(8c_1/\delta)}{c_2 n}}, \frac{\ln(8c_1/\delta)}{c_2 n} \right)+8\rho_\Phi \|\mu_\Phi\|_2 \cdot \sqrt{\frac{ \log(8/\delta)}{n}}$ (this has probability $1-\delta/4$ by Lemma \ref{lem:matrx-conc})
\end{itemize}
For the second term using Lemma \ref{lem:lin-reg}, we have with probability at least $1-\delta$:
\begin{equation}
 \label{bdd:alpha-err-2}
  \|{\boldsymbol{\alpha}} - \hat{\boldsymbol{\alpha}}^*\|\leq \frac{\sigma_{1,noise}}{\sqrt{n(1-\eta_{1,n,\delta})\lambda_\Phi}}\sqrt{LK + 2\sqrt{LK \log(3/\delta)}+2\log(3/\delta)+\frac{2\log(3/\delta)}{1-\eta_{1,n,\delta}}},
\end{equation}
where $ \eta_{1,n,\delta}= \frac{LK}{n}+(1+ 8 \rho_\Phi \| \Sigma^{-1/2}_\Phi\mu_\Phi\|_2)\sqrt{\frac{LK}{n}}+\max\left( \sqrt{\frac{\ln(3c_1/\delta)}{c_2 n}}, \frac{\ln(3c_1/\delta)}{c_2n} \right)+8\rho_\Phi \| \Sigma^{-1/2}_\Phi\mu_\Phi\|_2 \sqrt{\frac{ \log(3/\delta)}{n}},$ and $c_1, c_2 > 0$ are universal constants.

We bound the first term using Lemma \ref{lem:reg-error-2}. For that, we only need to find expressions for $\epsilon_{n,\delta,2}$, $\epsilon_{n,\delta,3}$ and $B_{\delta,n}$. We have already derived that $\epsilon_{n,\delta,1}$ is $\delta_{G,n}$ in this case.

\paragraph{Step 1: Expression for $\epsilon_{n,\delta,2}$}
The vector $\hPhi_i$ is composed of blocks $\hat{w}_{k}(Z^{(i)}) \phi(Z^{(i)})$. The oracle vector $\PhiStar_i$ has blocks ${w}_{k}(Z^{(i)}) \phi(Z^{(i)})$.
The squared Euclidean error is:
\[ \| \hPhi_i - \PhiStar_i \|^2 = \sum_{k=1}^K \| (\hat{w}_{k}(Z^{(i)}) - {w}_{k}(Z^{(i)})) \phi(Z^{(i)}) \|^2. \]
Using the uniform weight error $\ew$, on $\mathcal{E}_{W}$:
 \begin{equation}
 \label{eq:phi-bound}
      \| \hPhi_i - \PhiStar_i \| \le \ew\sqrt K  \| \phi(Z^{(i)}) \|.
 \end{equation}
and since $\sum_{k=1}^K \hat{w}_{k}(z)=1=\sum_{k=1}^K {w}_{k}(z)$,
\[ \| \hPhi_i  \|^2=\sum_{k=1}^K \| \hat{w}_{k}(Z^{(i)}) \phi(Z^{(i)}) \|^2 \le  \|  \phi(Z^{(i)}) \|^2  \]
and \[ \| \Phi_i  \|^2=\sum_{k=1}^K \| {w}_{k}(Z^{(i)}) \phi(Z^{(i)}) \|^2 \le  \|  \phi(Z^{(i)}) \|^2 . \]

We use the algebraic identity for rank-1 perturbations:
\[ \hPhi \hPhi^\top - \PhiStar \PhiStar^\top = (\hPhi - \PhiStar)(\hPhi - \PhiStar)^\top + (\hPhi - \PhiStar)\PhiStar^\top + \PhiStar(\hPhi - \PhiStar)^\top. \]
Applying the triangle inequality for the spectral norm,  on $\mathcal{E}_{W}$:
\begin{align*}
    \| \hPhi_i \hPhi_i^\top - \PhiStar_i \PhiStar_i^\top \| &\le \| \hPhi_i - \PhiStar_i \|^2 + 2 \| \hPhi_i - \PhiStar_i \| \| \PhiStar_i \| \\
    &\le\ew^2K  \| \phi(Z^{(i)}) \|^2 + 2 \ew K  \| \phi(Z^{(i)}) \|^2 \\
    &\le 3 \ew K  \| \phi(Z^{(i)}) \|^2 .
\end{align*}
Averaging over $n$ samples,  on $\mathcal{E}_{W}\cap \mathcal{E}_{\phi}$:
\begin{equation} 
    \| \hG - \G \| \le\frac{ 3 \ew K}{n}\sum_{i=1}^n  \| \phi(Z^{(i)}) \|^2\leq  3 \ew K C_{\phi,n}.
\end{equation}

\paragraph{Step 2: Expression for $\epsilon_{n,\delta,3}$ }

We bound the norm by factoring out the uniform feature error and the average magnitude of the treatment:
\begin{align*}
  \left\| \frac{1}{n}\sum_{i=1}^n (   \PhiStar_iA^{(i)}-\hPhi_iA^{(i)})\right\|  &\leq  \frac{1}{n}\sum_{i=1}^n  \| \hPhi_i - \PhiStar_i \| |A^{(i)}| \\
    &\le  \ew\sqrt K  \sum_{i=1}^n \frac{1}{n} \| \phi(Z^{(i)}) \| |A^{(i)}|\\
    &\le   \frac{\ew \sqrt K }{2}     \left( \frac{1}{n}\sum_{i=1}^n\|\phi(Z^{(i)}) \|^2 +\frac{1}{n}\sum_{i=1}^n|A^{(i)}|^2 \right)\\
    &\le \frac{\sqrt K }{2(\lambda_0-\delta_{G,n})}(C_{\phi,n}+C_{A,n})\ew,
\end{align*}
on $\mathcal{E}_{W}\cap \mathcal{E}_{A}\cap \mathcal{E}_{\phi}\cap\mathcal{E}_{G}$.

\paragraph{Step 3: Expression for $B_{n,\delta}$}
On $\mathcal{E}_{A}\cap \mathcal{E}_{\phi}$, we have
\begin{align*}
 \frac{1}{n} \sum_{i=1}^n \| \hPhi_i \| |A^{(i)}| 
    &\le  \frac{1}{2}  \left( \frac{1}{n}\sum_{i=1}^n\|\phi(Z^{(i)}) \|^2 +\frac{1}{n}\sum_{i=1}^n|A^{(i)}|^2 \right)\\
    & \le\frac{1}{2} (C_{\phi,n}+C_{A,n}),
\end{align*}

Using Lemma \ref{lem:reg-error-2} and substituting $\epsilon_{n,\delta,1}$, $\epsilon_{n,\delta,2}$, $\epsilon_{n,\delta,3}$ and $B_{\delta,n}$ by $\delta_{G,n},$ $3 \ew K C_{\phi,n}$, $\frac{\sqrt K }{2(\lambda_0-\delta_{G,n})}(C_{\phi,n}+C_{A,n})\ew$ and $\frac{1}{2} (C_{\phi,n}+C_{A,n})$, we derive that on $\mathcal{E}_{W}\cap \mathcal{E}_{A}\cap \mathcal{E}_{\phi}\cap\mathcal{E}_{G}$
(which has probability at least $1 - \delta$):
\begin{equation}
\label{bdd:alpha-err-1}
    \| \halpha - \talpha \| \le \frac{\sqrt K }{2(\lam-\delta_{G,n})} {\left[1+ \frac{3  K C_{\phi,n}}{\lam-\delta_{G,n}- 3 \ew K C_{\phi,n}} \right]} (C_{\phi,n}+C_{A,n})\ew.
\end{equation}

Combining \eqref{bdd:alpha-err-1} and \eqref{bdd:alpha-err-2}, we get the desired bound.
\end{proof}

\begin{proof}[Proof of \Cref{thm:sigma-est-error}]
Let $\hat{\mathbf{D}} = \frac{1}{n} \sum_{i=1}^n \hat{{W}}_i \hat{{W}}_i^\top$ and ${\mathbf{D}} = \frac{1}{n} \sum_{i=1}^n {{W}}_i {{W}}_i^\top$. We decompose the error as 
$$\|\hat{\boldsymbol{\sigma^2}} - \boldsymbol{\sigma^2}\|\leq\|\hat{\boldsymbol{\sigma^2}} - \tsigma\|+\|{\boldsymbol{\sigma^2}} - \tsigma\|,$$
where $\tsigma={\mathbf{D}}^{-1} \left( \frac{1}{n} \sum_{i=1}^n {W}_i C^{(i)} \right)$. 

We decompose the first term as:
\begin{align*}
    \hsigma - \tsigma &= \hD^{-1} \left( \frac{1}{n} \sum \hW_i \hC^{(i)} \right) - \D^{-1} \left( \frac{1}{n} \sum \W_i \C^{(i)} \right) \\
    &= \underbrace{\D^{-1} \left( \frac{1}{n} \sum (\hW_i \hC^{(i)} - \W_i \C^{(i)}) \right)}_{T_1 } + \underbrace{(\hD^{-1} - \D^{-1}) \left( \frac{1}{n} \sum \hW_i \hC^{(i)} \right)}_{T_2}.
\end{align*}

We define four events:
\begin{itemize}
    \item $\mathcal{E}_{W}$: $\max_{i=1,\cdots,n}\max_{k=1,\cdots,K}|\hat{w}_u(Z^{(i)}) - w_u(Z^{(i)})| \le \ew $ (this has probability $1-\delta/4$ by assumption)
     \item $\mathcal{E}_{\alpha}$: $ \|{\boldsymbol{\alpha}} - \hat{\boldsymbol{\alpha}}\|\leq \ea$ (this has probability $1-\delta/4$ by assumption)
    \item $\mathcal{E}_{\phi}$: $\frac{ 1}{n}\sum_{i=1}^n  \| \phi(Z^{(i)}) \|^2\leq  C_{\phi,n}$
    (this has probability $1-\delta/4$ by sub-exponential bound, since $\|\phi(Z^{(i)}) \|$ is $\sigma^2_\phi$-sub-gaussian implies $\|\phi(Z^{(i)}) \|^2$ is $(16\sigma^2_\phi,16\sigma^2_\phi)$-sub-exponential.)
    \item $\mathcal{E}_{D}$: $\| \Sigma_W - \D \|\leq  \delta_{W,n}$, where $\delta_{W,n}=\frac{K}{n}+(1+ 4\sqrt{K} \|\mu_W\|_2)\sqrt{\frac{K}{n}}+\max\left( \sqrt{\frac{\ln(8c_3/\delta)}{c_4 n}}, \frac{\ln(8c_3/\delta)}{c_4 n} \right)+4\sqrt{K}\|\mu_W\|_2 \cdot \sqrt{\frac{ \log(8/\delta)}{n}}$ (this has probability $1-\delta/4$ by Lemma \ref{lem:matrx-conc}, since $W(Z)$ is a sub-Gaussian vector with proxy variance $K/4$ since each component is bounded in $[0,1]$.)
\end{itemize}
For the second term using Lemma \ref{lem:lin-reg}, we have with probability at least $1-\delta$:
\begin{equation}
 \label{bdd:sigma-err-2}
  \|{\boldsymbol{\sigma}} - \hat{\boldsymbol{\sigma}}^*\|\leq \frac{\sigma_{2,noise}}{\sqrt{n(1-\eta_{2,n,\delta})\lambda_W}}\sqrt{K + 2\sqrt{K \log(3/\delta)}+2\log(3/\delta)+\frac{2\log(3/\delta)}{1-\eta_{2,n,\delta}}},
\end{equation}
where $ \eta_{2,n,\delta}= \frac{K}{n}+(1+ 4\sqrt{K}  \| \Sigma^{-1/2}_W\mu_W\|_2)\sqrt{\frac{K}{n}}+\max\left( \sqrt{\frac{\ln(3c_3/\delta)}{c_4 n}}, \frac{\ln(3c_3/\delta)}{c_4n} \right)+4\sqrt{K} \| \Sigma^{-1/2}_W\mu_W\|_2 \sqrt{\frac{ \log(3/\delta)}{n}},$ and $c_3, c_4 > 0$ are universal constants.

We bound the first term using Lemma \ref{lem:reg-error-2}. For that, we only need to find expressions for $\epsilon_{n,\delta,2}$, $\epsilon_{n,\delta,3}$ and $B_{\delta,n}$. We have already derived that $\epsilon_{n,\delta,1}$ is $\delta_{G,n}$ in this case.

\paragraph{Step 1: Expression for $\epsilon_{n,\delta,2}$}

Using the algebraic identity for rank-1 perturbations:
\[ \hW \hW^\top - \W \W^\top = (\hW - \W)(\hW - \W)^\top + (\hW - \W)\W^\top + \W(\hW - \W)^\top. \]
Applying the triangle inequality for the spectral norm:
\begin{align*}
    \| \hW_i \hW_i^\top - \W_i \W_i^\top \| &\le \| \hW_i - \W_i \|^2 + 2 \| \hW_i - \W_i \| \| \W_i \| \\
    &\le K\ew + 2 \sqrt{K}\ew\cdot \sqrt{K}=3K\ew.
\end{align*}
Averaging over $n$ samples, on $\mathcal{E}_{W}$:
\[ \| \hD- \D \| \le  3K\ew. \]

\paragraph{Step 2: Error bound for the response variable}
We have $$\hat C^{(i)}=\sum_{u=1}^K \hat w_u(Z^{(i)})(\hat\alpha_u^\top \phi(Z^{(i)}))^2 - \left(\sum_{u=1}^K\hat w_u(Z^{(i)})\cdot\hat\alpha_u^\top \phi(Z^{(i)})\right)^2$$ and the oracle is $$C^{(i)}=\sum_{u=1}^K  w_u(Z^{(i)})(\alpha_u^\top \phi(Z^{(i)}))^2 - \left(\sum_{u=1}^K w_u(Z^{(i)})\cdot\alpha_u^\top \phi(Z^{(i)})\right)^2.$$
Note that $|\hat{C}^{(i)} - C^{(i)}| \le |\hat{E}_2^{(i)} - E_2^{(i)}| + |(\hat{E}_1^{(i)})^2 - (E_1^{(i)})^2|,$ where
\begin{align*}
    E_1^{(i)}&=\sum_{u=1}^K w_u(Z^{(i)})\cdot\alpha_u^\top \phi(Z^{(i)}),\\
    \hat{E}_1^{(i)}&=\sum_{u=1}^K\hat w_u(Z^{(i)})\cdot\hat\alpha_u^\top \phi(Z^{(i)}),\\
     E_2^{(i)}&=\sum_{u=1}^K  w_u(Z^{(i)})(\alpha_u^\top \phi(Z^{(i)}))^2,\\
    \hat{E}_2^{(i)}&=\sum_{u=1}^K \hat w_u(Z^{(i)})(\hat\alpha_u^\top \phi(Z^{(i)}))^2.
\end{align*}

\begin{align*}
     |\hat{E}_1^{(i)} - E_1^{(i)}| &\le \sum_{u=1}^K | \hat{w}_u(Z^{(i)}) - w_u(Z^{(i)}) | |\alpha_u^{\top} \phi(Z^{(i)})| + \sum_{u=1}^K \hat{w}_u(Z^{(i)})| |(\hat\alpha_u-\alpha_u)^{\top} \phi(Z^{(i)})| \\
    &\le  \ew B_\alpha K \| \phi(Z^{(i)})\| +  \|\hat{\alpha}_u - \alpha_u^*\| \|\phi(Z^{(i)})\| \quad (\text{since } \sum_u \hat w_u = 1)\\
    &\le (K B_\alpha\ew +  \ea)\|\phi(Z^{(i)})\|.
\end{align*}
\begin{align*}
   |\hat{E}_1^{(i)} + E_1^{(i)}| &\le \sum_{u=1}^K | \hat{w}_u(Z^{(i)}) + w_u(Z^{(i)}) | |\alpha_u^{\top} \phi(Z^{(i)})| + \sum_{u=1}^K \hat{w}_u(Z^{(i)}) |(\hat\alpha_u+\alpha_u)^{\top} \phi(Z^{(i)})| \\
    &\le  2 B_\alpha K \| \phi(Z^{(i)})\| + \sum_{u=1}^K \hat{w}_u(Z^{(i)})  (\|\hat{\alpha}_u - \alpha_u\|+2\| \alpha_u\|)\| \phi(Z^{(i)})\| \\
    &\le (2 B_\alpha K +  \ea+2B_\alpha)\|\phi(Z^{(i)})\| \quad (\text{since } \sum_u \hat w_u = 1).
\end{align*}
\begin{align*}
   & |\hat{E}_2^{(i)} - E_2^{(i)}|\\
   &\le  \sum_{u=1}^K | \hat{w}_u(Z^{(i)}) - w_u(Z^{(i)}) | |\alpha_u^{\top} \phi(Z^{(i)})|^2 + \sum_{u=1}^K \hat{w}_u(Z^{(i)})| |(\hat\alpha_u-\alpha_u)^{\top} \phi(Z^{(i)})|\cdot|(\hat\alpha_u+\alpha_u)^{\top} \phi(Z^{(i)})|  \\
    &\le \ew B_\alpha^2 K \| \phi(Z^{(i)})\|^2 + \sum_{u=1}^K \hat{w}_u(Z^{(i)})| \|\hat{\alpha}_u - \alpha_u\| (\|\hat{\alpha}_u - \alpha_u\|+2\| \alpha_u\|)\|\phi(Z^{(i)})\|^2\\
    &\le (\ew B_\alpha^2 K  + \ea(\ea+2B_\alpha) )\| \phi(Z^{(i)})\|^2.
\end{align*}
Combining these three bounds, on $\mathcal{E}_{W}\cap\mathcal{E}_{\alpha}$,
\begin{align*}
    &|\hat{C}^{(i)} - C^{(i)}|\\
    &\le |\hat{E}_2^{(i)} - E_2^{(i)}| + |\hat{E}_1^{(i)} - E_1^{(i)}|\cdot|\hat{E}_1^{(i)}+E_1^{(i)}| \\
    &\le(\ew B_\alpha^2 K  + \ea(\ea+2B_\alpha) + (K B_\alpha\ew +  \ea)(2 B_\alpha (K+1) +  \ea))\|\phi(Z^{(i)})\|^2 \\
    &=\left( B_\alpha^2 (3K+2K^2)\ew  + 2\ea^2+  (3K+2+B_\alpha)\ea\right)\|\phi(Z^{(i)})\|^2
\end{align*}

\paragraph{Step 3: Expression for $\epsilon_{n,\delta,3}$}

We bound the norm of the average:
\begin{align*}
   & \left\| \frac{1}{n} \sum_{i=1}^n (\hW_i \hC^{(i)} - \W_i \C^{(i)}) \right\| \\
    &\le \frac{1}{n} \sum_{i=1}^n  \left(  \| \hW_i - \W_i \| |\C^{(i)}| +   \| \hW_i \| |\hC^{(i)} - \C^{(i)}| \right)\\
    &\le\frac{1}{n} \sum_{i=1}^n  \left( \ew |\C^{(i)}| +    \left( B_\alpha^2 (3K+2K^2)\ew  + 2\ea^2+  (3K+2+B_\alpha)\ea\right)\|\phi(Z^{(i)})\|^2 \right).
\end{align*}
Now,
$|C^{(i)}|\leq\sum_{u=1}^K  w_u(Z^{(i)})(\alpha_u^\top \phi(Z^{(i)}))^2\leq\sum_{u=1}^K  w_u(Z^{(i)})B_\alpha^2\|\phi(Z^{(i)})\|^2=B_\alpha^2\|\phi(Z^{(i)})\|^2.$ Therefore, on $\mathcal{E}_{W}\cap\mathcal{E}_{\alpha}\cap\mathcal{E}_{\phi}\cap\mathcal{E}_{D}$,
\begin{align*}
 & \left\| \frac{1}{n} \sum_{i=1}^n (\hW_i \hC^{(i)} - \W_i \C^{(i)}) \right\| \\
 &\leq\frac{\sqrt{K}}{\lambda_W-\delta_{W,n}}\left( B_\alpha^2 (3K+2K^2+1)\ew  + 2\ea^2+  (3K+2+B_\alpha)\ea\right)\frac{1}{n}\sum_{i=1}^n\|\phi(Z^{(i)})\|^2\\
    &\leq\frac{\sqrt{K}}{\lambda_W-\delta_{W,n}}\left( B_\alpha^2 (3K+2K^2+1)\ew  + 2\ea^2+  (3K+2+B_\alpha)\ea\right)C_{\phi,n}.
\end{align*}

\paragraph{Step 4: Expression for $B_{n,\delta}$}

On $\mathcal{E}_{W}\cap\mathcal{E}_{\alpha}\cap\mathcal{E}_{\phi}\cap\mathcal{E}_{D}$,
\begin{align*}
      \frac{1}{n} \sum_{i=1}^n \|\hW_i\|| \hC^{(i)}|  
    &\le \sum_{i=1}^n (\|\hC^{(i)}-\C^{(i)}\|+\|\C^{(i)}\|)\\
    &\le K^{1/2}\left( B_\alpha^2 (3K+2K^2)\ew  + 2\ea^2+  (3K+2+B_\alpha)\ea+B_\alpha\right)C_{\phi,n}.
\end{align*}

Using Lemma \ref{lem:reg-error-2} and substituting $\epsilon_{n,\delta,1}$, $\epsilon_{n,\delta,2}$, $\epsilon_{n,\delta,3}$ and $B_{\delta,n}$ by $\delta_{W,n}$ and the expressions derived in steps 1,3 and 4, we derive on $\mathcal{E}_{W}\cap\mathcal{E}_{\alpha}\cap\mathcal{E}_{\phi}\cap\mathcal{E}_{D}$ (which has probability at least $1-\delta$)
\begin{align}
\label{bdd:sigma-err-1}
 \nonumber  | \hsigma - \tsigma| &\le\frac{\sqrt{K}C_{\phi,n}}{\lambda_W-\delta_{W,n}}\left(B_\alpha^2 (3K+2K^2+1)\ew  +   (3K+2+B_\alpha+2\ea)\ea\right) \\ &+\frac{3 K^{3/2}\left( B_\alpha^2 (3K+2K^2)\ew  +   (3K+2+B_\alpha+2\ea)\ea+B_\alpha\right)C_{\phi,n}}{(\lambda_W-\delta_{W,n})(\lambda_W-\delta_{W,n}- 3K \ew)}\cdot\ew.
\end{align}

Combining \eqref{bdd:sigma-err-1} and \eqref{bdd:sigma-err-2}, we get the desired bound.
\end{proof}

\begin{proof}[Proof of \Cref{thm:wtilde-esr-error}]
Define the events
\begin{itemize}
    \item $\mathcal{E}_{W}$: $\displaystyle\max_{i=1,\cdots,n}\max_{k=1,\cdots,K}|\hat{w}_u(Z^{(i)}) - w_u(Z^{(i)})| \le \ew $ (this has probability $1-\delta/4$ by assumption)
     \item $\mathcal{E}_{\alpha}$: $ \|{\boldsymbol{\alpha}} - \hat{\boldsymbol{\alpha}}\|\leq \ea$ (this has probability $1-\delta/4$ by assumption)
      \item $\mathcal{E}_{\sigma}$: $ \|{\boldsymbol{\sigma}^2} - \hat{\boldsymbol{\sigma}^2}\|\leq \es$ (this has probability $1-\delta/4$ by assumption)
      \item $\mathcal{E}_{bound}$:
$\max_{1\leq i\leq n} \|\phi(Z^{(i)})\| \le M_\phi$ and $\max_{1\leq i\leq n} |A^{(i)}| \le M_A $ (this has probability $1-\delta/4$ because $\|\phi(Z)\|$ and $A$ are sub-gaussians)
\end{itemize}
Bound the error $| \hat{\tilde w}_u(a,z) - \tilde w_u(a,z) |$, where the posterior weight is defined as:$$\tilde w_u(a,z) = \frac{N_u(a,z)}{D(a,z)}.$$
Define,$N_u(a,z) = e_u(a,z)w_u(z)$,  $\hat{N}_u(a,z) = \hat e_u(a,z)\hat w_u(z)$, $D(a,z) = \sum_{k=1}^K N_k(a,z)$ and $\hat{D}(a,z) = \sum_{k=1}^K \hat{N}_k(a,z)$.
Let's define the error for the unnormalized term $N_u$.$$\Delta N_u(a,z) := | \hat{N}_u(a,z) - N_u(a,z) | = | \hat e_u(a,z)\hat w_u(z) - e_u(a,z)w_u(z) |.$$
We have
$$\Delta N_u(a,z) \le | \hat w_u(z) -  w_u(z) | \cdot \hat e_u(a,z) + w_u(z)\cdot | \hat e_u(a,z) - e_u(a,z) |.$$
Now,
\begin{align*}
   | \hat{\tilde w}_u(a,z) -\tilde w_u(a,z) |&=\frac{\hat N_u(a,z)}{\hat D(a,z)}-\frac{N_u(a,z)}{D(a,z)}\\
    &\le \frac{|\hat{N}_u(a,z) - N_u(a,z)|}{{D}(a,z)} + \frac{\hat N_u(a,z)}{\hat D(a,z)} \frac{|\hat{D}(a,z) - D(a,z)|}{{D}(a,z)}\\
    &\le \frac{|\Delta N_u(a,z)|}{{D}(a,z)} + \frac{\sum_{k=1}^K|\Delta N_k(a,z)|}{{D}(a,z)}\\
       &\le 2 \sum_{k=1}^K\frac{|\Delta N_k(a,z)|}{N_k(a,z)}
\end{align*}
Therefore, we need to bound each $\frac{|\Delta N_k(a,z)|}{N_k(a,z)}=\frac{\hat e_k(a,z) }{ e_k(a,z) }\cdot\left|\frac{\hat w_k(z) -  w_k(z) }{ w_k(z) }\right|+\left|\frac{\hat e_k(a,z) -  e_k(a,z) }{ e_k(a,z) }\right|$.
\begin{align}
  \left|\frac{\hat e_k(a,z) -  e_k(a,z) }{ e_k(a,z) }\right|&=\left|1-\exp\left(\frac{1}{2\sigma_k^2}(a-\alpha_k^\top\phi(z))^2-\frac{1}{2\hat\sigma_k^2}(a-\hat\alpha_k^\top\phi(z))^2\right)\right|\\
  &\leq \exp\left(\left|\frac{1}{2\sigma_k^2}(a-\alpha_k^\top\phi(z))^2-\frac{1}{2\hat\sigma_k^2}(a-\hat\alpha_k^\top\phi(z))^2\right|\right)-1.
\end{align}
On $\mathcal{E}_{\alpha}\cap\mathcal{E}_\sigma$,
\begin{align*}
    &\left|\frac{1}{2\sigma_k^2}(A-\alpha_k^\top\phi(Z))^2-\frac{1}{2\hat\sigma_k^2}(A-\hat\alpha_k^\top\phi(Z))^2\right|\\
    &\leq \left|\frac{1}{2\sigma_k^2}-\frac{1}{2\hat\sigma_k^2}\right|(A-\hat\alpha_k^\top\phi(Z))^2+ \frac{1}{2\sigma_k^2}\left|(A-\hat\alpha_k^\top\phi(Z))^2-(A-\alpha_k^\top\phi(Z))^2\right|\\
     &\leq \left|\frac{1}{\sigma_k^2}-\frac{1}{\hat\sigma_k^2}\right|(A^2+\|\hat\alpha_k\|^2\cdot\|\phi(Z)\|^2)+ \frac{1}{2\sigma_k^2}\|\alpha_k-\hat\alpha_k\|\|\phi(Z)\|\cdot(2|A|+\|\alpha_k+\hat\alpha_k\|\|\phi(Z))\|)\\
      &\leq\frac{(A^2+(\ea+B_\alpha)\|\phi(Z)\|^2)\es}{\sigma_k^2(\sigma_k^2-\es)}+ \frac{\ea}{2\sigma_k^2}\|\phi(Z)\|\cdot(2|A|+(\ea+2B_\alpha)\|\phi(Z))\|).
\end{align*}
Therefore, on $\mathcal{E}_{\alpha}\cap\mathcal{E}_\sigma\cap\mathcal{E}_{bound}$,
\begin{align*}
 & \max_{i=1,\cdots,n}\max_{k=1,\cdots,K}  \left|\frac{\hat e_k((A^{(i)},Z^{(i)}) -  e_k((A^{(i)},Z^{(i)}) }{ e_k((A^{(i)},Z^{(i)}) }\right|\\
  &\leq \exp\Bigg( \frac{\ea}{2B_\sigma}M_\phi\cdot(2M_A+(\ea+2B_\alpha)M_\phi)+\frac{(M_A^2+(\ea+B_\alpha)M_\phi^2) \es}{B_\sigma(B_\sigma-\es)}\Bigg)-1\\
    &\leq \frac{\ea}{B_\sigma}M_\phi\cdot(2M_A+(\ea+2B_\alpha)M_\phi)+\frac{2(M_A^2+(\ea+B_\alpha)M_\phi^2) \es}{B_\sigma(B_\sigma-\es)},
\end{align*}
when $n$ is large enough such that 
$\frac{\ea}{2B_\sigma}M_\phi\cdot(2M_A+(\ea+2B_\alpha)M_\phi)+\frac{(M_A^2+(\ea+B_\alpha)M_\phi^2) \es}{B_\sigma(B_\sigma-\es)}<1$. In the last inequality, we used the fact that $e^x-1\leq 2x,$ for $x\in[0,1]$. Therefore, on $\mathcal{E}_{\alpha}\cap\mathcal{E}_\sigma\cap\mathcal{E}_{bound}\cap\mathcal{E}_W$,
\begin{align*}
  &\max_{i=1,\cdots,n}\max_{k=1,\cdots,K} | \hat{\tilde w}_u((A^{(i)},Z^{(i)}) -\tilde w_u((A^{(i)},Z^{(i)}) |\\
  &\leq 2K\ew+2K(1+\ew)\cdot \Bigg[\frac{\ea}{B_\sigma}M_\phi\cdot(2M_A+(\ea+2B_\alpha)M_\phi)+\frac{2(M_A^2+(\ea+B_\alpha)M_\phi^2) \es}{B_\sigma(B_\sigma-\es)}\Bigg].
\end{align*}
\end{proof}

\begin{proof}[Proof of \Cref{thm:beta-error}]
Let $\hat{\mathbf{H}} = \frac{1}{n} \sum_{i=1}^n \hat{\Psi}_i \hat{\Psi}_i^\top$ and ${\mathbf{H}} = \frac{1}{n} \sum_{i=1}^n {\Psi}_i {\Psi}_i^\top$. 
We analyze the estimator $\hbeta$ compared to the oracle $\tbeta$:
\[ \hbeta = \hH^{-1} \left( \frac{1}{n} \sum_{i=1}^n \hPsi_i Y^{(i)} \right), \quad \tbeta = \Hmat^{-1} \left( \frac{1}{n} \sum_{i=1}^n \PsiStar_i Y^{(i)} \right). \]
We decompose the error as 
$$\|\hat{\boldsymbol{\beta}} - \boldsymbol{\beta}\|\leq\|\hat{\boldsymbol{\beta}} - \hat{\boldsymbol{\beta}}^*\|+\|{\boldsymbol{\beta}} - \hat{\boldsymbol{\beta}}^*\|.$$
We analyze the first term on the intersection of four high-probability events. 
\begin{itemize}
      \item $\mathcal{E}_{\tilde W}$:
$\max_{i=1,\cdots,n}\max_{k=1,\cdots,K} | \hat{\tilde w}_u((A^{(i)},Z^{(i)}) -\tilde w_u((A^{(i)},Z^{(i)}) |\leq \ewtilde$ (this has probability $1-\delta/4$ by assumption)
\item  $\mathcal{E}_{Y}$: $\frac{ 1}{n}\sum_{i=1}^n  |Y^{(i)}|^2\leq  C_{Y,n}$, 
    (this has probability $1-\delta/4$ by sub-exponential bound, since $Y$ is $\sigma^2_Y$-sub-gaussian implies $Y^2$ is $(16\sigma^2_Y,16\sigma^2_Y)$-sub-exponential.)
    \item $\mathcal{E}_{\psi}$: $\frac{ 1}{n}\sum_{i=1}^n  \| \psi(A^{(i)},Z^{(i)}) \|^2\leq  C_{\psi,n}$, 
    (this has probability $1-\delta/4$ by sub-exponential bound, since $\|\psi(A^{(i)},Z^{(i)}) \|$ is $\sigma^2_\psi$-sub-gaussian implies $\|\psi(A^{(i)},Z^{(i)}) \|^2$ is $(16\sigma^2_\psi,16\sigma^2_\psi)$-sub-exponential.)
    \item $\mathcal{E}_{H}$: $\| \Sigma_\Psi - \Hmat \|\leq  \delta_{H,n}$, where $\delta_{H,n}= \frac{MK}{n}+(1+ 8\rho_\Psi \|\mu_\Psi\|_2)\sqrt{\frac{MK}{n}}+\max\left( \sqrt{\frac{\ln(8c_5/\delta)}{c_6 n}}, \frac{\ln(8c_5/\delta)}{c_6 n} \right)+8\rho_\Psi \|\mu_\Psi\|_2 \cdot \sqrt{\frac{ \log(8/\delta)}{n}}$ (this has probability $1-\delta/4$ by Lemma \ref{lem:matrx-conc}.)
\end{itemize}

For the second term using Lemma \ref{lem:lin-reg}, we have with probability at least $1-\delta$:
\begin{equation}
 \label{bdd:beta-err-2}
  \|{\boldsymbol{\beta}} - \hat{\boldsymbol{\beta}}^*\|\leq \frac{\sigma_{3,noise}}{\sqrt{n(1-\eta_{n,\delta})\lambda_\Psi}}\sqrt{MK + 2\sqrt{MK \log(3/\delta)}+2\log(3/\delta)+\frac{2\log(3/\delta)}{1-\eta_{n,\delta}}},
\end{equation}
where $ \eta_{n,\delta}= \frac{MK}{n}+(1+ 8 \rho_\Psi \| \Sigma^{-1/2}_\Psi\mu_\Psi\|_2)\sqrt{\frac{MK}{n}}+\max\left( \sqrt{\frac{\ln(3c_5/\delta)}{c_6 n}}, \frac{\ln(3c_5/\delta)}{c_6 n} \right)+8\rho_\Psi \| \Sigma^{-1/2}_\Psi\mu_\Psi\|_2 \sqrt{\frac{ \log(3/\delta)}{n}},$ and $c_5, c_6 > 0$ are universal constants.

We bound the first term using Lemma \ref{lem:reg-error-2}. For that, we only need to find expressions for $\epsilon_{n,\delta,2}$, $\epsilon_{n,\delta,3}$ and $B_{\delta,n}$. We have already derived that $\epsilon_{n,\delta,1}$ is $\delta_{G,n}$ in this case.

\paragraph{Step 1: Expression for $\epsilon_{n,\delta,2}$}
The squared Euclidean error is:
\[ \| \hPsi_i - \PsiStar_i \|^2 = \sum_{k=1}^K \| (\hat{\tilde{w}}_{k}(A^{(i)},Z^{(i)}) - {\tilde{w}}_{k}(A^{(i)},Z^{(i)}))) \psi(A^{(i)},Z^{(i)}) \|^2. \]
Using the uniform weight error $\ew$, on $\mathcal{E}_{W}$:
\[ \| \hPsi_i - \PsiStar_i \|^2 \le \ewtilde^2K  \| \psi(A^{(i)},Z^{(i)}) \|^2. \]
and since $\hat{w}_{k}(z),{w}_{k}(z)\leq 1$,
\[ \| \hPsi_i  \|^2=\sum_{k=1}^K \| \hat{\tilde{w}}_{k}(A^{(i)},Z^{(i)}) \psi(A^{(i)},Z^{(i)}) \|^2 \le  \|  \psi(A^{(i)},Z^{(i)}) \|^2  \]
and \[ \| \Psi_i  \|^2=\sum_{k=1}^K \| \tilde{w}_{k}(A^{(i)},Z^{(i)}) \psi(A^{(i)},Z^{(i)}) \|^2 \le  \|  \psi(A^{(i)},Z^{(i)}) \|^2 . \]

We use the algebraic identity for rank-1 perturbations:
\[ \hPsi \hPsi^\top - \PsiStar \PsiStar^\top = (\hPsi - \PsiStar)(\hPsi - \PsiStar)^\top + (\hPsi - \PsiStar)\PsiStar^\top + \PsiStar(\hPsi - \PsiStar)^\top. \]
Applying the triangle inequality for the spectral norm,  on $\mathcal{E}_{W}$:
\begin{align*}
    \| \hPsi_i \hPsi_i^\top - \PsiStar_i \PsiStar_i^\top \| &\le \| \hPsi_i - \PsiStar_i \|^2 + 2 \| \hPsi_i - \PsiStar_i \| \| \PsiStar_i \| \\
    &\le\ewtilde^2K  \| \psi((A^{(i)},Z^{(i)}) \|^2 + 2 \ewtilde K  \| \psi(A^{(i)},Z^{(i)}) \|^2 \\
    &\le 3 \ewtilde K  \| \psi((A^{(i)},Z^{(i)}) \|^2 .
\end{align*}
Averaging over $n$ samples,  on $\mathcal{E}_{W}\cap \mathcal{E}_{\psi}$:
\begin{equation} \label{eq:G_bound}
    \|  \Hmat -\hH\| \le\frac{ 3 \ewtilde K}{n}\sum_{i=1}^n  \| \psi((A^{(i)},Z^{(i)}) \|^2\leq  3 \ewtilde K C_{\psi,n}.
\end{equation}

\paragraph{Step 2:  Expression for $\epsilon_{n,\delta,3}$}
On $\mathcal{E}_{\tilde W}\cap \mathcal{E}_{Y}\cap \mathcal{E}_{\psi}\cap\mathcal{E}_{H}$, we have
\begin{align*}
    \left\| \frac{1}{n}\sum_{i=1}^n (   \PsiStar_iY^{(i)}-\hPsi_iY^{(i)})\right\|  &\leq  \sum_{i=1}^n \frac{1}{n} \| \hPsi_i - \PsiStar_i \| |Y^{(i)}| \\
    &\le \| \Hmat^{-1} \| \cdot \left( \frac{\ewtilde\sqrt{K}}{2}     \left( \frac{1}{n}\sum_{i=1}^n\|\psi(A^{(i)},Z^{(i)}) \|^2 +\frac{1}{n}\sum_{i=1}^n|Y^{(i)}|^2 \right)\right)\\
    &\le \frac{\sqrt{K}}{2(\lambda_\Psi-\delta_{H,n})}(C_{\psi,n}+C_{Y,n})\ewtilde.
\end{align*}

\paragraph{Step 3: Expression for $B_{n,\delta}$}

To bound the vector term, we apply the triangle inequality:
\begin{align*}
    \frac{1}{n} \sum_{i=1}^n \| \hPsi_i \| |Y^{(i)}| 
    &\le  \frac{\sqrt{K}}{2}  \left( \frac{1}{n}\sum_{i=1}^n\|\psi(A^{(i)},Z^{(i)}) \|^2 +\frac{1}{n}\sum_{i=1}^n|Y^{(i)}|^2 \right)\\
    & \le\frac{\sqrt{K}}{2} (C_{\psi,n}+C_{Y,n}),
\end{align*}
on $\mathcal{E}_{Y}\cap \mathcal{E}_{\psi}$.

Using Lemma \ref{lem:reg-error-2} and substituting $\epsilon_{n,\delta,1}$, $\epsilon_{n,\delta,2}$, $\epsilon_{n,\delta,3}$ and $B_{\delta,n}$ by $\delta_{H,n},$ $3 \ewtilde K C_{\psi,n}$, $\frac{\sqrt K }{2(\lambda_0-\delta_{H,n})}(C_{\psi,n}+C_{Y,n})\ewtilde$ and $\frac{1}{2} (C_{\psi,n}+C_{Y,n})$, we derive that on $\mathcal{E}_{\tilde W}\cap \mathcal{E}_{Y}\cap \mathcal{E}_{\psi}\cap\mathcal{E}_{H}$
(which has probability at least $1 - \delta$):
\begin{equation}
\label{bdd:beta-err-1}
    \| \hbeta - \tbeta \| \le \frac{\sqrt{K}}{2(\lambda_\Psi-\delta_{H,n})} {\left[1+ \frac{3  K C_{\psi,n}}{\lambda_\Psi-\delta_{H,n}- 3 \ewtilde K C_{\psi,n}} \right]} (C_{\psi,n}+C_{Y,n})\ewtilde.
\end{equation}

Combining \eqref{bdd:beta-err-1} and \eqref{bdd:beta-err-2}, we get the desired bound.
\end{proof}

\begin{proof}[Proof of \Cref{thm:ate-error}]
\begin{align*}
&|\hat\tau(a)-\tau(a)| \\
 &= \left| \sum_{u=1}^K \int \left( \hat\beta_u^\top \psi(a,z) \hat{N}_u(z) - \beta_u^\top \psi(a,z) N_u(z) \right) dz \right| \\
&= \left| \sum_{u=1}^K \int \left( \hat\beta_u^\top \psi(a,z) (\hat{N}_u(z) - N_u(z)) + (\hat\beta_u - \beta_u)^\top \psi(a,z) N_u(z) \right) dz \right| \\
&\leq \sum_{u=1}^K \int \left| \hat\beta_u^\top \psi(a,z) \right| |\hat{N}_u(z) - N_u(z)| dz + \sum_{u=1}^K \int \left| (\hat\beta_u - \beta_u)^\top \psi(a,z) \right| N_u(z)dz.
\end{align*}
Applying the Cauchy-Schwarz inequality and the assumption that $\|\psi(a,z)\|_2 \leq M_a$:
\begin{align*}|\hat\tau(a)-\tau(a)| &\leq \sum_{u=1}^K \int \|\hat\beta_u\|_2 \|\psi(a,z)\|_2 |\hat{N}_u(z) - N_u(z)| dz + \sum_{u=1}^K \int \|\hat\beta_u - \beta_u\|_2 \|\psi(a,z)\|_2 N_u(z) dz \\
&\leq M_a \sum_{u=1}^K \|\hat\beta_u\|_2 \int |\hat{N}_u(z) - N_u(z)| dz + M_a \sum_{u=1}^K \|\hat\beta_u - \beta_u\|_2 \int N_u(z) dz.\end{align*}
Notice that integrating the unnormalized density over the sample space evaluates exactly to the prior probability of the cluster: $\int N_u(z) dz = \pi_u$.
\begin{align*}
M_a \sum_{u=1}^K \|\hat\beta_u - \beta_u\|_2 \int N_u(z) dz &= M_a \sum_{u=1}^K \|\hat\beta_u - \beta_u\|_2 \pi_u \\
&\leq M_a \left( \max_u \|\hat\beta_u - \beta_u\|_2 \right) \sum_{u=1}^K \pi_u \\
&\leq M_a \|\hat\beta - \beta\|_2\\
&\leq M_a \epsilon_{n,\delta,\beta},
\end{align*}
with probability $1-\delta/2$.
 Since $Z$ takes values in a compact space $\mathcal{X}^3$, let $V = \text{Vol}(\mathcal{X}^3)$ be its finite Lebesgue volume. Then:
\begin{align*}\int |\hat{N}_u(z) - N_u(z)| dz & \leq V\times \sup_{z} |\hat{N}_u(z) - N_u(z)|\leq V \epsilon_{N,n,\delta}.
\end{align*}
Furthermore, by applying the Cauchy-Schwarz inequality to the sum of the norms, we get $\sum_{u=1}^K \|\hat\beta_u\|_2 \leq \sqrt{K} \|\hat\beta\|_2 \leq \sqrt{K} (\|\beta\|_2 + \epsilon_{n,\delta,\beta})$. Substituting this back into the first term yields:
\begin{align*}M_a \sum_{u=1}^K \|\hat\beta_u\|_2 \int |\hat{N}_u(z) - N_u(z)|dz &\leq M_a V \epsilon_{N,n,\delta} \sqrt{K} (\|\beta\|_2 + \epsilon_{n,\delta,\beta}),
\end{align*}
with probability $1-\delta/2$.
Combining both terms, with probability $1 -\delta$, the absolute error in the estimated ATE satisfies:
\begin{align*}|\hat\tau(a)-\tau(a)| &\leq M_a V \sqrt{K} (\|\beta\|_2 + \epsilon_{n,\delta,\beta}) \epsilon_{N,n,\delta} + M_a \epsilon_{n,\delta,\beta} \\
&= M_a \left( \epsilon_{n,\delta,\beta} + V \sqrt{K} (\|\beta\|_2 + \epsilon_{n,\delta,\beta}) \epsilon_{N,n,\delta} \right).
\end{align*}
\end{proof}

\end{document}